\documentclass{article}

\usepackage[utf8]{inputenc}
\usepackage[english]{babel}
\usepackage{authblk}
\usepackage{amssymb,amsmath,amsthm,twoopt,xargs,mathtools}
\usepackage{times,dsfont,ifthen}
\usepackage{fancyhdr,xcolor}
\usepackage{algorithm,algorithmic,breakcites,natbib}

\usepackage{url}            
\usepackage{booktabs}       
\usepackage{amsfonts}       
\usepackage{nicefrac}       
\usepackage{comment}
\usepackage{graphics}
\usepackage{graphicx}
\usepackage{subfigure}
\usepackage{enumitem}

\usepackage{easyeqn}
\usepackage{xargs}
\usepackage{upgreek}
\usepackage{paracol}
\usepackage{stmaryrd}

\usepackage{aliascnt}
\usepackage{cleveref}

\makeatletter
\newtheorem{theorem}{Theorem}
\crefname{theorem}{theorem}{Theorems}
\Crefname{Theorem}{Theorem}{Theorems}

\newaliascnt{proposition}{theorem}

\aliascntresetthe{proposition}
\crefname{proposition}{proposition}{propositions}
\Crefname{Proposition}{Proposition}{Propositions}

\newaliascnt{lemma}{theorem}
\newtheorem{lemma}[lemma]{Lemma}
\aliascntresetthe{lemma}
\crefname{lemma}{lemma}{lemmas}
\Crefname{Lemma}{Lemma}{Lemmas}

\newaliascnt{corollary}{theorem}

\aliascntresetthe{corollary}
\crefname{corollary}{corollary}{corollaries}
\Crefname{Corollary}{Corollary}{Corollaries}

\crefname{definition}{definition}{definitions}
\Crefname{Definition}{Definition}{Definitions}

\newtheorem{remark}{Remark}
\crefname{remark}{remark}{remarks}
\Crefname{Remark}{Remark}{Remarks}

\newtheorem{assumption}{\textbf{H}\hspace{-3pt}}
\Crefname{assumption}{\textbf{H}\hspace{-3pt}}{\textbf{H}\hspace{-3pt}}
\crefname{assumption}{\textbf{H}}{\textbf{H}}

 \def\elboneq{\mathcal{L}_{\IFIS}}

 \def\mass{\mathrm{M}}
 \def\Normal{\mathrm{N}}

\def\simiid{\overset{\operatorname{iid}}{\sim}}
\def\IFIS{\ensuremath{\operatorname{NEO}}}
\def\InFiNE{{\small \IFIS}}
\def\NEO{{\small \IFIS}}

\def\transfo{\operatorname{T}}

\def\rmd{\operatorname{d}\hspace{-2pt}}
\def\Id{\operatorname{Id}}
\def\PP{\mathbb{P}}

\def\PE{\mathbb{E}}

\def\mcf{\mathcal{F}}
\def\iid{i.i.d.}

\def\ELBO{\operatorname{ELBO}}
\def\rset{\mathbb{R}}

\def\nset{\mathbb{N}}
\def\nsets{\mathbb{N}^*}

\def\dummy{f}

\newcommandx{\marginal}[2][1=]{\xi^{#2}_{#1}}

\newcommandx{\margindensm}[2]{m_{#1}^{#2}}
\newcommandx{\margindensmu}[4]{m_{#1}^{#2}(#3|#4)}

\newcommandx{\margindens}[4][4=]{\ifthenelse{\equal{#3}{}}{q_{#1}^{#4}(#2)}{q_{#1,#3}^{#4}(#2)}}
\newcommandx{\margindensw}[3][3=]{\ifthenelse{\equal{#3}{}}{q_{#1}^{#2}}{q_{#1,#3}^{#2}}}

\newcommand{\chunk}[3]{#1_{#2:#3}}

\def\rmd{\mathrm{d}}

\def\eqsp{\,}

\def\fwdtransfo{T}
\newcommandx{\fwdtransfoparam}[2]{\fwdtransfo_{#1,#2}}

\def\msa{\mathsf{A}}
\def\borel{\mathcal{B}}

\def\eg{\text{e.g.}}

\def\wrt{w.r.t.}

\newcommand{\abs}[1]{\left\vert #1 \right\vert}
\newcommand{\absLigne}[1]{\vert #1 \vert}
\newcommand{\tvnorm}[1]{\| #1 \|_{\mathrm{TV}}}

\newcommandx{\Vnorm}[2][1=V]{\| #2 \|_{#1}}
\newcommandx{\VnormEq}[2][1=V]{\ensuremath{\left\| #2 \right\|_{#1}}}

\newcommandx\probaMarkovTilde[2][2=]
{\ifthenelse{\equal{#2}{}}{{\widetilde{\mathbb{P}}_{#1}}}{\widetilde{\mathbb{P}}_{#1}\left[ #2\right]}}

\def\ie{\textit{i.e.}}

\def\eqsp{\;}
\newcommand{\coint}[1]{\left[#1\right)}
\newcommand{\ocint}[1]{\left(#1\right]}
\newcommand{\ooint}[1]{\left(#1\right)}
\newcommand{\ccint}[1]{\left[#1\right]}

\newcommand{\1}{\mathds{1}}
\newcommand{\indi}[1]{\1_{#1}}
\newcommand{\indiacc}[1]{\mathds{1}_{\{ #1   \}}}
\newcommandx{\weight}[2][2=n]{\omega_{#1,#2}^N}

\newcommandx\sequence[3][2=,3=]
{\ifthenelse{\equal{#3}{}}{\ensuremath{\{ #1_{#2}\}}}{\ensuremath{\{ #1_{#2}, \eqsp #2 \in #3 \}}}}
\newcommandx\sequenceD[3][2=,3=]
{\ifthenelse{\equal{#3}{}}{\ensuremath{\{ #1_{#2}\}}}{\ensuremath{( #1)_{ #2 \in #3} }}}

\newcommandx{\sequencen}[2][2=n\in\N]{\ensuremath{\{ #1_n, \eqsp #2 \}}}
\newcommandx\sequenceDouble[4][3=,4=]
{\ifthenelse{\equal{#3}{}}{\ensuremath{\{ (#1_{#3},#2_{#3}) \}}}{\ensuremath{\{  (#1_{#3},#2_{#3}), \eqsp #3 \in #4 \}}}}
\newcommandx{\sequencenDouble}[3][3=n\in\N]{\ensuremath{\{ (#1_{n},#2_{n}), \eqsp #3 \}}}

\def\iid{i.i.d.}

\def\eg{e.g.}

\newcommand{\opnorm}[1]{{\left\vert\kern-0.25ex\left\vert\kern-0.25ex\left\vert #1
    \right\vert\kern-0.25ex\right\vert\kern-0.25ex\right\vert}}

\def\Id{\operatorname{Id}}

\newcommandx{\CPE}[3][1=]{{\mathbb E}_{#1}\left[\left. #2 \middle \vert #3 \right. \right]} 
\newcommandx{\CPVar}[3][1=]{\mathrm{Var}^{#3}_{#1}\left\{ #2 \right\}}
\newcommand{\CPP}[3][]
{\ifthenelse{\equal{#1}{}}{{\mathbb P}\left(\left. #2 \, \right| #3 \right)}{{\mathbb P}_{#1}\left(\left. #2 \, \right | #3 \right)}}

\newcommandx{\osc}[2][1=]{\mathrm{osc}_{#1}(#2)}

\def\Id{\operatorname{Id}}

\def\target{\pi}
\def\proposal{\rho}
\def\weightfunc{\tilde{w}}
\newcommand{\chunku}[3]{#1^{#2:#3}}

\newcommand{\chunkum}[4]{#1^{#2:#3 \setminus \{#4\}}}

\def\Jac{\mathbf{J}}
\newcommand{\JacOp}[1]{\Jac_{#1}}

\newcommand{\intentier}[2]{[#1:#2]}
\newcommand{\intentierU}[1]{[#1]}

\def\rmi{\mathrm{I}}

\def\const{Z}

\newcommand{\estConstC}[1]{\widehat{Z}_{#1}}

\def\measpi{\boldsymbol{\pi}}

\newcommand\coupling[2]{\Gamma(\mu,\nu)}

\def\tpi{\tilde{\pi}}

\def\msa{\mathsf{A}}

\def\msb{\mathsf{B}}

\def\mso{\mathsf{O}}

\def\mcbb{\mathcal{B}}  

\def\mcf{\mathcal{F}}

\def\rset{\mathbb{R}}
\def\rsets{\mathbb{R}^*}

\def\zset{\mathbb{Z}}
\def\nset{\mathbb{N}}
\def\nsets{\mathbb{N}^*}

\def\rmd{\mathrm{d}}

\def\rme{\mathrm{e}}

\def\rmc{\mathrm{C}}

\def\rmC{\mathrm{C}}

\def\Idd{\mathrm{Id}}
\def\likelihood{\mathrm{L}}

\newcommandx{\estpiNaive}[3][1=g,2=N,3=f]{\hat{\pi}_{#1}^{#2}(#3)}

\newcommandx{\norm}[2][1=]{\ifthenelse{\equal{#1}{}}{\left\Vert #2 \right\Vert}{\left\Vert #2 \right\Vert^{#1}}}
\newcommandx{\normLigne}[2][1=]{\ifthenelse{\equal{#1}{}}{\Vert #2 \Vert}{\Vert #2 \Vert^{#1}}}

\def\rhoT{\rho_{\transfo}}
\def\constT{\const_{\transfo}}

\def\wrt{with respect to}

\def\infineIS{I^{\NEO}_{\varpi,N}}
\def\bound{M_{\transfo}^{\varpi}}
\def\infineISh{I^{\NEO}_{\varpi,N,h}}
\def\infineSNIS{J^{\NEO}_{\varpi,N}}
\newcommandx{\Exp}[2][1=]{\ensuremath{\PE_{#1}\left[ #2\right]}}
\def\NEIS{\ensuremath{\operatorname{NEIS}}}
\def\wcont{w_t^\mathrm{c}}
\def\rhoTcont{\rhoT^{\operatorname{c}}}
\def\obs{z}
\def\argmax{ \arg \max}

\def\Var{\operatorname{Var}}
\def\tU{\tilde{U}}
\def\txts{\textstyle}
\def\varpic{\varpi^{\mathrm{c}}}
\newcommand{\partint}[1]{ \left \lfloor #1 \right \rfloor}
\def\divergence{\operatorname{div}}
\def\support{\mathrm{support}}
\def\Tvarpi{T_\varpi}
\def\trace{\operatorname{tr}}

\def\adj{\operatorname{adj}}

\title{\IFIS: Non Equilibrium Sampling on the Orbit of a Deterministic Transform}
\date{}

\author[$\dag$]{Achille Thin}
\author[$\ddag$]{Yazid Janati}
\author[$\ddag$]{Sylvain Le Corff}
\author[$\dag$]{Charles Ollion}
\author[$\top$]{Arnaud Doucet}
\author[$\star$]{Alain Durmus}
\author[$\dag$]{\'Eric Moulines}
\author[$\wr$]{Christian Robert}
\affil[$\dag$]{{\small CMAP, \'Ecole Polytechnique, Institut Polytechnique de Paris, Palaiseau.}}
\affil[$\ddag$]{{\small Samovar, T\'el\'ecom SudParis, d\'epartement CITI, TIPIC, Institut Polytechnique de Paris, Palaiseau.}}
\affil[$\top$]{{\small Department of Statistics, University of Oxford.}}
\affil[$\star$]{{\small CMLA, \'Ecole Normale Sup\'erieure Paris-Saclay.}}
\affil[$\wr$]{{\small Ceremade, Université Paris-Dauphine \& Department of Statistics, University of Warwick.}}

\usepackage{geometry}
\begin{document}

\maketitle

\begin{abstract}
Sampling from a complex distribution $\pi$ and approximating its intractable normalizing constant $\mathrm{Z}$ are challenging problems. 
In this paper, a novel family of importance samplers (IS) and Markov chain Monte Carlo (MCMC) samplers is derived. 
Given an invertible map $\mathrm{T}$, these schemes combine (with weights) elements from the forward and backward Orbits   through points sampled from a proposal distribution $\rho$. The map $\mathrm{T}$ does not leave the target $\pi$ invariant, hence the name NEO, standing for Non-Equilibrium Orbits. 
NEO-IS provides unbiased estimators of the normalizing constant and self-normalized IS estimators of expectations under $\pi$ while NEO-MCMC combines multiple NEO-IS estimates of the normalizing constant and an iterated sampling-importance resampling mechanism to sample from $\pi$. 
For $\mathrm{T}$ chosen as a discrete-time integrator of a conformal Hamiltonian system, NEO-IS achieves state-of-the art performance on difficult benchmarks and NEO-MCMC is able to explore highly multimodal targets. Additionally, we provide detailed theoretical results for both methods. In particular, we show that NEO-MCMC is uniformly geometrically ergodic and establish explicit mixing time estimates under mild conditions.
\end{abstract}

\section{Introduction}
Consider a target distribution of the form $\target(x)\propto\proposal(x)\likelihood(x)$ where $\proposal$ is a probability density function (pdf)  on $\rset^d$ and $\likelihood$ is a nonnegative function. Typically, in a Bayesian setting, $\pi$ is a posterior distribution associated with a prior distribution $\rho$ and a likelihood function $\likelihood$. 
An other situation of interest is generative modeling where $\pi$ is the distribution implicitly defined by a Generative Adversarial Networks (GAN) discriminator-generator pair where $\rho$ is the distribution of the generator and $\likelihood$ is derived from the discriminator \citep{turner:hung:2019, che:bengio:2020}.
An other situation of interest is generative modeling where $\pi$ is the distribution implicitly defined by a Variational Auto Encoder (VAE) encoder-decoder pair where $\rho$ is the distribution output by the encoder and $\likelihood$ is an importance weight between the distribution of the decoder and of the encoder \citep{kingma:welling:2013, burda:grosse:2015}.
We are interested in this paper in sampling from $\pi$ and  approximating  its intractable normalizing constant $\const=\int \proposal(x) \likelihood(x) \rmd x$. These problems arise in many applications in statistics, molecular dynamics or machine learning, and remain challenging. 

Many approaches to compute normalizing constants are based on Importance Sampling (IS) - see \cite{agapiou2017importance,akyildiz2021convergence} and the references therein - and its variations, among others, Annealed Importance Sampling (AIS) \citep{neal-annealed:2001,wu:burda:grosse:2016,ding2019learning} and Sequential Monte Carlo (SMC) \citep{del2006sequential}. More recently, Neural IS has also become very popular in machine learning; see e.g. \cite{el2012bayesian,muller2018neural,papamakarios2019normalizing,prangle2019distilling,wirnsberger2020targeted}. Neural IS is an adaptive IS which relies on an importance function obtained by applying a normalizing flow to a reference distribution. The parameters of this normalizing flow are chosen by minimizing a divergence between the proposal and the target (such as the Kullback-Leibler \cite{muller2018neural} or the $\chi^2$-divergence \cite{agapiou2017importance}).

More recently, the \emph{Non-Equilibrium IS} (NEIS) method has been introduced by \cite{rotskoff:vanden-eijden:2019} as an alternative to these approaches. Similar to Neural IS, NEIS consists in transporting samples $\{X^i\}_{i=1}^N$ from a reference distribution using a family of deterministic mappings. This family for NEIS is chosen to be an homogeneous differential flow $(\phi_t)_{t \in \rset}$. In contrast to Neural IS, for any $i \in [N]$, the sample $X^i$ is propagated both forward and backward in time  along the orbits associated with $(\phi_t)_{t \in \rset}$ until stopping conditions are met. Moreover, the resulting estimator of the normalizing constant is obtained by computing weighted averages of the whole orbit $(\phi_t(X^i))_{t \in [\tau_{+,i},\tau_{-,i}]}$, where $\tau_{+,i},\tau_{-,i}$ are the resulting stopping times, and not only the endpoints $\phi_{\tau_{+,i}}(X^i),\phi_{\tau_{-,i}}(X^i)$. In \cite{rotskoff:vanden-eijden:2019}, the authors  provide an application of NEIS with $(\phi_t)_{t \in \rset}$ associated with a conformal Hamiltonian dynamics, and reports impressive numerical results on difficult  normalizing constants estimation problems, in particular for high-dimensional multimodal distributions.

We propose in this work \IFIS-IS which alleviates the shortcomings of NEIS. Similar to NEIS, samples are drawn from a reference distribution, typically set to $\proposal$, and are propagated under the forward and backward orbits of a \emph{discrete-time} dynamical system associated with an invertible transform $\transfo$. An estimator of the normalizing constant is obtained by reweighting all the points on the whole orbits using the IS rule. 
Contrary to  NEIS, the \IFIS-IS estimator of $\mathrm{Z}$ is unbiased under assumptions that are mild and easy to verify. It is more flexible than NEIS because it does not rely on the accuracy of the  discretization of a continuous-time dynamical system. 

We then show how it is possible to leverage the unbiased estimator of $\mathrm{Z}$ defined by \IFIS-IS to obtain \IFIS-MCMC, a novel massively parallel MCMC algorithm to sample from $\pi$. In a nutshell, \IFIS-MCMC relies on parallel walkers which each estimates the normalizing constant but are allowed to interact through a resampling mechanism.  Our contributions can be summarized as follows.
\begin{enumerate}[label=\textbf{(\roman*)}]
    \item We present a novel class of IS estimators of the normalizing constant $\const$ referred to as  NEO-IS. More broadly, a small modification of this algorithm also allows us to estimate integrals with respect to $\target$. 
    Both finite sample and asymptotic guarantees are provided for these two  methodologies.
    \item We develop a new massively parallel MCMC method, \IFIS-MCMC. \IFIS-MCMC combines \IFIS-IS unbiased estimator of the normalizing constant with iterated sampling-importance resampling methods. We prove that it is $\target$-reversible and ergodic under very general conditions. We derive also conditions which imply that \IFIS-MCMC is  uniformly geometrically ergodic (with an explicit expression of the mixing time).
    \item We illustrate our findings using numerical benchmarks which show that both \IFIS-IS and \IFIS-MCMC outperform  state-of-the-art (SOTA) methods in difficult settings. 
\end{enumerate}

\section{NEO-IS algorithm}
In this section, we derive the  NEO-IS algorithm.
The two key ingredients for this algorithm are  (1) the reference distribution $\proposal$ and (2) a transformation $\transfo$ assumed to be a $\rmC^1$-diffeomorphism with inverse $\transfo^{-1}$. Write, for $k \in \nsets = \nset \setminus\{0\}$, $\transfo^{k}=\transfo\circ\transfo^{k-1}$, $\transfo^{0}=\Idd_{d}$ and similarly $\transfo^{-k}=\transfo^{-1}\circ\transfo^{-(k-1)}$.
For any $k \in \zset$, denote by $\rho_k : \rset^d \to \rset_+$ the pushforward of $\rho$ by $\transfo^k$,  defined for $x\in\rset^d$ by $\rho_k(x)= \rho(\transfo^{-k}(x)) \JacOp{\transfo^{-k}}(x)$,
where ${\JacOp{\Phi}(x)}\in\rset^+$ is the absolute value of the Jacobian determinant of $\Phi: \rset^d\to \rset^d$ evaluated at $x$. In line with multiple importance sampling \emph{\`a la} Owen and Zhou \cite{owen:zhou:2000}, we introduce the proposal density
\begin{equation}\label{eq:rhoT}
    \rhoT(x) = \Omega^{-1}\sum\nolimits_{k \in\zset} \varpi_k \rho_k(x)\eqsp,
  \end{equation}
where $\{\varpi_k\}_{k \in\zset}$ is a  nonnegative sequence and $\Omega= \sum_{k \in \zset} \varpi_k$. Note that we assume in the sequel that the support of the weight sequence defined as $\{ k \in \zset\, :\, \varpi_k \neq 0\}$ is finite. Thus, the mixture distribution in \eqref{eq:rhoT} is a \textbf{finite mixture}. Given $x \in \rset^d$, $\rhoT(x)$ is a function of the forward and backward orbit of $\transfo$ through $x$.  For any nonnegative function $\dummy$, the definition of $\rhoT$ implies that 
$$
\int\dummy(y)\rhoT(y)\rmd y = \Omega^{-1} \int \sum_{k\in\zset}\varpi_k f(\transfo^k(x)) \rho(x)\rmd x\eqsp.
$$
Assuming that $\varpi_0>0$,  the ratio  $\rho(x)/\rhoT(x) \leq \varpi_0^{-1}\Omega<\infty$ is bounded. We can therefore apply the IS principle  which allows to write the identity
\begin{equation}
\label{eq:key-relation}
\int \dummy(x) \rho(x)  \rmd x =\int \left(\dummy(y) \frac{\rho(y)}{\rhoT(y)}\right) \rhoT(y)  \rmd y=\int \sum_{k\in\zset}  \dummy(\transfo^{k}(x)) w_k(x) \rho(x)  \rmd x \eqsp,
\end{equation}
where  the weights are given by (see  \Cref{app:proof:def_w_k} for a detailed derivation),
\begin{equation}
    \label{eq:def_w_k}
  { w_k(x)=\varpi_k \rho(\transfo^{k}(x)) / \{\Omega\rhoT(\transfo^k(x))\} = \left.  \varpi_k \rho_{-k}(x) \middle / \left. \sum\nolimits_{i\in\zset}\varpi_{k+i} \rho_{i}(x) \right. \right. \eqsp.}
\end{equation} 
We assume in the sequel that $\varpi_0>0$. In particular, note that under this condition, the weights $w_{k}$ are also upper bounded uniformly in $x$: for any $x \in \rset^d$,  $w_{k}(x) \leq \varpi_{k}/\varpi_{0}$. Eqs. \eqref{eq:key-relation} and \eqref{eq:def_w_k} suggest to estimate the integral $\int f(x) \rho(x) \rmd x$ by $\infineIS(f)= N^{-1} \sum_{i=1}^N \sum_{k\in\zset} w_k(X^i)  f(\transfo^k(X^i))$ where $\{X^i\}_{i=1}^N$ are \iid\ samples from the proposal $\rho$, which is denoted by $\chunku{X}{1}{N} \simiid \proposal$. 
\begin{figure}
\begin{algorithm}[H]
\begin{enumerate}[leftmargin=0cm,itemindent=.5cm,labelwidth=\itemindent,labelsep=0cm,align=left]
\item Sample $\chunku{X}{1}{N} \simiid \ \rho$ for $i\in[N]$.
\item For $i \in \intentierU{N}$, compute the
  path $(\transfo^j(X^i))_{j =0}^K$ and weights $(w_j(X^i))_{j =0}^K$.
\item$\infineIS(f) = N^{-1} \sum_{i=1}^N \sum_{k\in\zset} w_k(X^i)  f(\transfo^k(X^i))$.
\end{enumerate}
\caption{\InFiNE-IS Sampler}
\label{algo:infine_partial}
\end{algorithm}
\end{figure}

This estimator is obtained by a weighted combination of  the elements of the independent forward and backward orbits $\{\transfo^k(X^i)\}_{k \in\zset}$ with $X^{1:N} \simiid \proposal$. This estimator is referred to as \NEO-IS. Choosing $f\equiv \likelihood$ provides the \InFiNE-IS estimator of the normalizing constant of $\target$:
\vspace{-0.05cm}
\begin{equation}  
\label{eq:def_estimator_normal_const} 
\textstyle \estConstC{X^i}=\sum\nolimits_{k\in\zset}\likelihood(\transfo^{k}(X^i))w_k(X^i)\eqsp, \quad
 \estConstC{X^{1:N}}=N^{-1}\sum_{i=1}^{N} \estConstC{X^i} \eqsp.
\end{equation}
We  now study the performance of the \NEO-IS estimator. The following two quantities play a fundamental role in the analysis:
\begin{equation}
\label{eq:bound_z_infine}
\textstyle E^{\varpi} _{\transfo} = \PE_{X \sim \rho}\big[\big(\sum_{k\in\zset}w_k(X) \likelihood(\transfo^k(X))/\const\big)^2\big], 
    \bound = \sup_{x\in\rset^d}\sum_{k\in\zset}w_k(x)\likelihood(\transfo^k(x))/\const\eqsp.
\end{equation}
\begin{theorem}
\label{theo:clt_const_infine}
$\estConstC{\chunku{X}{1}{N}}$ is an unbiased estimator of $\const$. If $E^{\varpi} _{\transfo}<\infty$, then,
    $\PE[|\estConstC{\chunku{X}{1}{N}}/\const-1|^2] = N^{-1} (E^{\varpi} _{\transfo}-1)$. If $\bound < \infty$, then, for any $\delta \in (0,1)$, with probability $1-\delta$, $\sqrt{N} \left|\estConstC{\chunku{X}{1}{N}}/\const - 1\right| \leq \bound \sqrt{\log(2/\delta)/2}$.
\end{theorem}
The (elementary) proof is postponed to \Cref{subsec:proof:clt_const_infine}. $ E^{\varpi} _{\transfo}$ plays the role of the second-order moment of the importance weights $\PE_{X \sim \proposal} [\likelihood^2(X)]$ which is key to the performance of IS algorithms \cite{agapiou2017importance,akyildiz2021convergence}. In addition, since the \NEO-IS estimator  $\estConstC{\chunku{X}{1}{N}}$ is unbiased, the Cauchy-Schwarz inequality shows that $\PE_{X \sim \rho}\big[\big(\sum_{k\in\zset}w_k(X) \likelihood(\transfo^k(X)) )^2\big] \geq \const^2$ and hence that $ E^{\varpi} _{\transfo} \geq 1$.
Note that if $\| \likelihood \|_\infty= \sup_{x \in \rset^d} \likelihood(x) < \infty$, then since the weights are uniformly bounded by $\Omega\varpi_0^{-1}$, we have $M^{\varpi} _{\transfo}\leq \|\likelihood\|_\infty\Omega\varpi_0^{-1}/\const$.

Using the NEO-IS estimate $\estConstC{\chunku{X}{1}{N}}$ of the normalizing constant, we can construct a self-normalized IS  estimate of $\int f(x) \target(x) \rmd x$: 
\begin{equation}
  \label{eq:def_estimator_naive_monte_carlo}
 \infineSNIS(\dummy)= N^{-1}\sum_{i=1}^{N}
  \frac{\estConstC{X^i}}{\estConstC{\chunku{X}{1}{N}}}\sum_{k\in \zset}\frac{\likelihood(\transfo^{k}(X^i)) w_k(X^i)}{\estConstC{X^i}} \dummy(\transfo^k(X^i))\eqsp,
  \end{equation}
referred to as \IFIS-SNIS estimator. This expression may seem unnecessarily complicated but highlights the hierarchical structure of the estimator. We combine estimators $(\estConstC{X^i})^{-1}\sum_{k\in \zset}\likelihood(\transfo^{k}(X^i)) w_k(X^i) \dummy(\transfo^k(X^i))$ evaluated on the forward and backward orbits through the points $\{ X^i \}_{i=1}^N$ using weights $\{\estConstC{X^i}/\estConstC{\chunku{X}{1}{N}}\}_{i=1}^N$. 
Although the \IFIS-IS estimator is unbiased, the \IFIS-SNIS is in general biased.  However, for bounded functions, both the bias and the variance of the \IFIS-SNIS estimator are $O(N^{-1})$, with constants proportional to $E_{\transfo}^\varpi$.  For $g$ a $\target$-integrable function, we set $\target(g)= \int g(x) \target(x) \rmd x$.
\begin{theorem}
\label{theo:bias_mse_snis}
Assume that $E^{\varpi}_{\transfo} < \infty$. Then, 
for any function $g$ satisfying $\sup_{x \in \rset^d} \abs{g(x)} \leq 1$ on $\rset^d$, and $N \in \nset$,
\begin{align}
     &\Exp[\chunku{X}{1}{N} \simiid \rho]{|\infineSNIS(g) - \pi(g)|^2} \leq 4 \cdot N^{-1} E^{\varpi}_{\transfo} \eqsp,\\
     &\left|\Exp[\chunku{X}{1}{N} \simiid \rho]{\infineSNIS(g) - \pi(g)}\right|\leq  2 \cdot N^{-1} E^{\varpi}_{\transfo} \eqsp.
\end{align}
If $\bound <\infty$, then for $\delta \in (0,1]$, with probability at least $1-\delta$, 
\begin{equation}
    \label{eq:exp_concentration}
   \sqrt{N}| \infineSNIS (g) - \pi(g) | \leq {\|g\| _{\infty} \bound \sqrt{32\log(4/\delta)} }\eqsp.
\end{equation}
\end{theorem}
The proof is postponed to \Cref{subsec:proo:bias_mse_snis}. These results extend to \IFIS-SNIS estimators the results known for self-normalized estimators; see e.g., \cite{agapiou2017importance,akyildiz2021convergence} and the references therein. The upper bounds stated in this result suggest it is good practice to keep $E_{\transfo}^\varpi/N$ small in order to obtain sensible approximations. For two pdfs $p$ and $q$ on $\rset^d$, denote by $\operatorname{D}_{\chi^2}(p,q)= \int \{p(x)/q(x) -1\}^2 q(x) \rmd x$  the $\chi^2$-divergence between $p$ and $q$. 
\begin{lemma}
\label{lem:chi2_E_T}
For any nonnegative sequence $(\varpi_k)_{k\in\zset}$, we have
$ E^{\varpi}_{\transfo} \leq  \operatorname{D}_{\chi^2}( \pi \Vert \rhoT)+1$.
\end{lemma}
The proof is postponed to \Cref{sec:proof_lem:chi2_E_T}. \Cref{lem:chi2_E_T} suggests that  accurate sampling requires $N$ to scale linearly with the $\chi^2$-divergence between the target $\pi$ and the extended proposal $\rhoT$.
\begin{remark}\em
We can extend \IFIS\ to non homogeneous flows, replacing the family $\{\transfo^k\colon k\in\zset\}$ with a collection of mappings $\{\mathsf{T}_k\colon k\in\zset\}$.
This would allow us to consider further flexible classes of transformations such as normalizing flows; see \eg\ \cite{papamakarios2019normalizing}. The $\chi^2$-divergence $\operatorname{D}_{\chi^2}( \pi \Vert \rhoT)$ provides natural criteria for learning the transformation.   We leave this extension to future work.
\end{remark}

\paragraph{Conformal Hamiltonian transform}
\label{subsec:NISestimators}
The efficiency of $\InFiNE$ relies heavily on the choice of $\transfo$. Intuitively, a sensible choice of
$\transfo$ requires that  (i) $E_{\transfo}^\varpi$ is small, i.e. $\rhoT$ should be close to $\pi$ by \Cref{lem:chi2_E_T} (see \eqref{eq:bound_z_infine}), (ii) the inverse $\transfo^{-1}$ and the Jacobian of $\transfo$ are easy to compute. Following
\cite{rotskoff:vanden-eijden:2019}, we use for $\transfo$ a discretization of a conformal Hamiltonian dynamics.  Assume that $U(\cdot) = -\log \target(\cdot)$ is continuously differentiable. We consider the augmented distribution $\tilde\target(q,p) \propto \exp \{-U(q)-K(p)\}$
on  $\rset^{2d}$, where $q$ is the position, $p$ is the momentum, and $K(p)=p^T\mass^{-1} p/2$ is the kinetic energy, with $\mass$ a positive definite mass matrix. By construction, the marginal distribution of the momentum under $\tilde{\target}$ is the target pdf $\target(q)= \int \tilde{\target}(q,p) \rmd p$.  The conformal Hamiltonian ODE associated with $\tilde \target$ is defined by 
\begin{align}
  \label{eq:ODE_hamiltonian}
  &\rmd{q_t}/\rmd t =\nabla_{p} H(q_t,p_t) = \mass^{-1} p_t \eqsp, \\
  \nonumber
&\rmd {p}_t/\rmd t =-\nabla_{q} H(q_t,p_t)-\gamma p_t = -\nabla U(q_t) - \gamma p_t \eqsp,
\end{align}
where $H(q,p)= U(q)+ K(p)$, and $\gamma >0$ is a damping constant. Any solution $(q_t,p_t)_{t \geq 0}$ of \eqref{eq:ODE_hamiltonian} satisfies setting $H_t = H(q_t,p_t)$, $\rmd H_t/\rmd t   = - \gamma p_t^T\mass^{-1} p_t\leq 0$. Hence, all orbits converge to fixed points that satisfy $\nabla U(q)=0$ and $p=0$; see e.g. \cite{francca2019conformal,maddison2018hamiltonian}.

 \begin{figure*}[h!]
     \centering
    \includegraphics[width=0.\linewidth]{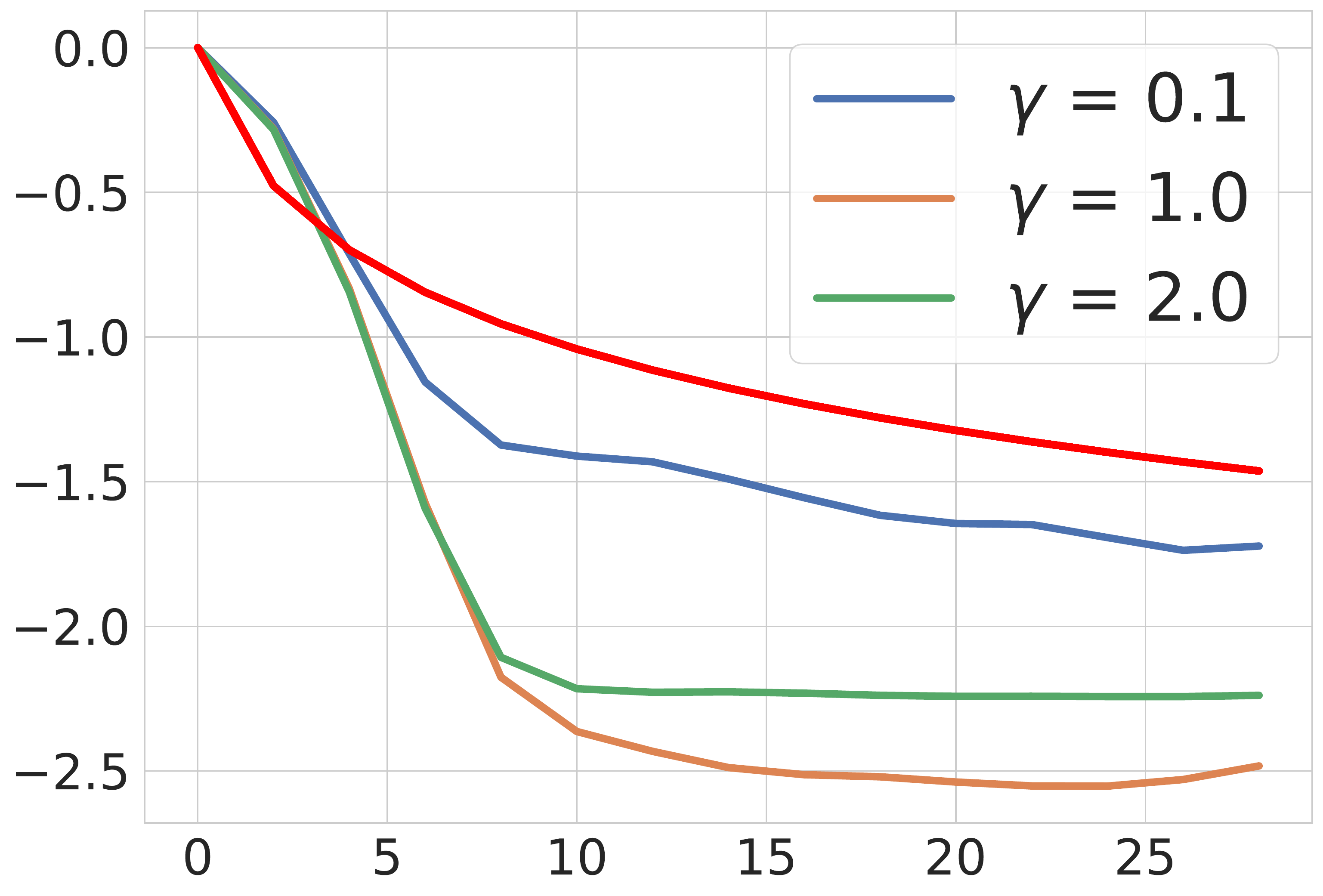}
   \includegraphics[width=0.24\linewidth]{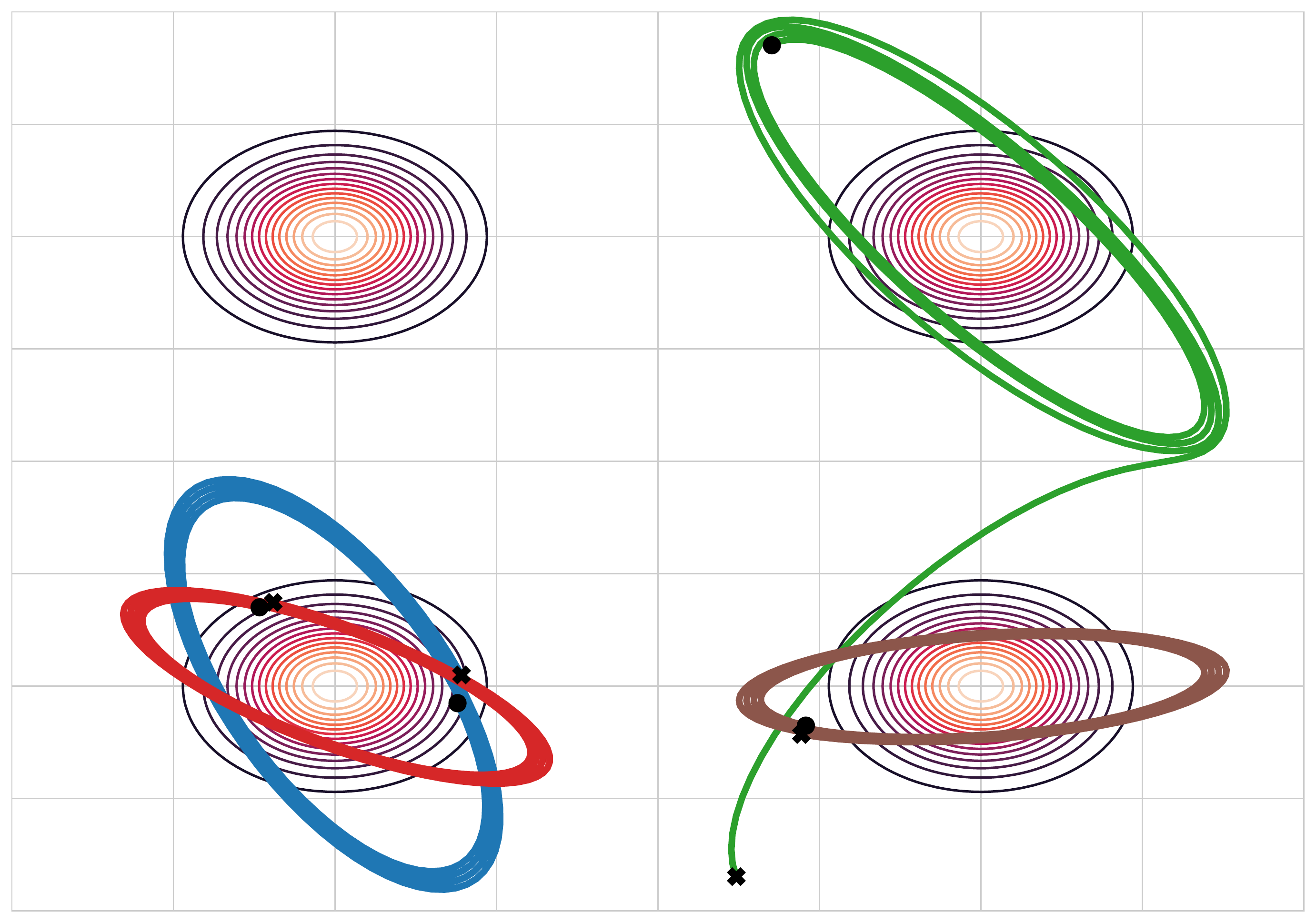} 
   \includegraphics[width=0.24\linewidth]{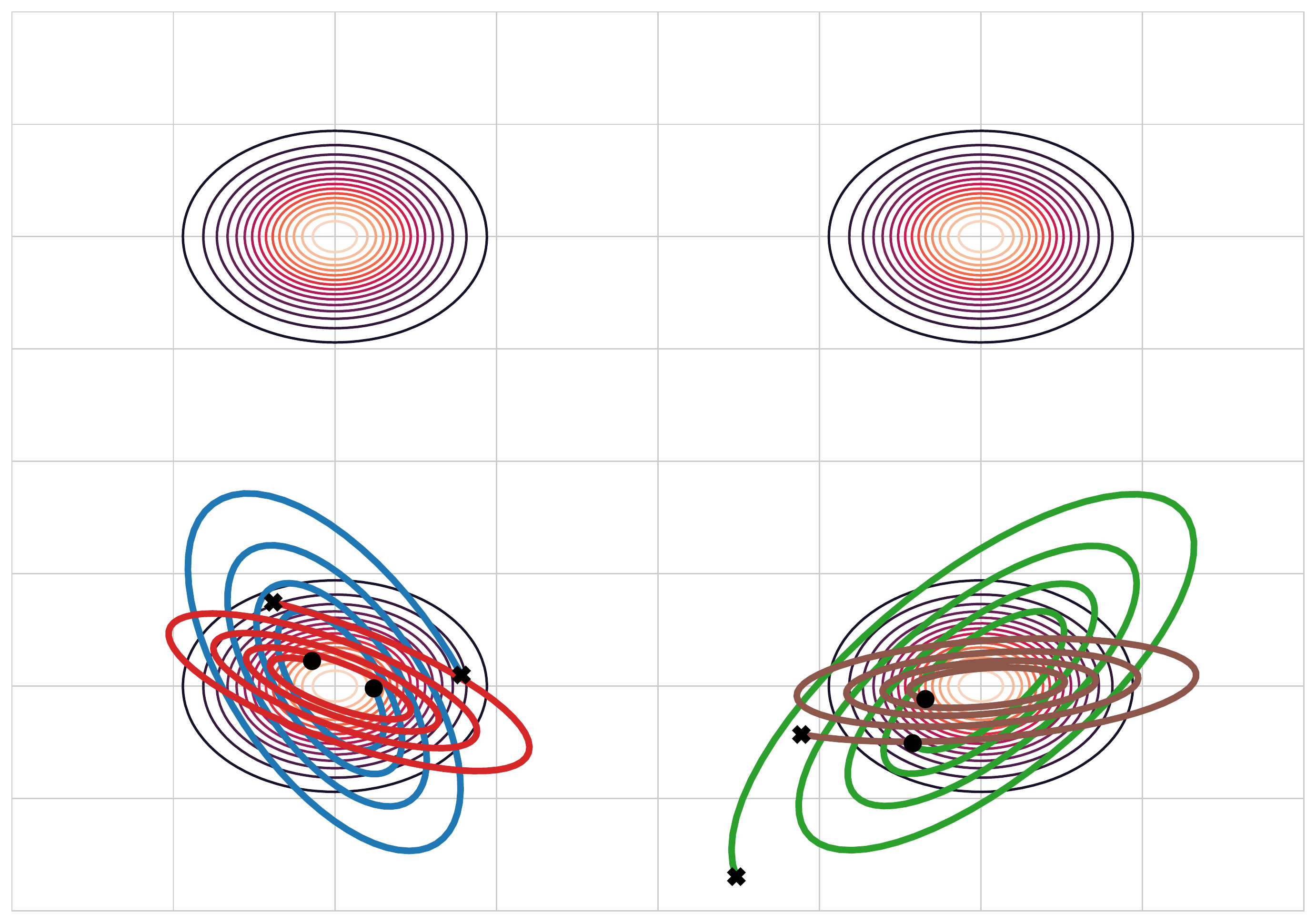}
   \includegraphics[width=0.24\linewidth]{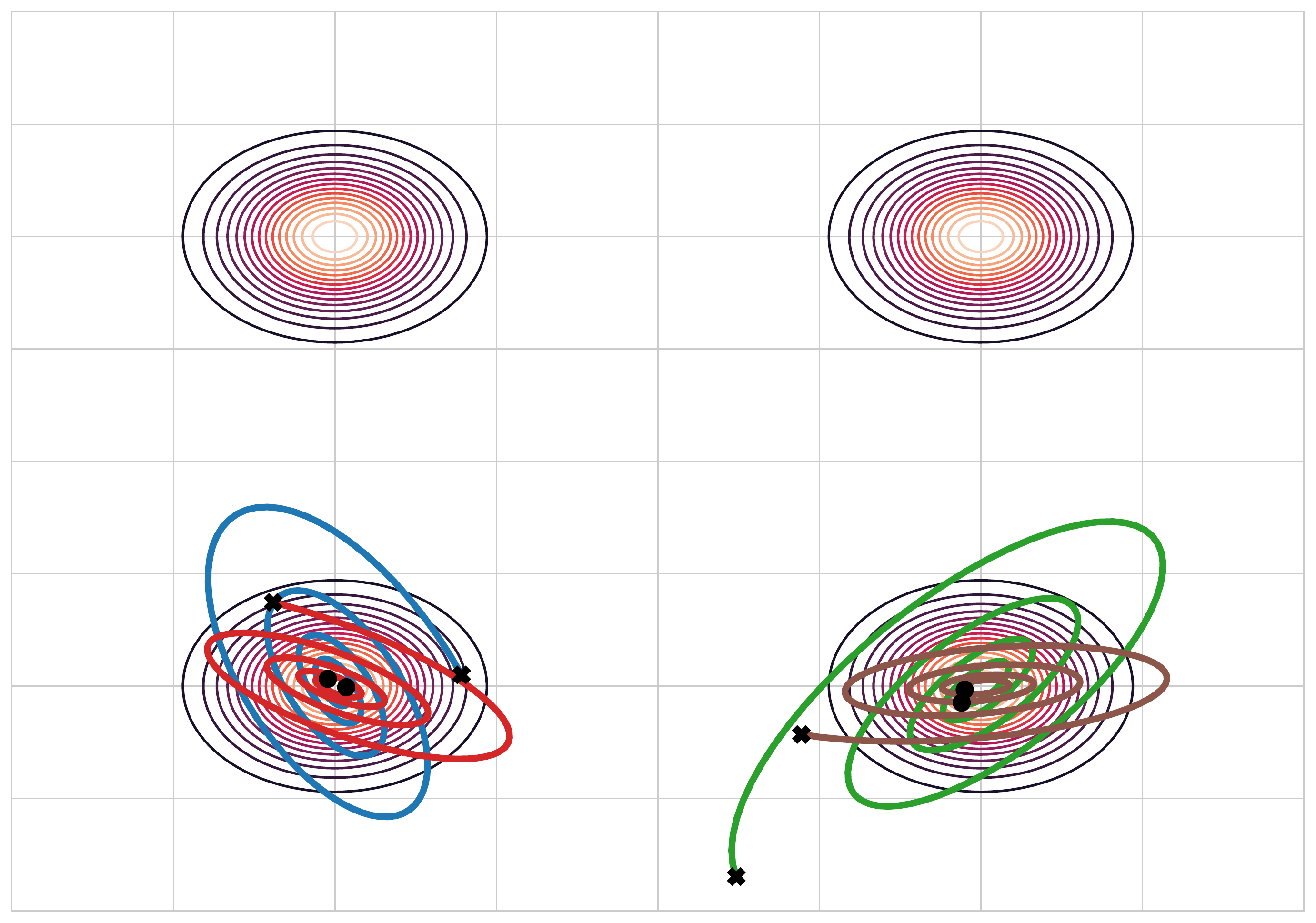}
     \caption{Left: $E^{\indi{[K]}} _{\transfo_h}(K)$ vs $E^{\operatorname{IS}}(K)$ (red) in $\log_{10}$-scale as a function of optimization step $K$.  Second left to right: Corresponding orbits for $\gamma = 0.1, 1, 2$.
     }
     \label{fig:orbit_E_T}
 \end{figure*}
In the applications below, we consider the conformal version of the symplectic Euler (SE) method of \eqref{eq:ODE_hamiltonian}, see \cite{francca2019conformal}.
This integrator can be constructed as a splitting of the two conformal and conservative parts of the system \eqref{eq:ODE_hamiltonian}. When composing a dissipative with a  symplectic operator, we set for all $(q,p) \in \rset^{2d}$, $\transfo_h(q,p)=
 (q+h\mass^{-1}\{ \rme^{-h\gamma} p -h \nabla U(q)\},
 \rme^{-h \gamma } p -h \nabla U(q))$,
where $h >0$ is a discretization stepsize.
This transformation can be connected with classical momentum optimization schemes, see \cite[Section 4]{francca2019conformal}.
By \cite[Section 3]{francca2019conformal}, for any $h >0$ $\transfo_h$ is a $\rmC^1$-diffeomorphism on $\rset^{2d}$ with Jacobian given by $\JacOp{\transfo_h}(q,p) = \rme^{-\gamma h d}$. In addition,  its inverse is
$  \transfo_h^{
-1}(q,p) = (q-h\mass^{-1} p,\rme^{\gamma h}\{p+h \nabla U(q-h\mass^{-1} p)\})$.
Therefore, the weight \eqref{eq:def_w_k} of the  \IFIS\  estimator is given by
$$
w_{k}(q,p) = \frac{ \varpi_k \tilde \rho(\transfo^k_h(q,p)) \rme^{-\gamma k h d} }{
    \sum_{j \in\zset}\varpi_{k+j}  \tilde \rho(\transfo^j_h(q,p)) \rme^{-\gamma j h d} }\eqsp,
$$
where $\tilde \rho(q,p) \propto \rho(q) \rme^{-K(p)}$.
\Cref{fig:orbit_E_T} displays for different values of $\gamma$, the bound  $E^{\indi{[0:K]}}_{\transfo_h}-1$ \eqref{eq:bound_z_infine} as a function of $K$ corresponding to the sequence of weights $(\varpi_k)_{k\in\zset} =  (\indi{[0:K]}(k))_{k\in\zset}$ (i.e. only the $K+1$ first elements of the forward orbits are used and are equally weighted). 
For comparison, we also present on the same plot, the bounds achieved  by averaging $K+1$ independent IS estimates, $E^{\operatorname{IS}}(K) -1  = (K+1)^{-1} \PE_{X \sim \proposal}[ \likelihood(X)^2]$. Interestingly, \Cref{fig:orbit_E_T} shows that there is a trade-off in the choice of $\gamma$ which controls the exploration of the state space by the Hamiltonian dynamics since the higher $\gamma$, the faster the orbits converge towards the modes.

\section{\IFIS-MCMC algorithm}
\label{sec:mcmc}
We now derive sampling methods based on the \IFIS-IS estimator. 
A  natural idea consists in adapting the Sampling Importance Resampling procedure (SIR) (see for example \cite{rubin1987comment,skare2003improved}) to the \InFiNE\ framework.
\begin{figure}
\begin{algorithm}[H]
At step $n \in \nsets$, given the conditioning orbit point $Y_{n-1}$.\\
\textbf{Step 1: Update the conditioning point} 
\begin{enumerate}
\item \label{algo:gibbs_partial_1} Set $X^{1}_{n} = Y_{n-1}$ and for any
  $i \in \{2,\dots,N\}$, sample $X_{n}^{i}\simiid \rho$. 
\item Sample the orbit index $I_{n}$  with probability proportional to
$(\estConstC{X_{n}^{i}})_{i \in [N]}$, \eqref{eq:def_estimator_normal_const}.
\item Set $Y_{n} = X^{I_{n}}_{n}$
\end{enumerate}
\textbf{Step 2: Output a sample}
\begin{enumerate}[resume] 
\item Sample index $K_{n}$ with probability proportional to $\{w_k(Y_{n})
\likelihood(\transfo^k(Y_{n}))/ \estConstC{Y_{n}}\}_{k\in\zset}$
\item Output $U_{n}= \transfo^{K_{n}}(Y_{n})$. 
\end{enumerate}
\caption{\InFiNE-MCMC Sampler}
\label{algo:gibbs_partial}
\end{algorithm}
\end{figure}
The SIR method  to sample $\infineSNIS$ (see  \eqref{eq:def_estimator_naive_monte_carlo}) consists of 4 steps. 
\\(SIR-1) Draw independently $\chunku{X}{1}{N}\simiid\rho$ and compute the associated forward and backward orbits $\{ \transfo^k(X^i)\}_{k \in \zset}$ of the point.
\\(SIR-2) Compute the normalizing constants associated with each orbit $\{\estConstC{X^i} \}_{i=1}^N$.
\\(SIR-3) Sample  an orbit index $I_N \in[N]$ with probability $\{\estConstC{X^i}/{\sum_{j=1}^N \estConstC{X^j}}\}_{i=1}^N$.
\\(SIR-4) Draw the iteration index $K^N$ on the $I^N$-th orbit with probability $\{ \likelihood(\transfo^{k}(X^{I^N})) w_k(X^{I^N})/{\estConstC{X^{I^N}}} \}_{k \in \zset}$.

The resulting draw is denoted by $U^N= \transfo^{K^N}(X^{I^N})$. 
By construction, for any bounded function $f$, we get that $\CPE{f(U^N)}{\chunku{X}{1}{N},I^N}= \{ \estConstC{X^{I^N}} \}^{-1} \sum_{k \in \zset} w_k(X_{I^N})\likelihood(\transfo^k(X_{I^N}))$  which implies $\CPE{f(U^N)}{\chunku{X}{1}{N},I^N} =\infineSNIS(f)$ (see \eqref{eq:def_estimator_naive_monte_carlo}). Using \Cref{theo:bias_mse_snis}, we therefore obtain $|\PE[f(U^N)]- \int f(z) \target(z) \rmd z| \leq 10^{1/2}\|f\|_\infty E^{\varpi} _{\transfo} N^{-1}$, showing that the law of the random variable $\mu_N= \operatorname{Law}(U^N)$ converges in total variation to $\target$ as $N \to \infty$, 
\begin{equation}
\label{eq:convergence-tvnorm}
\tvnorm{\mu_N - \pi} = \sup_{\|f\|_\infty \leq 1} |\mu_N(f) - \pi(f)| \leq  10^{1/2} E^{\varpi} _{\transfo} N^{-1} \eqsp.
\end{equation}

Based on  \cite{andrieu2010particle}, we now derive the \InFiNE-MCMC procedure, which in a nutshell consists in iterating the SIR procedure while keeping a conditioning point (or equivalently, orbit); see \Cref{app:i-SIR}. The convergence of \IFIS-MCMC does not rely on letting $N \to \infty$: the \IFIS-MCMC works as soon as $N \geq 2$, although as we will see below the mixing time decreases as $N$ increases.

This procedure is summarized in \Cref{algo:gibbs_partial}. The \IFIS-MCMC procedure is an iterated algorithm which produces a sequence $\{(Y_n,U_n)\}_{n \in \nset}$ of points in $\rset^d$. The $n$-th iteration of the \IFIS-MCMC algorithm consists in two main steps: 1) updating the conditioning point $Y_{n-1} \to Y_n$ 2) sampling $U_n$ by selecting a point in the  orbit $\{ \transfo^k(Y_{n}) \}_{k \in \zset}$ of the conditioning point. Compared to SIR, only the generation of the points (step (SIR-1)) is modified: we set $X_n^1= Y_{n-1}$ (the \textbf{conditioning point}), and then draw $\chunku{X_n}{2:N} \simiid \proposal$. The sequence $\{Y_n\}_{n \in \nset}$ defined by \Cref{algo:gibbs_partial} is a Markov chain: 
$$
\CPP{Y_n \in \msa}{Y_{0:n-1}}= \CPP{Y_n \in \msa}{Y_{n-1}}=P(Y_n,\msa)\eqsp,
$$
where 
\begin{equation}
\label{eq:definition-P}
    P(y, \msa) = \int \updelta_{y}(\rmd x^1) \prod_{j=2}^N\rho(x^j)\rmd x^j \sum_{i=1}^N\frac{\estConstC{x^i}}{\sum_{j=1}^N \estConstC{x^j}}\indi{\msa}(x^i)\eqsp, \quad y \in \rset^d\eqsp, \msa \in \borel(\rset^d)\eqsp.
\end{equation}
Note that this Markov kernel describes the way, at stage $n+1$, the conditioning point $Y_{n+1}$ is selected given $Y_n$, which \textbf{depends only on} the estimator of the normalizing constants associated with each orbit, \textbf{but not} on the sample $U_n$ selected on the conditioning orbit. In addition, given the conditioning point $Y_n$ at the $n$-th iteration, the conditional distribution of the output sample $U_n$ is $\CPP{U_n \in \msb}{I_n, \chunku{X_n}{1}{N}}= \CPP{U_n \in \msb}{Y_n} = Q(Y_n,\msb)$  where 
\begin{equation}
\label{eq:definition-Q}
 Q(y, \msb) =  \sum_{k\in\zset} \frac{w_k(y)\likelihood(\transfo^k(y)) }{\estConstC{y}}\1_\msb(\transfo^k(y))\eqsp, \quad y \in \rset^d \eqsp, \msb \in \borel(\rset^d)\eqsp.
\end{equation}
With these notations, if the Markov chain is started at $Y_0= y$, then for any $n\in\nset$, the law of the $n$-th conditioning point is $\CPP{Y_n \in \msa}{Y_0=y}= P^n(y, \msa)$ and the law of the $n$-th sample  is $\CPP{U_n \in \msb}{Y_0}= P^nQ(y, \msb)$. Define $\tpi$ the pdf given, for $y\in\rset^d$, by
\begin{equation}
\tpi(y)= \frac{\rho(y)}{\const} \sum_{k\in\zset}w_k(y) \likelihood(\transfo^k(y)) = \frac{\rho(y)\estConstC{y}}{\const}\eqsp.    
\end{equation}
The following theorem shows that, for any initial condition $y \in \rset^d$, the distribution of the variable $Y_n$ converges in total variation to $\tpi$ and that the distribution of $U_n$ converges to $\target$.
\begin{theorem}
\label{lem:invariant_P}
The Markov kernel $P$ is reversible \wrt~the distribution $\tpi$,  ergodic and Harris positive, \ie, for all $y\in \rset^d$,
$\lim_{n\to\infty} \tvnorm{P^n(y, \cdot) - \tpi} = 0$. In addition,  $\pi = \tpi Q$ and 
$\lim_{n\to\infty} \tvnorm{P^n Q(y, \cdot) - \pi} = 0$.
Moreover, for any bounded function $g$ and any $y\in\rset^d$, $\lim_{n\to\infty}n^{-1}\sum_{i=0}^{n-1} g(U_i) = \pi(g)$, $\PP$-almost surely, where $\{U_i\}_{i \in \nset}$ is defined in \Cref{algo:gibbs_partial} with $Y_0 = y$. 
\end{theorem}
The proof is postponed to \Cref{sec:supp:proof_mcmc}. 
\begin{remark}\em
\label{rem:gibbs-interpretation}
We may provide another sampling procedure of $\{ Y_n \}_{n\in\nset}$.
Define the pdf on the extended space $[N] \times\rset^{d N}$  by  
$\check{\pi}(i, \chunku{x}{1}{N}) = N^{-1} \tpi(x^i) \prod_{j=1, j\neq i}^N \rho(x^j)$.
Consider a Gibbs sampler  targeting $\check{\pi}$ consisting in   (a) sampling $\chunkum{X_n}{1}{N}{I_{n-1}}|(I_{n-1}, X_{n-1})\sim\prod_{j\neq I_{n-1}}\rho(x^j)$,   (b) sampling  $I_n|\chunku{X_n}{1}{N} \sim \operatorname{Cat}( \{\estConstC{X_n^i} / \sum_{j=1}^N \estConstC{X_n^j} \}_{i=1}^N $ and (c) set $Y_n = X_n^{I_n}$.  This algorithm is a Gibbs sampler on $\check{\pi}$ and we easily verify that the distribution of $\{Y_n\}_{n\in\nset}$ is the same as \Cref{algo:gibbs_partial}.
\end{remark}
The next theorem provides non asymptotic quantitative bounds on the convergence in total variation. The main  interest of \InFiNE-MCMC algorithm is motivated empirically from observed behaviour: the mixing time of the corresponding Markov chain improves as $N$ increases. This behaviour is quantified theoretically in the next theorem. Moreover, this improvement is  obtained with little extra computational overhead, since  sampling $N$ points  from the proposal distribution $\proposal$, computing the forward and backward orbits of the points and evaluating the normalizing constants $\{\estConstC{X_n^i}\}_{i=1}^N$ can be performed in parallel.
\begin{theorem}
\label{theo:geom_ergodicity_infine}
Assume that $\bound<\infty$, see \eqref{eq:bound_z_infine}. Set  $\epsilon_N =  {(N -1)}/{(2\bound + N-2)}$ and $\kappa_N= 1 -\epsilon_N$.
Then, for any $y\in\rset^d$ and $k \in \nset$, $\tvnorm{P^{k}(y, \cdot) - \tpi}\leq \kappa_N^k$ and  $\tvnorm{P^{k}Q(y, \cdot) - \pi}\leq \kappa_N^k$.
\end{theorem}
Instead of sampling the new points $\chunku{X}{2}{N}_n$ independently from $\proposal$ (Step 1 in \Cref{algo:gibbs_partial}), it is possible to draw the proposals $\chunku{X_n}{1}{N}$ conditional to the current point $Y_{n-1}$; see \cite{so2006bayesian,craiu2007acceleration,shestopaloff:neal:2018,ruiz:titsias:doucet:2020} for related works. 
Following \cite{ruiz:titsias:doucet:2020}, we use a reversible Markov kernel  \wrt\ the proposal $\proposal$, i.e., such that $\proposal(x) m(x,x')= \proposal(x') m(x',x)$, assuming for simplicity that this kernel has density $m(x,x')$.  
If $\proposal= \Normal(0, \sigma^2\Id_d)$ , an appropriate choice is an autoregressive kernel $m(x,x')= \Normal(x';\alpha x,  \sigma^2(1-\alpha^2)\Id_d)$. More generally, we can use a Metropolis--Hastings kernel with invariant distribution $\proposal$.
In this case, for each $i \in [N]$, define for $i \in [N]$, 
\begin{equation}
\label{eq:defintion-r-i}
r_i(x^i,\chunkum{x}{1}{N}{i}) =
 \prod_{j=1}^{i-1} m(x^{j+1},x^j)\prod_{j=i+1}^N m(x^{j-1},x^j)\eqsp.
\end{equation}
Since $m$ is reversible \wrt\ $\proposal$,   for all $i,j\in[N]$, $\proposal(x^i) r_i(x^i,\chunkum{x}{1}{N}{i}) = \proposal(x^j) r_j(x^j,\chunkum{x}{1}{N}{j})$. The only modification in \Cref{algo:gibbs_partial} is Step 1, which is replaced by:  \emph{Draw $U_n \in[N]$ uniformly, set $X_{n}^{U_n}= Y_{n-1}$ and sample $\chunkum{X_{n}}{1}{N}{U_n} \sim r_{U_n}(X^{U_n}_n,\cdot)$}.
The validity of this procedure is established in \Cref{sec:supp:proof_mcmc}.

\section{Continuous-time version of \IFIS\ and \NEIS}
The \NEO\ framework takes up and extends \NEIS\ introduced in \cite{rotskoff:vanden-eijden:2019}.  \NEIS\ focuses on normalizing constant estimation and should be therefore compared with \NEO-IS. In \cite{rotskoff:vanden-eijden:2019}, the authors do not consider possible extensions of these ideas to sampling problems. Proofs of the statements and detailed technical conditions are postponed to  \Cref{sec:continuous-time-limits}.
We first consider how \NEO\ can be adapted to continuous-time dynamical system. 
Consider the Ordinary Differential Equation (ODE) $\dot{x}_t = b(x_t)\eqsp,$ where $b\colon\rset^d\to\rset^d$ is a smooth vector field. Denote by $(\phi_t)_{t\in\rset}$ the flow of this ODE (assumed to be well-behaved). Under appropriate regularity condition $ \JacOp{\phi_t}(x) = \exp(\int_0^t\nabla\cdot b(\phi_s(x))\rmd s)$; see \Cref{lem:jacobianflow}. Let  $\varpi: \rset \to \rset_+$ be a nonnegative smooth function with finite support, with $\Omega^c = \int_{-\infty}^{\infty} \varpi(t)\rmd t$.
The continuous-time counterpart of the proposal distribution \eqref{eq:rhoT} is $\rhoTcont(x) =(\Omega^c)^{-1} \int_{-\infty}^\infty \varpi(t) \rho(\phi_{-t}(x))\JacOp{\phi_{-t}}(x)\rmd t$, which is a continuous mixture of the pushforward of the proposal $\proposal$ by the flow of $( \phi_{s} )_{s \in\rset}$. 
Assuming for simplicity that $\proposal(x) > 0$ for all $x \in \rset^d$, then $\rhoTcont(x) > 0$ for all $x \in \rset^d$, and using again the IS formula, for any nonnegative function $f$, 
\begin{align}
\label{eq:infinecontinuous}
    &\int f(x) \rho(x) \rmd x = \int f(x) \frac{\rho(x)}{\rhoTcont(x)}\rhoTcont(x) \rmd x= \int \left[\int_{-\infty}^{\infty} \wcont(x) f(\phi_t(x))\rmd t\right]\rho(x) \rmd x\eqsp,\\
\label{eq:continuous_limit_infine_weights}
   &\wcont( x) = \left.{\varpi(t)\rho(\phi_t(x)) \JacOp{\phi_t}(x)}\middle/{\int_{-\infty}^{\infty}\varpi({s+t})\rho(\phi_s(x)) \JacOp{\phi_s}(x)\rmd s}\right.\eqsp.
\end{align}
These relations are the continuous-time counterparts of \eqref{eq:key-relation}. Eqs. \eqref{eq:infinecontinuous}-\eqref{eq:continuous_limit_infine_weights} define a version of \NEIS\  \cite{rotskoff:vanden-eijden:2019}, with a finite support weight function $\varpi$; see \Cref{sec:neis,sec:infine_stopping_times} for weight functions with infinite support. 
This identity is of theoretical interest but must be discretized to obtain a  computationally tractable estimator. For  $h>0$, denote by $\transfo_h$  an integrator with stepsize $h > 0$ of the ODE $\dot{x} = b(x)$. We may construct \InFiNE-IS and \InFiNE-SNIS estimators based on the transform $\transfo \leftarrow\transfo_h$ and weights $\varpi_{k} \leftarrow {\varpi(kh)}$. 
We might show that for any bounded function $f$ and for any $x\in\rset^d$, 
$\lim_{h\downarrow0}\sum_{k\in\zset} w_{k}(x) f(\transfo^k_h(x)) = \int_{-\infty}^\infty \wcont(x)f(\phi_t(x))\rmd t$, where we omitted here the dependency in $h$ of $w_k$.
Therefore, taking $h \downarrow 0^+$, the  \NEO-IS converges to the continuous-version \eqref{eq:infinecontinuous}-\eqref{eq:continuous_limit_infine_weights}. 
There is however an important difference between \IFIS\ and the  \NEIS\ method in \cite{rotskoff:vanden-eijden:2019}  which stems from the way \eqref{eq:infinecontinuous}-\eqref{eq:continuous_limit_infine_weights} are discretized. Compared to \NEIS, \NEO-IS using $\transfo \leftarrow\transfo_h$ and weights $\varpi_{k} \leftarrow {\varpi(kh)}$ is unbiased for any stepsize $h > 0$. \NEIS\ uses an approach inspired by the nested-sampling approach, which amounts to discretizing the integral in \eqref{eq:infinecontinuous} also in the state-variable $x$; see \cite{skilling2006nested,chopin:robert:2010}. This discretization is biased which prevents the use of this approach to develop MCMC sampling algorithm; see \Cref{sec:continuous-time-limits}.  
 
\section{Experiments and Applications}
\paragraph{Normalizing constant estimation}
\label{subsec:estim_constant}
The performance of \IFIS-IS is assessed on different normalizing constant  estimation benchmarks; see \cite{jia2020normalizing}. 
We focus on two challenging examples. Additional experiments and discussion on hyperparameter choice are given in the supplementary material, see \Cref{sup:sec:additional_xp}. 
\textbf{(1) Mixture of Gaussian \texttt{(MG25)}}: $\target(x)= P^{-1} \sum_{i=1}^{P} \Normal(x;\mu_{i,j},D_d)$, where $d\in\{10,20,40\}$, $D_d= \mathrm{diag}(0.01,0.01,0.1,\dots,0.1)$ and $\mu_{i,j}= [i,j,0,\ldots,0]^T$ with $i,j \in \{-2,\ldots, 2\}$.  \textbf{(2) Funnel distribution \texttt{(Fun)}} $\target(x)= \Normal(x_1;0,a^2)\prod_{i=1}^d \Normal(x_i;0,\rme^{2b x_1})$ with $d\in\{10,20,40\}$, $a=1$, and $b=0.5$. In both case, the proposal is $\rho= \Normal(0,\sigma^2_\rho \Id_d)$ with $\sigma^2_\rho=5$. 

 The $\IFIS$-IS estimator is compared with (i) the IS estimator using the proposal $\rho$, (ii) the   Adaptive Importance Sampling (AIS) estimator of \cite{tokdar2010importance} and (iii) the Neural Importance Sampling (NIS)\footnote{We used the implementation provided by: \url{https://github.com/ndeutschmann/zunis}}.  For \NEO-IS, we use $\varpi_k= \indi{[K]}(k)$  with $K=10$ (ten steps on the forward orbit), and conformal Hamiltonian dynamics  $\gamma = 1$, $M= 5 \cdot \Id_d$ for dimensions  $d=\{10,20\}$, and $\gamma = 2.5$ for $d=40$ (where $\gamma$ is the damping factor, $M$ the mass matrix,  $h$ is the stepsize of the integrator). The parameters of AIS are set to obtain a complexity comparable to \NEO-IS; see \Cref{sup:sec:additional_xp}. For NIS, we use the default parameters.
 In \texttt{Fun}, we set $\gamma = 0.2$, $K = 10$, $M = 5 \cdot \Id_d$, and $h=0.3$. The IS estimator was based on $5 \cdot 10^5$ samples, and NIS, \NEO-IS and AIS were computed with $5 \cdot 10^4$  samples.
 \Cref{fig:funnel_estimation} shows that \IFIS-IS consistently outperforms the competing methods. 
\begin{figure}[!ht]
    \centering
    \includegraphics[width=0.32\linewidth]{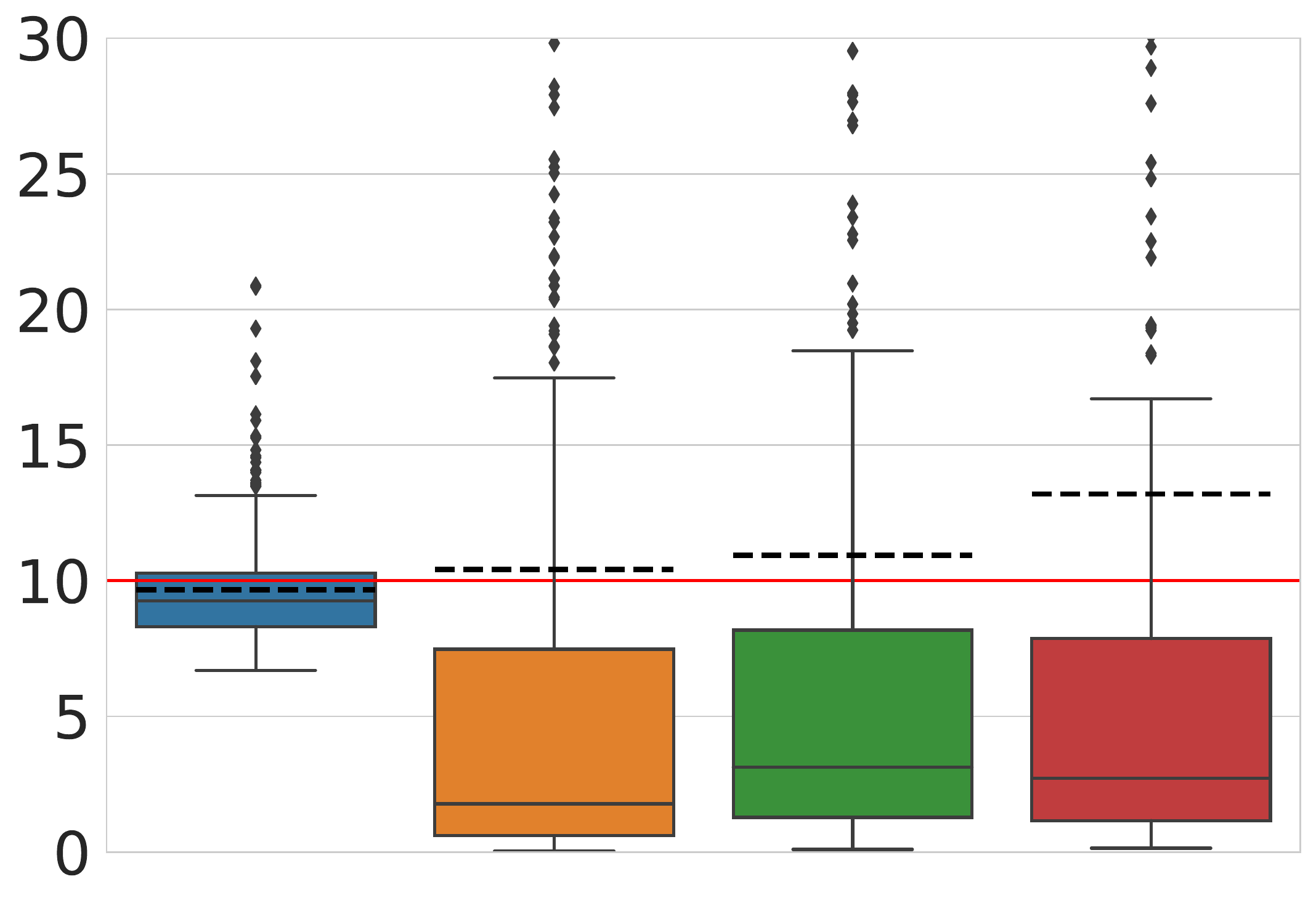}
    \includegraphics[width=0.32\linewidth]{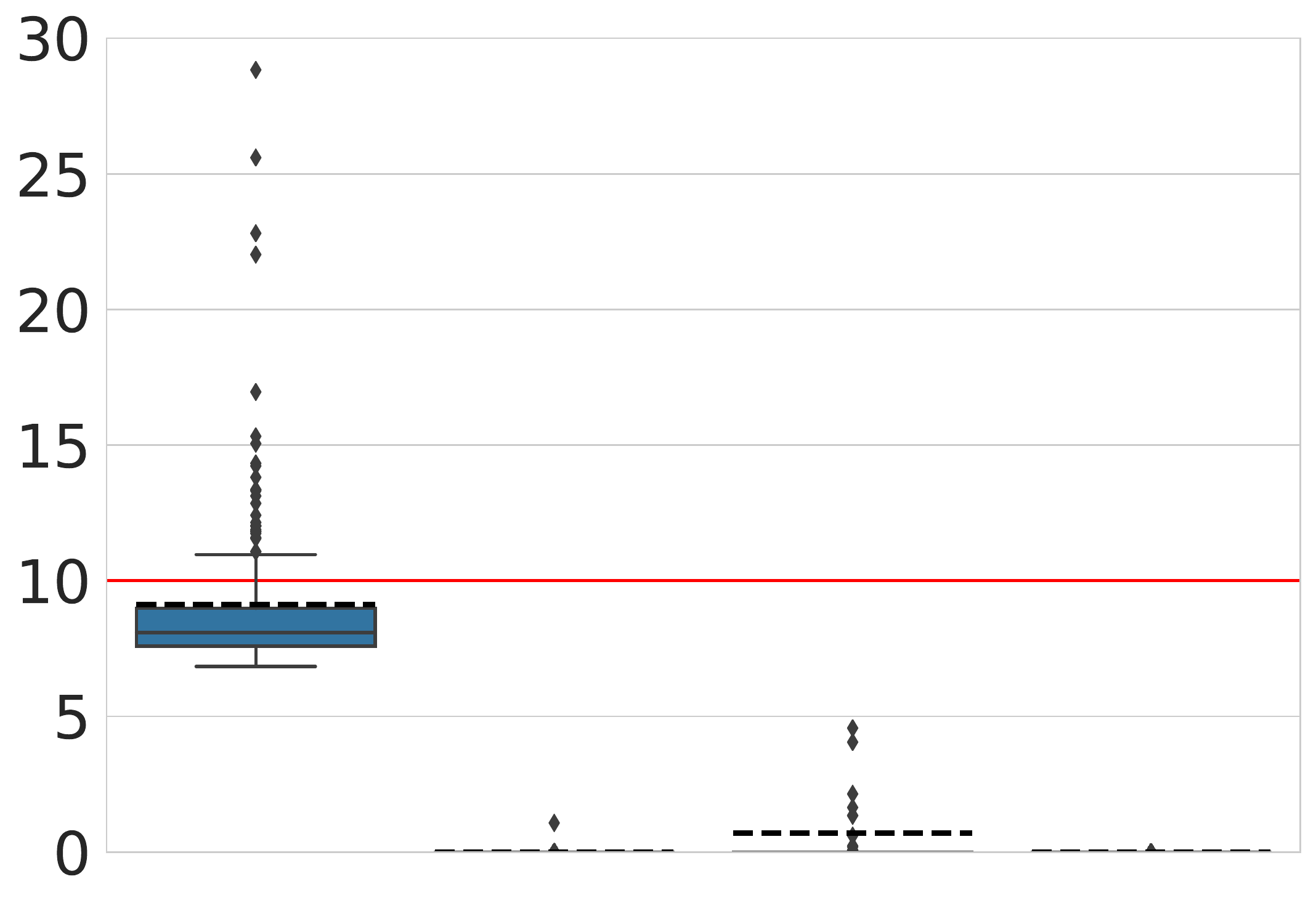}
    \includegraphics[width=0.32\linewidth]{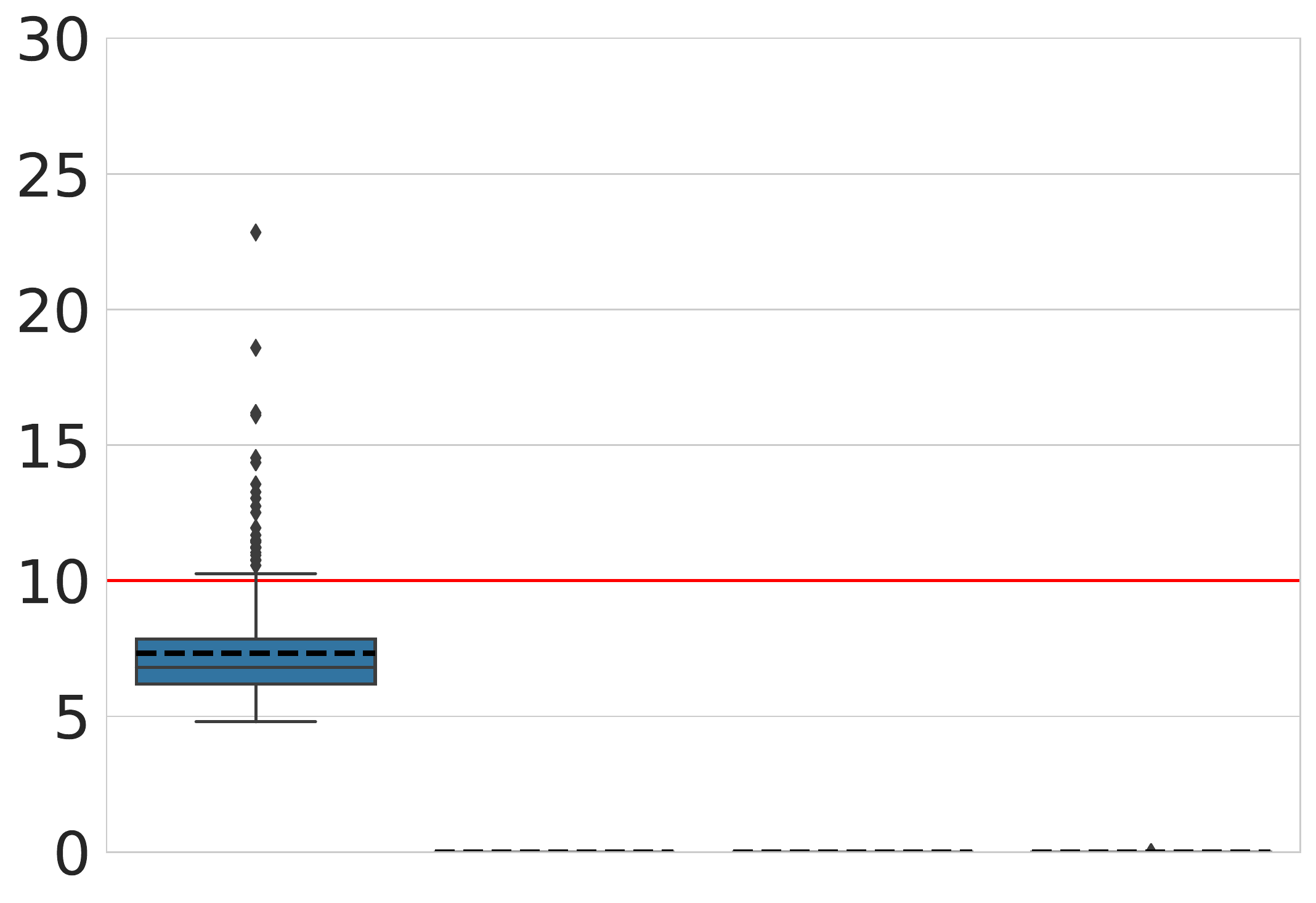}
\\
    \centering
    \includegraphics[width=0.32\linewidth]{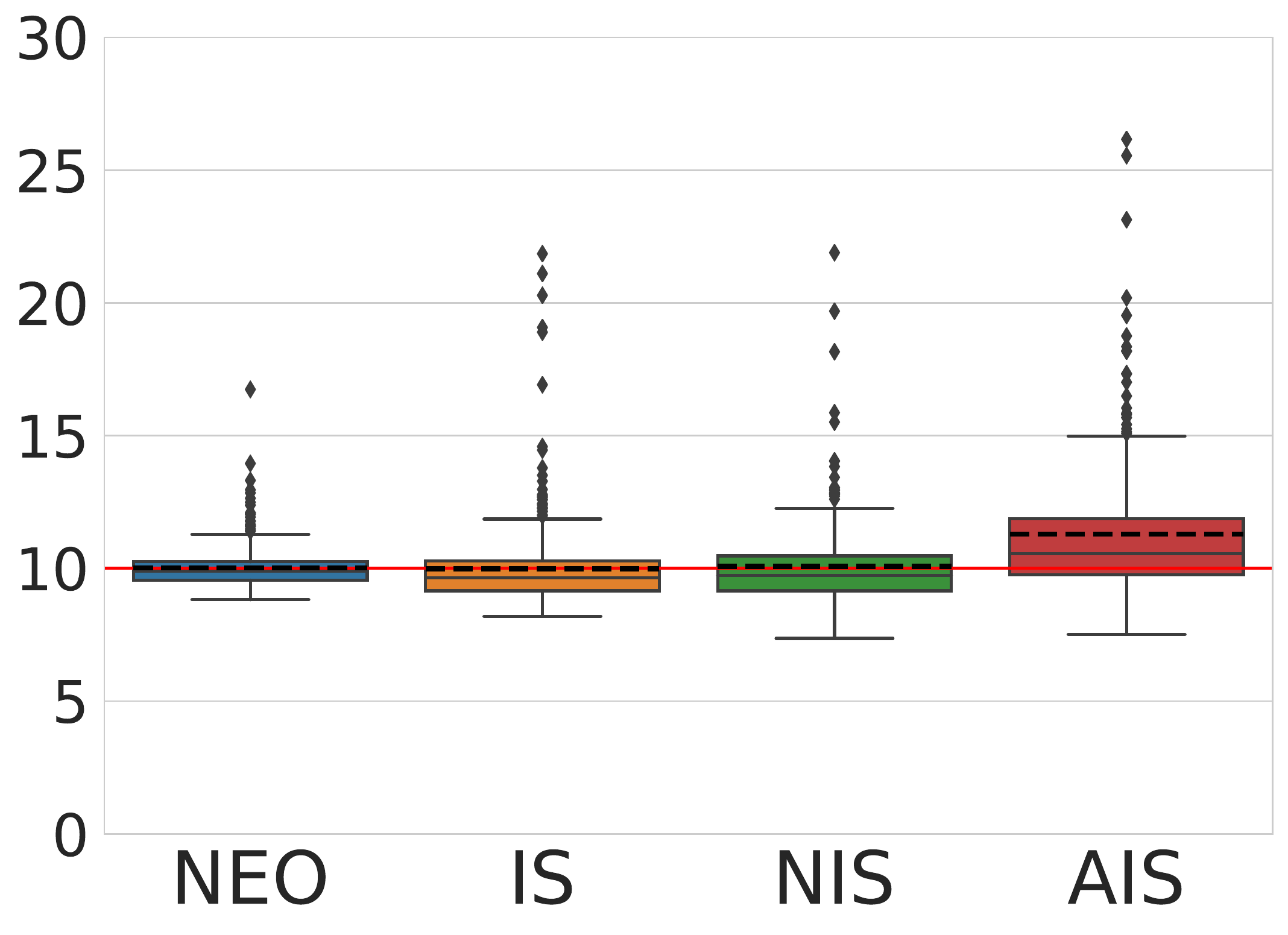}
    \includegraphics[width=0.32\linewidth]{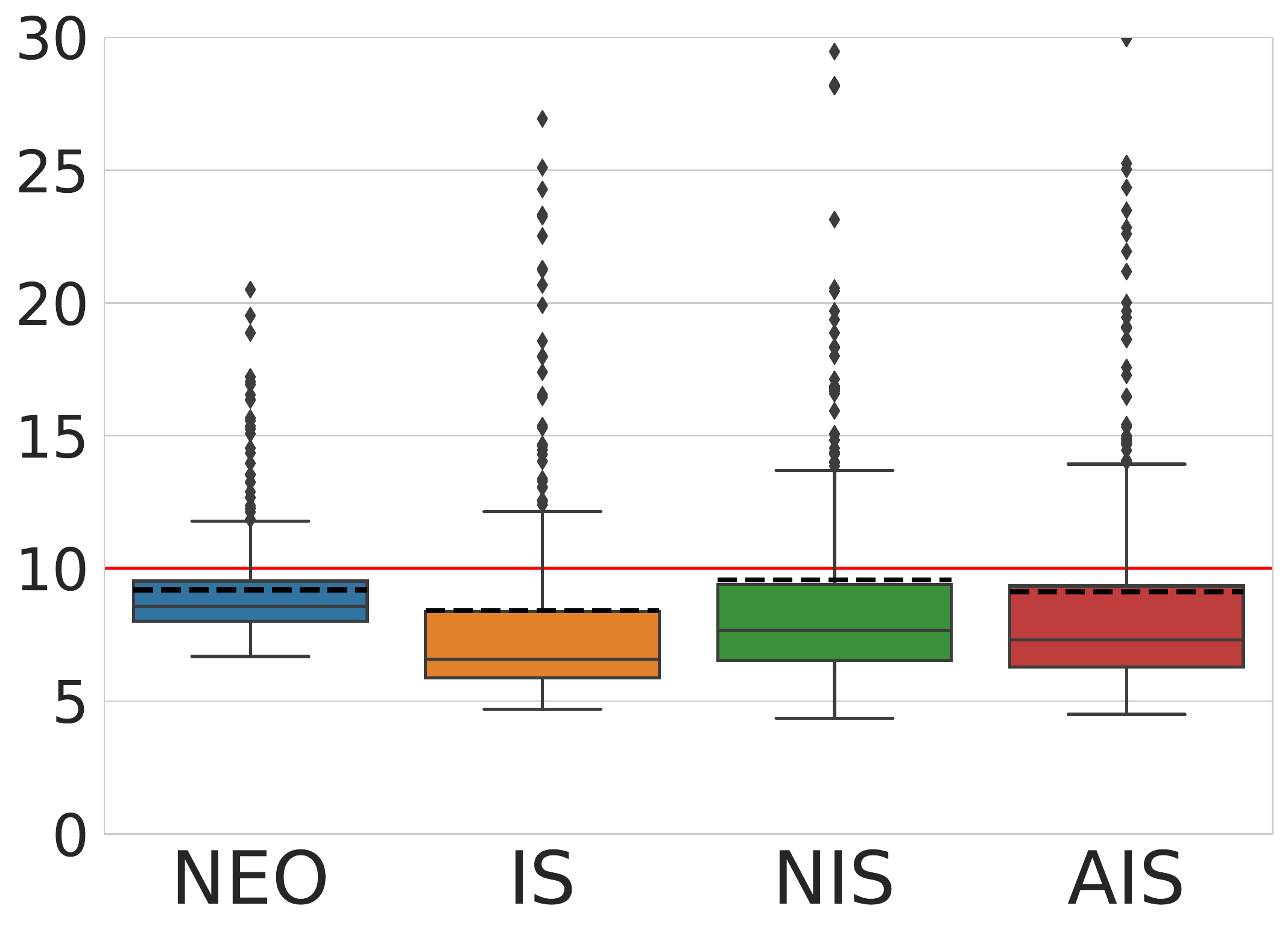}
    \includegraphics[width=0.32\linewidth]{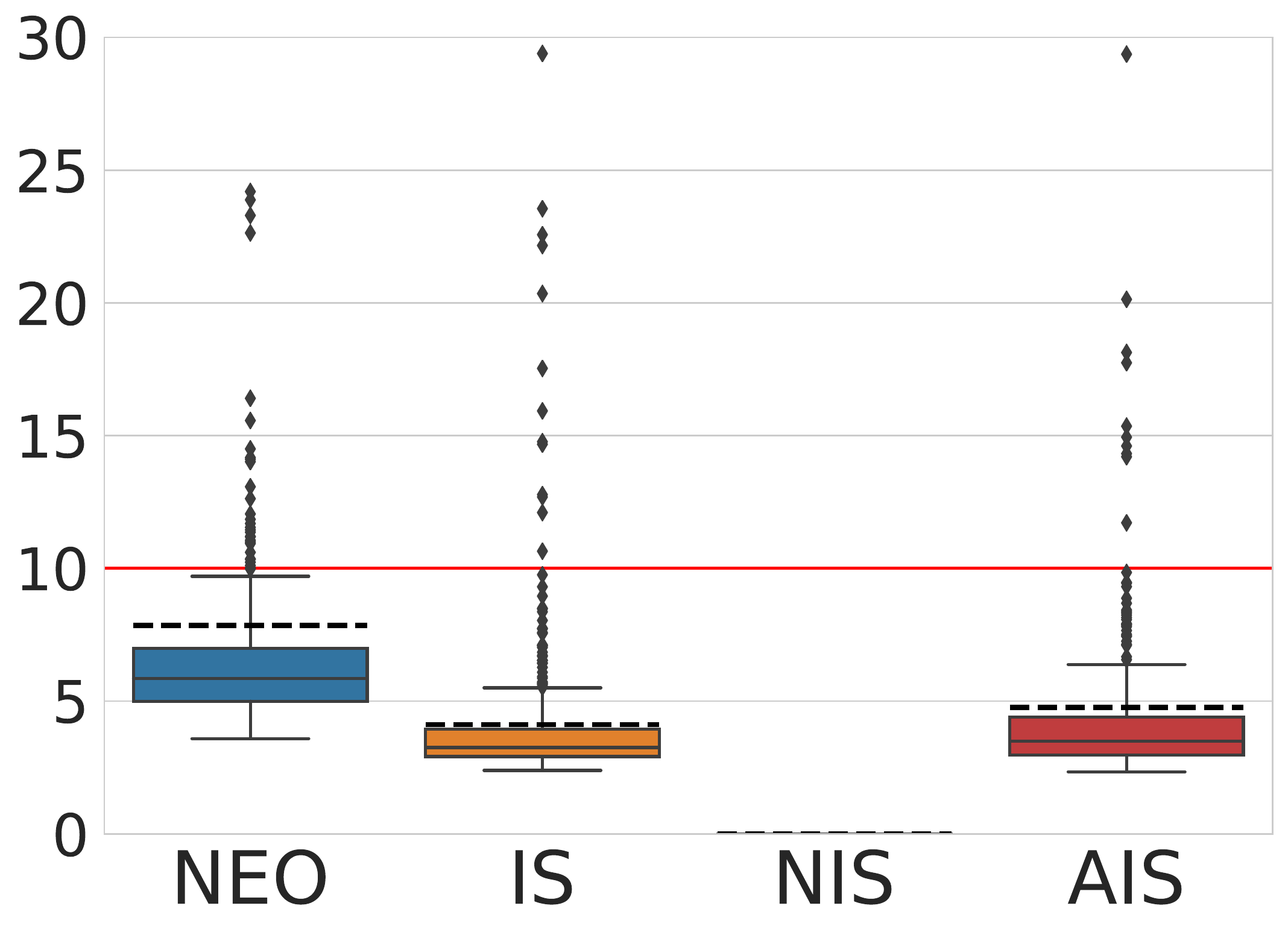}
    \caption{ Boxplots of 500 independent estimations of the normalizing constant in dimension $d=\{10, 20, 45\}$ (from left to right) for \texttt{MG25}  (top) and \texttt{Fun}  (bottom). The true value is given by the red line. The figure displays the median (solid lines), the interquartile range, and the mean (dashed lines) over the 500 runs.}
    \label{fig:funnel_estimation}
        \centering
        \end{figure}
\paragraph{Sampling}
\label{subsec:mcmc_exp}
\InFiNE-MCMC  is assessed  for the distributions \texttt{(MG25)} ($d=40$) and \texttt{Fun} ($d=20$).   \InFiNE-MCMC sampler is compared with (i) the No-U-Turn Sampler - Pyro library \cite{bingham2019pyro} - and (ii) i-SIR algorithm \cite{ruiz:titsias:doucet:2020}. The proposal distribution is $\proposal = \Normal(0,\sigma^2_\rho \Id_d)$ with $\sigma^2_\rho = 5$.
Dependent proposals are used (see \eqref{eq:defintion-r-i}) with $m(x,x')= \Normal(x';\alpha x, \sigma^2_\rho(1-\alpha^2)  \Id_d)$ with $\alpha=0.99$.
For NUTS, the default parameters are used. 
For i-SIR, we use the same number of proposals $N=10$,  proposal distribution and dependent proposal as for \IFIS-MCMC.  
To make a fair comparison, we use the same clock time for all three algorithms. The number of iterations for correlated i-SIR, \IFIS-MCMC, and NUTS are $n= 4\cdot 10^6$, $n= 4 \cdot 10^5$, and $n= 5 \cdot 10^5$, respectively.
\begin{figure}[!ht]
    \centering
      \includegraphics[width = .22\linewidth]{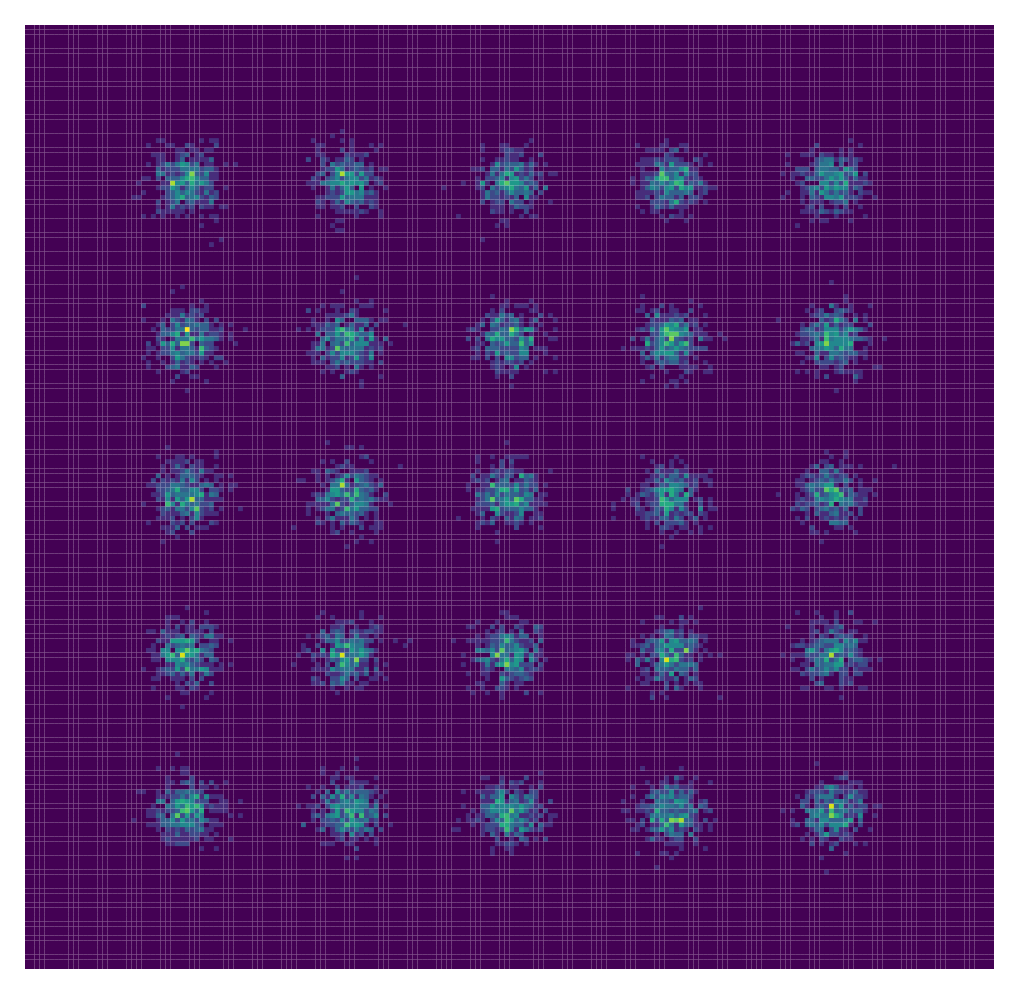}
      \includegraphics[width = .22\linewidth]{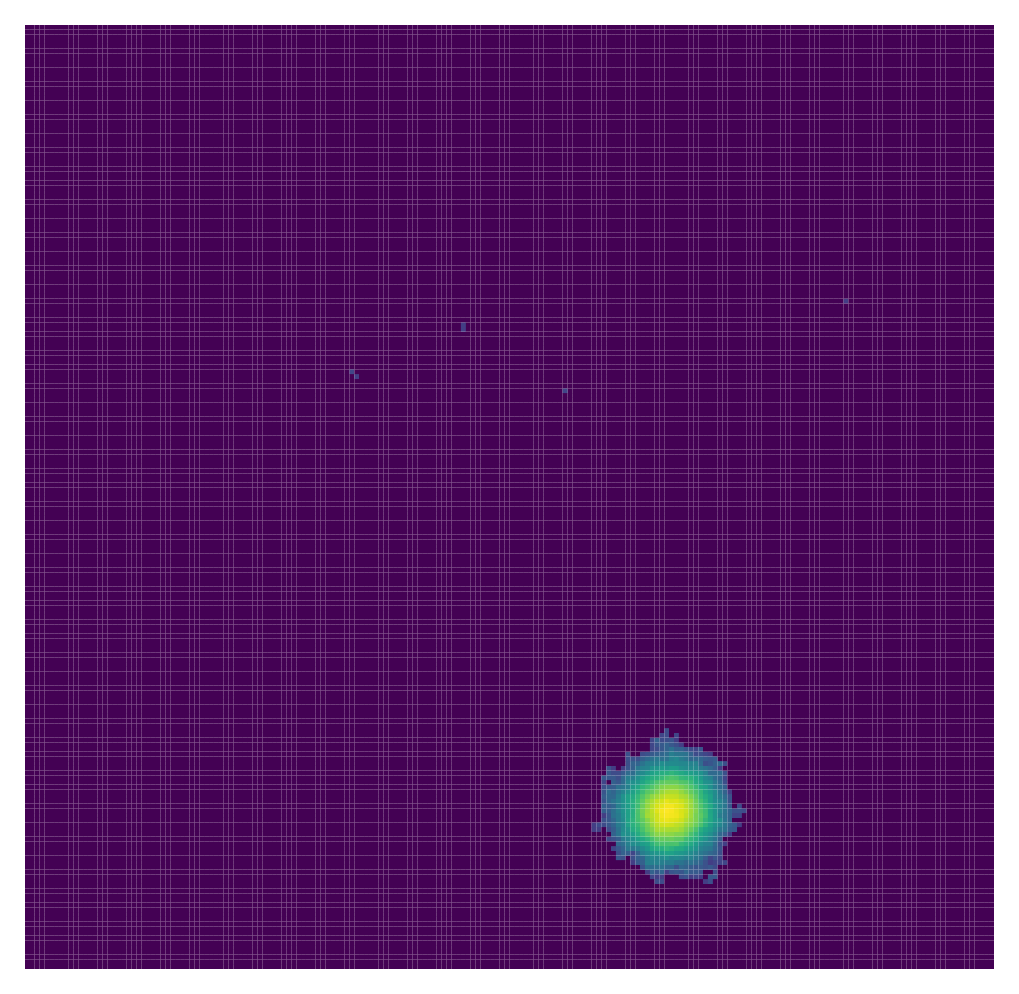} 
       \includegraphics[width = .22\linewidth]{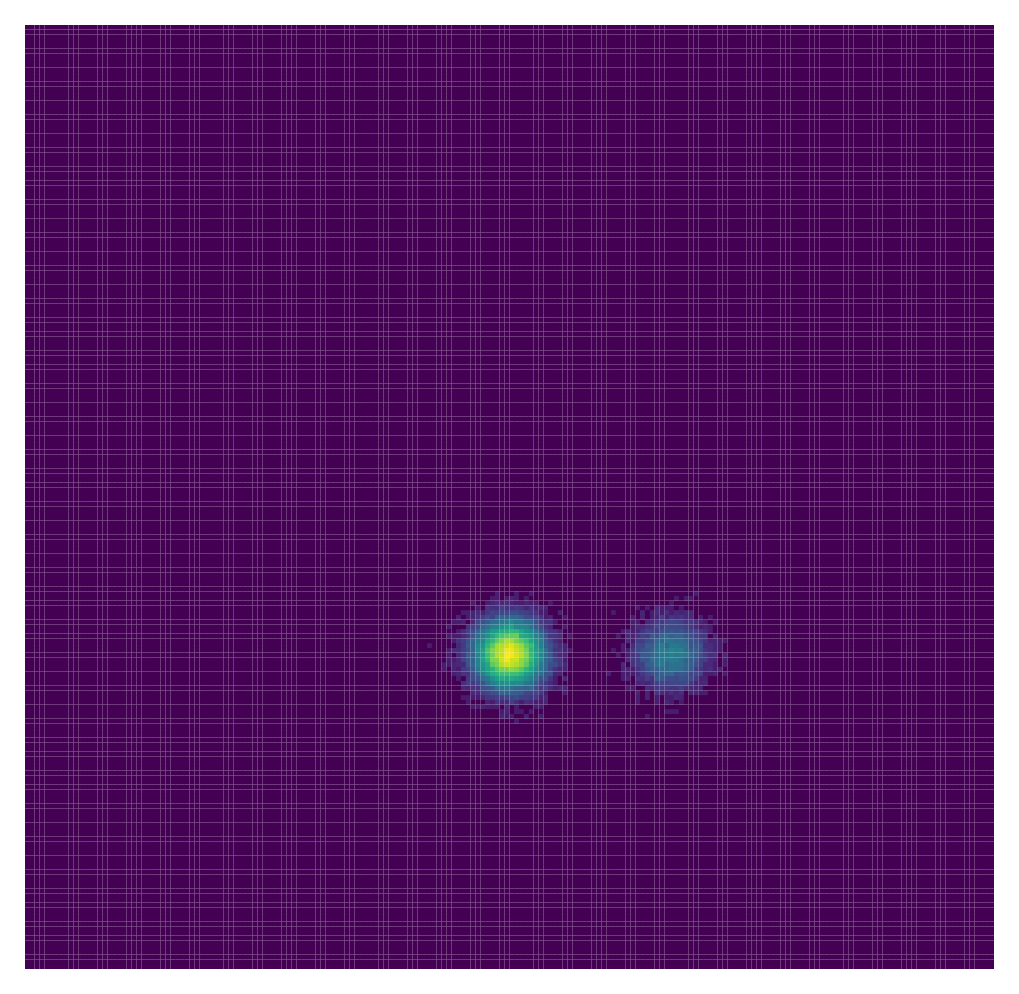}
       \includegraphics[width = .22\linewidth]{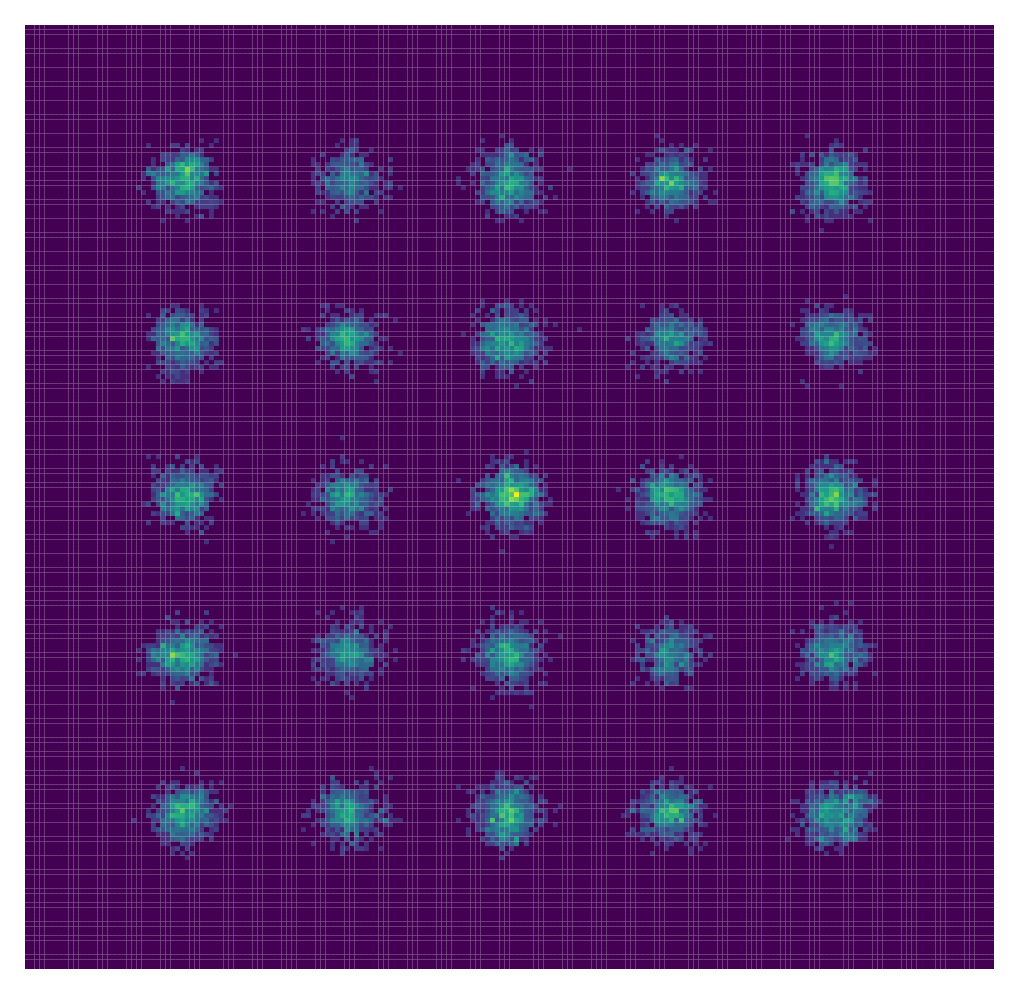}
    \\
    \centering
       \hspace{1pt}\includegraphics[width=.22\linewidth]{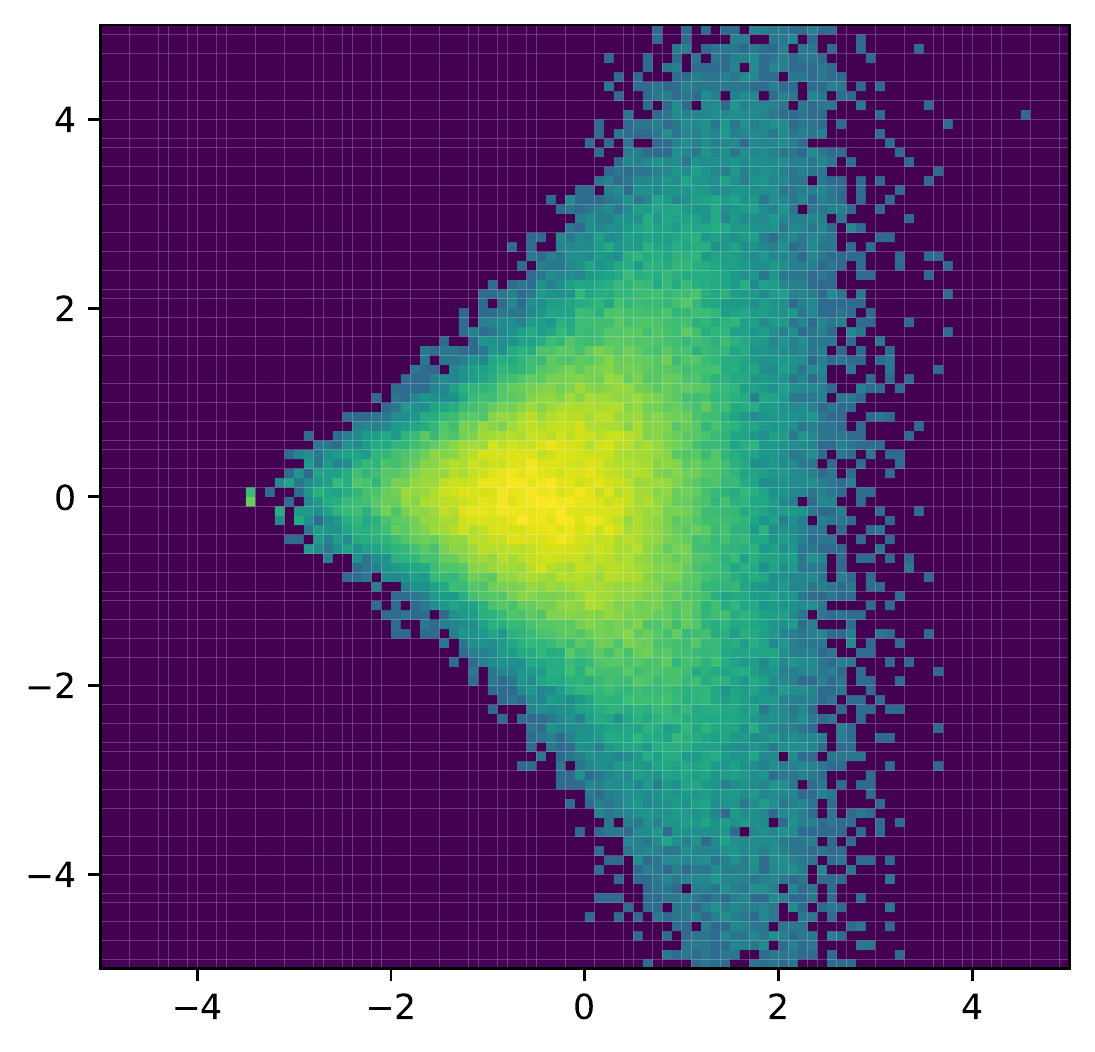}
         \hspace{1pt} \includegraphics[width=.22\linewidth]{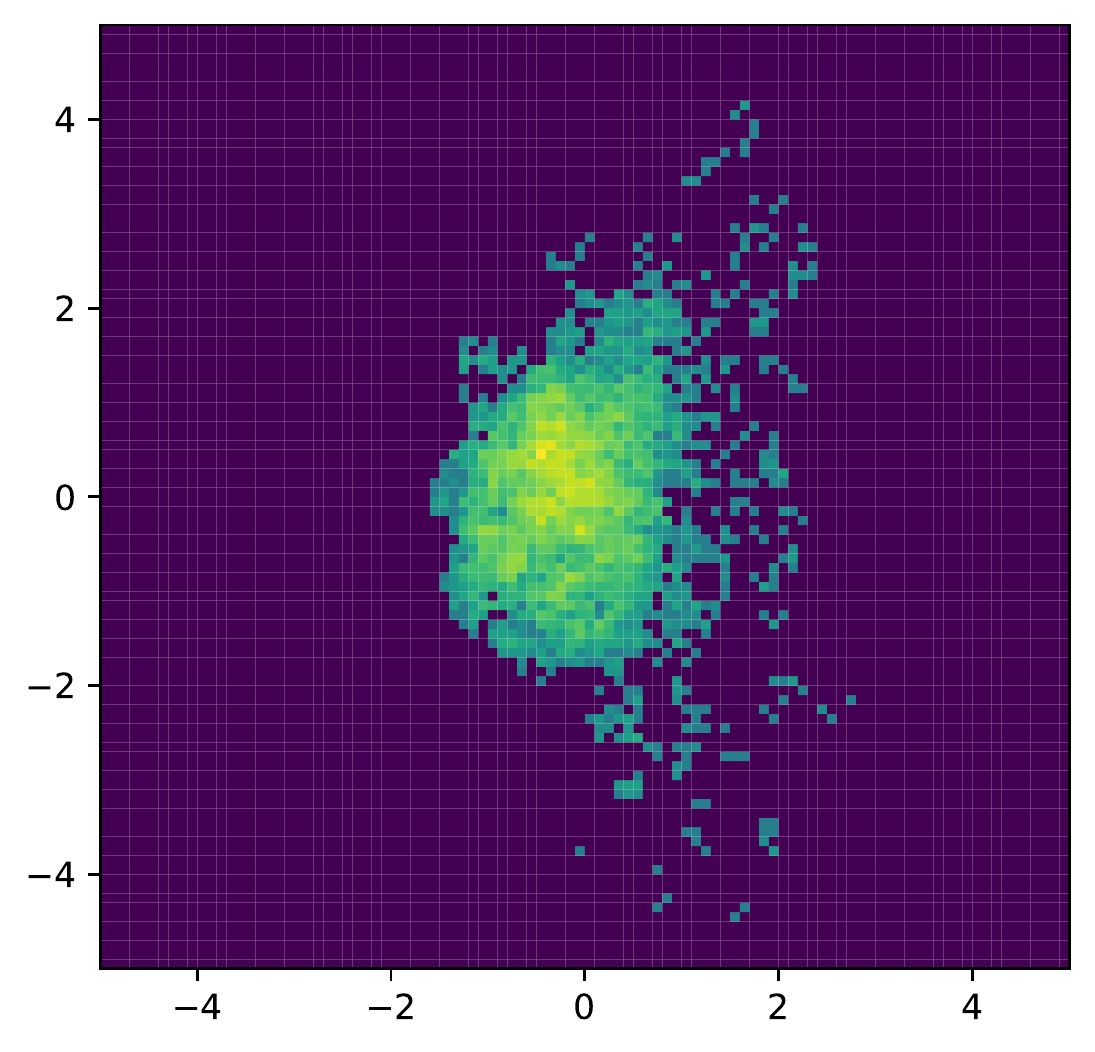}
         \hspace{1pt}   \includegraphics[width=.22\linewidth]{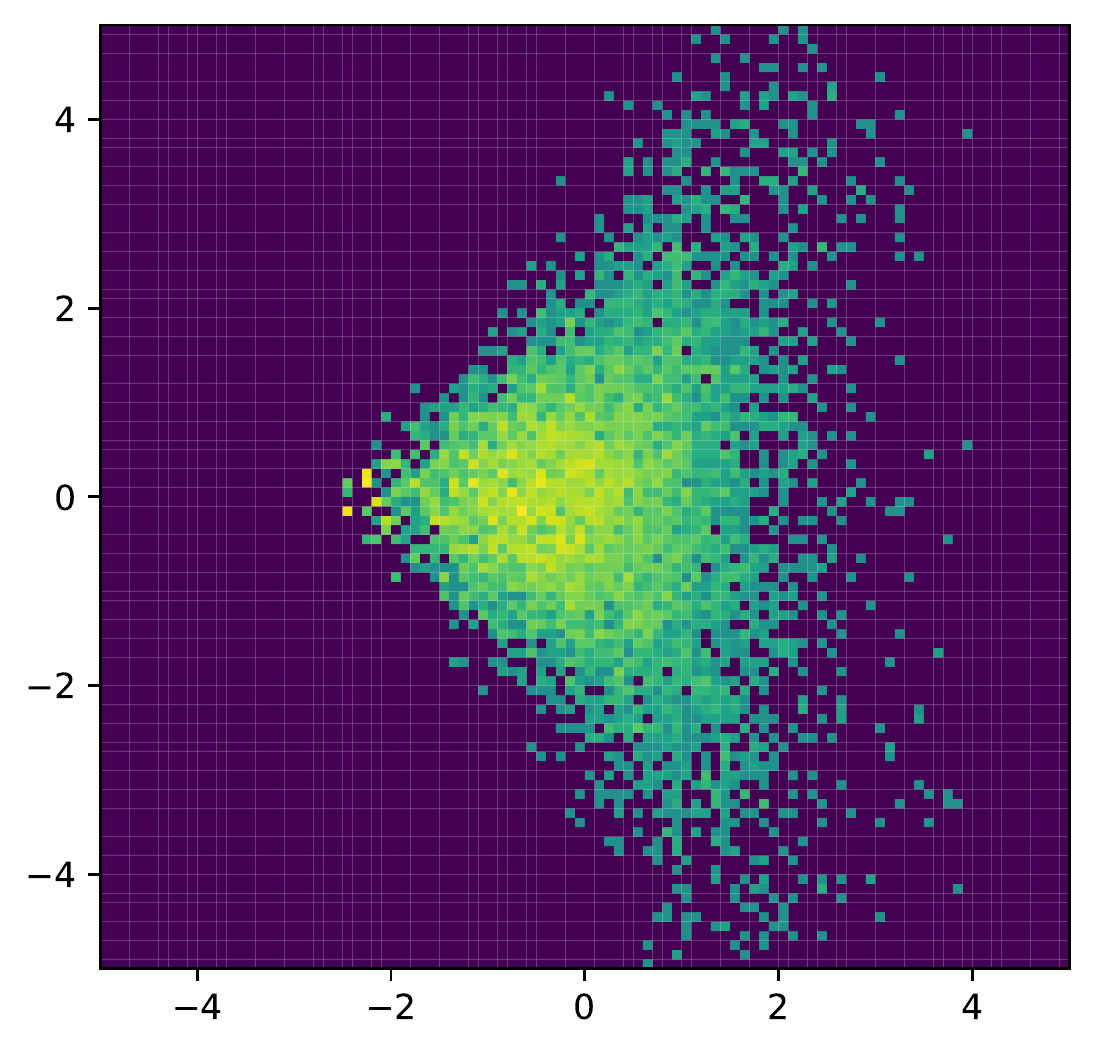}
             \hspace{1pt}
        \includegraphics[width=.22\linewidth]{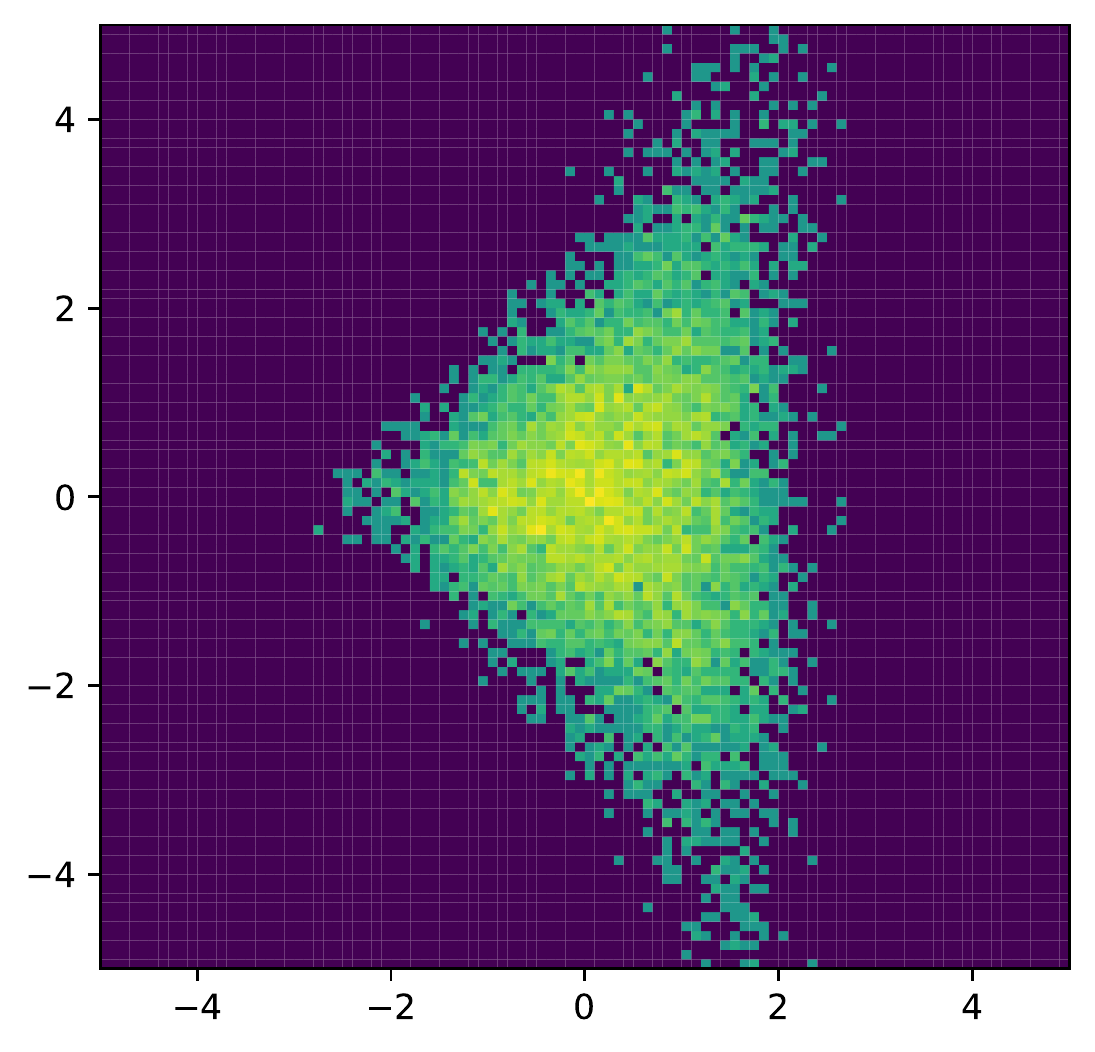}
     \caption{Empirical 2-D histogram of the samples of different algorithms targeting \texttt{MG25}  example (top) and  \texttt{Fun} (bottom).  From left to right: samples from the target distribution, correlated i-SIR, NUTS, \IFIS-MCMC.}
    \label{fig:samples_mcmc}
\end{figure}
\Cref{fig:samples_mcmc} displays the empirical two-dimensional histograms of the two first coordinates of samples from the ground truth, i-SIR, NUTS and \InFiNE-MCMC sampler. It is worthwhile to note that \NEO-MCMC algorithm performs much better for \texttt{MG25} which is a very challenging distribution, even for SOTA algorithm such as NUTS, which struggles to cross energy barriers between modes. For \texttt{Fun}, \NEO-MCMC performs favourably  \wrt\ NUTS, which is well adapted for this type of distributions.

\paragraph{Block Gibbs Inpainting with Deep Generative models and \IFIS-MCMC}
\label{subsec:gibbs_inpaint}
We apply \NEO-MCMC to the task of sampling the posterior of a deep latent variable model.
The model consists of a latent variable $x\sim \Normal(0, \Id_d)$ and a conditional distribution $p(\obs\mid x)$ which generates an image $\obs = (\obs^1,\dots, \obs^D) \in \rset^D$.  
Given a family of parametric \emph{decoders} $\{x\mapsto p_\theta(\obs\mid x)$, $\theta\in\Theta\}$, and a training set  $\mathcal{D}=\{\obs_i\}_{i=1}^{M}$, training involves finding the MLE $\theta^* = \argmax_{\theta \in \Theta} p_\theta(\mathcal{D})$.  As $p_\theta(\obs) =\int p_\theta(\obs\mid x)p(x)\rmd x$, the likelihood is intractable and to alleviate this problem, \cite{kingma:welling:2013} proposed to train jointly an approximate posterior $q_\phi(x|\obs)$ that maximizes a tractable lower-bound on the log-likelihood: $\ELBO(\obs,\theta,\phi) = \PE_{X \sim q_\phi( \cdot |\obs)}[\log p_\theta(\obs,X)/q_\phi(X|\obs)] \leq p_\theta(\obs)$, where $q_\phi(x\mid \obs)$ is a tractable conditional distribution with parameters $\phi \in \Phi$. It is assumed in the sequel that conditional to the latent variable $x$, the coordinates are independent, \ie\ $p_{\theta}(\obs \mid x)= \prod_{i=1}^{D} p_\theta(\obs^i | x)$.

It is possible to train VAE with the \NEO\ algorithm using the unbiased estimate of the normalizing constant to construct an ELBO. This approach is described in the supplement \Cref{sec:NEO-VAE}. It is assumed here that the VAE has been trained and we are only interested in the sampling problem. In our experiment, we use a VAE trained on CelebA dataset \footnote{See  \url{https://github.com/YannDubs/disentangling-vae/tree/master/results/betaH_celeba}} \cite{liu2018large}.  
We consider the Block Gibbs inpainting task introduced in \citep[Section~5.2.2]{levy:hoffman:sohl}. 
 We in-paint the bottom of an image using Block Gibbs sampling. Given an image $\obs$, denote by $[\obs^t, \obs^b] $ the top and the bottom half pixels. A two-stage Gibbs sampler amounts to (a) sampling $p_{\theta^*}(x | \obs^{t}, \obs^b)$ and (b) sampling $p_{\theta^*}(\obs^b|x, \obs^t)= p_{\theta^*}(\obs^b | x)$ (since $\obs^b$ and $\obs^t$ are independent conditional on $x$). 
 Starting from an image $\obs_0$, we sample at each step $x_k\sim p_{\theta^*}(\obs\mid \obs_k)$ and then $\tilde{\obs}_k\sim p_{\theta^*}(\obs\mid x_k$).   We then set $\obs_{k+1}=  ( {\obs}^t_*,\tilde{\obs}^b_{k+1})$. Stage (b) is elementary but stage (a) is challenging. We use the following decomposition  of $p_{\theta^*}(x\mid \obs) = q^\beta_{\phi^*}(x\mid \obs) p_{\theta^*}(x,\obs)/(q^\beta_{\phi^*}(x\mid \obs)p_{\theta^*}(\obs))$ with $\beta \in \ooint{0,1}$. We identify the \emph{proposal} $\proposal(x)\propto  q^\beta_{\phi^*}(x\mid \obs)$ ,  the \emph{likelihood} $\likelihood\leftarrow  p_{\theta^*}(x,\obs)/q^\beta_{\phi^*}(x\mid \obs)$ and the normalizing constant $\const = p_{\theta^*}(\obs)$ and apply i-SIR and \NEO-MCMC sampler for stage (a). 
 We compare  different algorithms in stage (a): i-SIR, HMC and \NEO-MCMC, with the same computational complexity ($N=10$, $K=12$, $\gamma = 0.2$ for \NEO-MCMC, $N=120$ for i-SIR, and HMC is run with $K=20$ leap-frog steps). For each algorithm,  10 steps are performed.
 \Cref{fig:gibbs_inpainting} displays the evolution of the resulting Markov chains. The samples  clearly illustrate that \NEO-MCMC\ mixes better than i-SIR and HMC. More are showcased in the supplementary.
 \begin{figure}
     \centering
     \includegraphics[width = .91\linewidth]{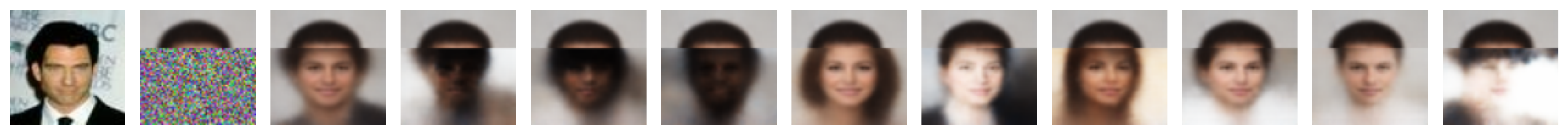}
\includegraphics[width = .91\linewidth]{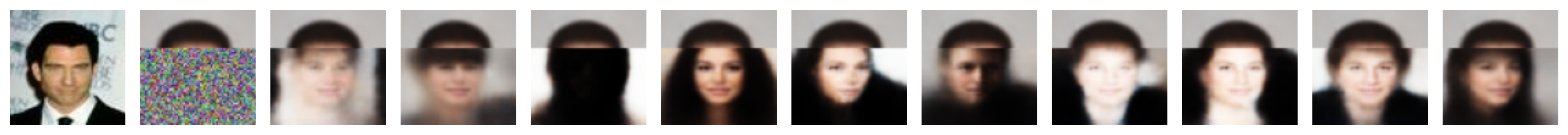}
\includegraphics[width = .91\linewidth]{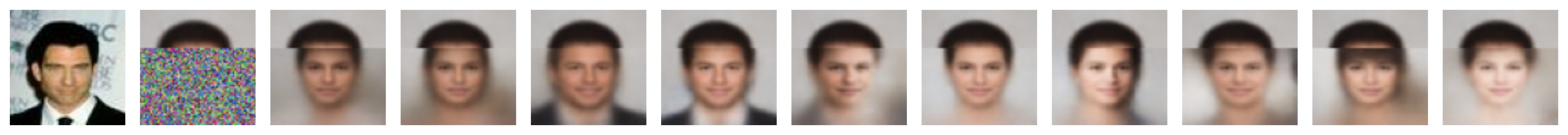}
\includegraphics[width = .91\linewidth]{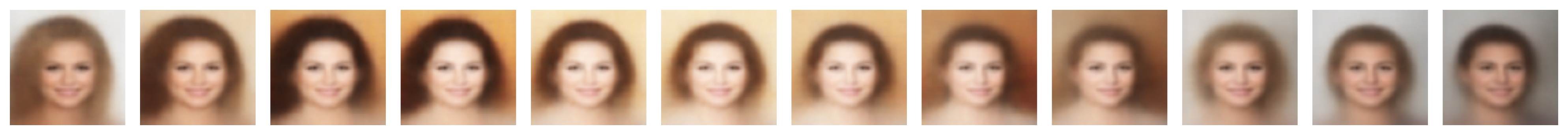}
     \caption{Gibbs inpainting for CelebA dataset. 
    From top to bottom: i-SIR, HMC and \NEO-MCMC: From left to right, original image, blurred image to reconstruct, and output every 5 iterations of the Markov chain. Last line: a forward orbit used in \NEO-MCMC.}
     \label{fig:gibbs_inpainting}
 \end{figure}
 
\section{Conclusion}
\label{sec:conclusion}
In this paper, we have proposed a new family of algorithms, \NEO, for computing normalizing constants and sampling from complex distributions. This methodology comes with asymptotic and non-asymptotic convergence guarantees. For normalizing constant estimation, \NEO-IS compares favorably to state-of-the-art algorithms on difficult benchmarks.  \NEO-MCMC is also very efficient for sampling  complex distributions: it is particularly well-adapted to sampling multimodal distributions, thanks to its proposal mechanism which avoids being trapped in local modes. There are numerous potential extensions to this work. For example, it would be interesting to consider deterministic transformations other than conformal Hamiltonian dynamics integrators. These transformations could be trained, as for Neural IS, using a variation lower bound. It would also be interesting to further investigate the influence of the mixture weights $\{\varpi_k\}_{k\in\zset}$ on the efficiency of \NEO.

\bibliographystyle{apalike}
\bibliography{bibliography}

\appendix

\section{Proofs}
\label{sec:proofs}
\subsection{Additional notation}

By abuse of notation, we denote by $\proposal$ and $\tpi$ the probability measures with density \wrt~the Lebesgue measure $\proposal$ and $\tpi$  respectively. 

\subsection{Proof of (\ref{eq:def_w_k})}
\label{app:proof:def_w_k}

The second expression of $w_k$ follows from $\JacOp{\transfo^{-j}}(\transfo^k(x)) = \JacOp{\transfo^{k-j}}(x)/\JacOp{\transfo^k}(x)$ which implies
\begin{multline}
    w_{k}(x) = 
     \varpi_k \rho(\transfo^{k}(x)) / \sum\nolimits_{j \in\zset}\varpi_j \rho(\transfo^{k-j}(x))\JacOp{\transfo^{-j}}(\transfo^k(x))\eqsp,\\
    \nonumber=  \varpi_k \rho(\transfo^{k}(x))\JacOp{\transfo^k}(x) / \sum\nolimits_{j\in\zset}\varpi_j \rho(\transfo^{k-j}(x))\JacOp{\transfo^{k-j}}(x) = \left.  \varpi_k \rho_{-k}(x) \middle / \left. \sum\nolimits_{i\in\zset}\varpi_{k+i} \rho_{i}(x) \right. \right. \eqsp.
\end{multline}

\subsection{Proof of \Cref{theo:clt_const_infine}}
\label{subsec:proof:clt_const_infine}
The unbiasedness of $\estConstC{\chunku{X}{1}{N}}$ follows directly from  \eqref{eq:key-relation}.
Moreover, as $\estConstC{\chunku{X}{1}{N}}$ is unbiased and $E^{\varpi} _{\transfo}<\infty$, we can write
\begin{equation}
    \Var_\rho[\estConstC{X}/\const] = \PE_\rho[(\estConstC{X}/\const)^2] - 1 = E^{\varpi} _{\transfo} -1\eqsp.
\end{equation}
As $\chunku{X}{1}{N}\simiid\rho$, $\Var_\rho[\estConstC{\chunku{X}{1}{N}}/\const] = N^{-1}  \Var_\rho[\estConstC{X}/\const]$. 
Finally, if $\bound < \infty$, then Hoeffding's inequality applies and we can write for any $\epsilon>0$,
\begin{equation}
\label{eq:hoeffding_constant_epsilon}
    \PP(|\estConstC{\chunku{X}{1}{N}}/\const-1|>\epsilon)\leq 2\exp( -2N\epsilon^2/(\bound)^2)\eqsp.
\end{equation}
Writing $\delta = 2\exp( -2N\epsilon^2/(\bound)^2)$, we identify
$\log(2/\delta)=2 {N\epsilon^2}/(\bound)^2$ and $\epsilon = \bound \sqrt{\log(2/\delta)/(2N)}$.
Plugging this expression of $\epsilon$ in \eqref{eq:hoeffding_constant_epsilon} concludes the proof.

\subsection{Proof of \Cref{theo:bias_mse_snis}}
\label{subsec:proo:bias_mse_snis}
We preface the proof of \Cref{theo:bias_mse_snis} with two auxiliary lemmas.
\begin{lemma}
  \label{lemma:estimator_ratio_bounds}
  Let $A,B$ be two integrable random variables satisfying $\abs{A/B} \leq M$ almost surely and denote $a=\PE[A]$, $b=\PE[B]$. Then,
 \begin{align}
     \left| \PE[A/B] - a/b\right|& \leq \frac{\sqrt{\Var(A/B) \Var(B)}}{b}\eqsp,\\
\mathrm{Var}(A/B) & \leq     \Exp{\left| A/B - a/b \right|^2} \leq \frac{2}{B^2}\left(\Exp{|A_N-A|^2} + M^2 \Exp{|B_N-B|^2}\right)\eqsp.
 \end{align}
\end{lemma}
 \begin{proof}
 Write first, using the Cauchy-Schwarz  inequality,
 \begin{align*}
\left| \Exp{\frac{A}{B}} - \frac{a}{b} \right| &= \left| \Exp{\frac{A}{B}} - \frac{\Exp{A}}{b} \right|
= \left| \Exp{A\left(\frac{1}{B} - \frac{1}{b}\right)} \right|\eqsp,\\
&= \left| \Exp{\frac{A}{B}\left(\frac{b - B}{b}\right)} \right|
= \left| \Exp{\left(\frac{A}{B} - \Exp{\frac{A}{B}}\right)\left(\frac{B - b}{b}\right)} \right|\eqsp,\\
&\leq \frac{\sqrt{\Var(A/B)}\sqrt{\Var(B)}}{b}\eqsp.
\end{align*}
Moreover, using  $|A/B| \leq M$ yields
\begin{align*}
    \left|\frac{A}{B} - \frac{a}{b}\right| = \left|\frac{1}{b}(A - a) + A\left(\frac{1}{B} - \frac{1}{b}\right)\right|
    &\leq\frac{1}{b}|A - a| + \frac{|A|}{B b}|B - b|\eqsp,\\
    &\leq \frac{1}{b}|A - a| + \frac{M}{b}|B - b|\eqsp.
\end{align*}
Therefore,
$$
\left|A/B - a/b\right|^2 \leq \frac{2}{b^2}\left(|A - a|^2 + M^2|B - b|^2\right)\eqsp,
$$
Using that $ \Exp{\left| A/B - a/b \right|^2} = \mathrm{Var}(A/B) + \left| \PE[A/B] - a/b \right|^2$ concludes the proof.
\end{proof}
We get the following lemma from \citep[Lemma 4]{douc2011sequential}.
\begin{lemma}
\label{lemma:hoeffding_ratio}
Assume that $A$ and $B$ are random variables and that there exist positive constants $b, M, C, K$ such that
\begin{enumerate}[label=(\roman*)]
    \item $|A / B| \leq M$, $\mathbb{P}$-a.s.\,,
    \item for all $\epsilon > 0$ and all $N \geq 1$, $\mathbb{P}\left( | B - b | > \epsilon \right) \leq K \exp(-R\epsilon^2)$\,,  \label{assump2}
    \item for all $\epsilon > 0$ and all $N \geq 1$, $\mathbb{P}\left(|A| > \epsilon\right) \leq K\exp\left(-R\epsilon^2/M^2\right)$\,, \label{assump3}
\end{enumerate}
then,
$$\mathbb{P}(|A/B| \geq \epsilon) \leq 2K \exp(-Rb^2\epsilon^2/4M^2)\eqsp.$$
\end{lemma}
\begin{proof}
By the triangle inequality,
\begin{align*}
    \left|A/B\right| &= \left| \frac{A}{B} (b - B) b^{-1}  + b^{-1}A\right|\eqsp,\\
    &\leq b^{-1}\left| A/B\right| \left|b - B\right| + b^{-1} \left|A\right| \leq Mb^{-1}\left|b - B\right| + b^{-1}\left|A\right|\eqsp.
\end{align*}
Therefore,
\begin{align*}
    \left\{ \left|A/B\right| \geq \epsilon \right\} \subseteq \left\{ \left|B - b\right| \geq \frac{\epsilon b}{2M}\right\} \cup \left\{ |A| \geq \frac{\epsilon b}{2}\right\}\eqsp.
\end{align*}
Then, conditions \ref{assump2} and \ref{assump3} imply that
\begin{align*}
    \mathbb{P}\left( |A/B| \geq \epsilon \right) & \leq \mathbb{P}\left(\left|B - b\right| \geq \frac{\epsilon b}{2M}\right) + \mathbb{P}\left(|A| \geq \frac{\epsilon b}{2}\right)\eqsp, \\
    &\leq 2K \exp(-R b^2\epsilon^2/(4M^2))\eqsp.
\end{align*}
\end{proof}

\begin{proof}[Proof of \Cref{theo:bias_mse_snis}]
  Let $g :\rset^d \to \rset$ such that  $\sup_{x \in \rset^d} \abs{g}(x) \leq 1$ and denote $\pi(g) = \int g \rmd \pi$.
We use \Cref{lemma:estimator_ratio_bounds} with $A = A_N$ and $B= \estConstC{\chunku{X}{1}{N}}$ where
\begin{equation}
    A_N = \frac{1}{N}\sum_{i=1}^N \sum_{k\in\zset}w_k(X^i) \likelihood(\transfo^k(X^i)) g(\transfo^k(X^i))\eqsp, \eqsp \estConstC{\chunku{X}{1}{N}} = \frac{1}{N}\sum_{i=1}^N \sum_{k\in\zset}w_k(X^i) \likelihood(\transfo^k(X^i))\eqsp.
\end{equation}
By construction, since  $\sup_{x \in \rset^d} \abs{g}(x) \leq 1$, almost surely $A_N/\estConstC{\chunku{X}{1}{N}} \leq 1$ and $\mathrm{Var}(\estConstC{\chunku{X}{1}{N}}) = N^{-1} \mathrm{Var}(\estConstC{X^1})$. Then, using  \eqref{eq:key-relation} with $a = \PE[A_N] = \const \pi(g)$ and $b = \PE[\estConstC{\chunku{X}{1}{N}}] = \const$,   \Cref{lemma:estimator_ratio_bounds} implies  
\begin{equation}
  \left|\infineSNIS(g) - \pi(g)\right|=
    \left|\PE[A_N/\estConstC{\chunku{X}{1}{N}}] - a/b\right| \leq N^{-1/2}\sqrt{\Var(A_N/\estConstC{\chunku{X}{1}{N}}) \mathrm{Var}(\estConstC{X^1})}\eqsp.
\end{equation}
On the other hand,
\begin{equation*}
  \txts\Exp{|A_N - a|^2}  = N^{-1} \PE_{X \sim \rho}\big[\big\{\sum_{k \in \mathbb{Z}} w_k(X) \likelihood(\transfo^k(X))g(\transfo^k(X)) - \const \pi(g)\big\}^2 \big]
 \leq N^{-1} \const^2 E^{\varpi}_{\transfo}\eqsp.
\end{equation*}
These inequalities yield using $\mathrm{Var}(\estConstC{X^1}) \leq E^{\varpi}_{\transfo}$ and \Cref{lemma:estimator_ratio_bounds} again:
 \begin{align*}
     &\Exp{|\infineSNIS(g) - \pi(g)|^2} \leq \frac{2}{N}(E^{\varpi}_{\transfo} + \mathrm{Var}(\estConstC{X^1})) \leq\frac{4}{N}E^{\varpi}_{\transfo}\eqsp,\\
     &|\Exp{\infineSNIS(g) - \pi(g)}|\leq \frac{\sqrt{2(E^{\varpi}_{\transfo} + \mathrm{Var}(\estConstC{X^1}))\mathrm{Var}(\estConstC{X^1})}}{N}\leq \frac{2 E^{\varpi}_{\transfo}}{N}\eqsp,
 \end{align*}
 which concludes the proof.

 Define
 \begin{equation*}
   \tilde{A}_N = N^{-1} \sum_{i=1}^N \sum_{k \in \zset} w_k(X^i) \likelihood(\transfo^k(X^i))\left(g(\transfo^k(X^i)) - \pi(g)\right)  \eqsp.
 \end{equation*}
 With this notation, the proof of \eqref{eq:exp_concentration} relies on the application of Lemma~\ref{lemma:hoeffding_ratio} to $A = \tilde{A}_N$ and $B = \estConstC{\chunku{X}{1}{N}}$, since
 $$ \infineSNIS(g) - \pi(g) = A_N / \estConstC{\chunku{X}{1}{N}}\eqsp.$$
As $\sup_{x \in \rset^d} \abs{g}(x) \leq 1$, we get that $\tilde{A}_N / \estConstC{\chunku{X}{1}{N}} \leq 2$.
By \eqref{eq:key-relation}, $\mathbb{E}[\estConstC{\chunku{X}{1}{N}}] =\const$ and $\estConstC{\chunku{X}{1}{N}} = N^{-1} \sum_{i=1}^N W_i $ with $W_i =  \sum_{k \in\zset} w_k(X^i) \likelihood(\transfo^k(X^i)) \leq \bound$. Then, by Hoeffding's inequality, for all $\varepsilon > 0$,
$$ \mathbb{P}(|B_N - \const| > \varepsilon) \leq 2\exp(-2N(\varepsilon/\bound)^2)\eqsp.$$
Similarly, $A_N$ is centered and $A_N= N^{-1} \sum_{i=1}^N U_i $ with $$U_i =  \sum_{k \in\zset} w_k(X^i) \likelihood(\transfo^k(X^i))\{g(\transfo^k(X^i)) -\pi(g)\}$$ and $\abs{U_i}  \leq 2\bound$ almost surely. By Hoeffding's inequality, for all $\varepsilon > 0$,
$$ \mathbb{P}(\left|A_N\right| > \epsilon) \leq 2\exp(-N\varepsilon^2/(8(\bound)^2))\eqsp.$$
The assumptions of \Cref{lemma:hoeffding_ratio} are met so that
$$ \mathbb{P}(|\infineSNIS(g) - \pi(g) | > \varepsilon) \leq 4\exp(-\varepsilon^2N \const^2/[32(\bound)^2])\eqsp,$$
which concludes the proof.
\end{proof}

\subsection{Proof of \Cref{lem:chi2_E_T}}
\label{sec:proof_lem:chi2_E_T}
As $w_k(x) = \varpi_k \rho(\transfo^k(x))/\{\Omega\rhoT(\transfo^k(x)) \}$, by Jensen's inequality,
\begin{align*}
    E_{\transfo}^{\varpi} = \int \left( \sum_{k\in\zset} w_k(x) \likelihood(\transfo^k(x))/\const\right)^2 \rho(x)\rmd x &= \int \left( \sum_{k\in\zset} \frac{\varpi_k}{\Omega} \frac{\pi(\transfo^k(x))}{\rhoT(\transfo^k(x))}\right)^2 \rho(x)\rmd x\eqsp,\\
    &\leq \int  \sum_{k\in\zset} \frac{\varpi_k}{\Omega} \left(\frac{\pi(\transfo^k(x))}{\rhoT(\transfo^k(x))}\right)^2 \rho(x)\rmd x\eqsp, \\
    &\leq  \Omega^{-1}\sum_{k\in\zset} \varpi_k \int  \left(\frac{\pi(\transfo^k(x))}{\rhoT(\transfo^k(x))}\right)^2 \rho(x)\rmd x\eqsp.
\end{align*}
Using the change of variables $y= \transfo^k(x)$ yields, by \eqref{eq:rhoT},
$$
    E_{\transfo}^{\varpi} \leq \Omega^{-1}\sum_{k\in\zset} \varpi_k \int  \left(\frac{\pi(y)}{\rhoT(y)}\right)^2 \rho(\transfo^{-k}(y))\JacOp{\transfo^{-k}}(y)\rmd y\leq \int\left(\frac{\pi(y)}{\rhoT(y)}\right)^2\rhoT(y) \rmd y\eqsp.
$$
\subsection{Proofs of \InFiNE\ MCMC sampler}
\label{sec:supp:proof_mcmc}

\begin{proof}[Proof of \Cref{lem:invariant_P}]
Note first that by symmetry, we have
\begin{equation}
\label{eq:symmetrized_P}
   P(y, \msa) = N^{-1}\int\sum_{i=1}^N \updelta_{y}(\rmd x^i) \prod_{j=1,\\j\neq i}^N\rho(x^j)\rmd x^j \sum_{k=1}^N\frac{\estConstC{x^k}}{\sum_{j=1}^N \estConstC{x^j}}\indi{\msa}(x^k)\eqsp.
\end{equation}
We begin with the proof of reversibility of $P$ \wrt\ $\tpi$.
Let $f,g$ be nonnegative measurable functions. By definition of $P$,
\begin{align*}
\int \tpi(\rmd y) P(y, \rmd y') f(y) g(y') &= \frac{1}{N\const}\int \sum_{i=1}^N \proposal(\rmd y) \estConstC{y} f(y) \updelta_y(\rmd x^i) \prod_{l=1, l\neq i}^N \proposal(\rmd x^l) \sum_{k=1}^N \frac{\estConstC{x^k}}{\sum_{j=1}^N\estConstC{x^j}} g(x^k)\eqsp,\\
&= \frac{1}{N\const}\int \sum_{i=1}^N  \estConstC{x^i} f(x^i)  \prod_{l=1}^N \proposal(\rmd x^l) \sum_{k=1}^N \frac{\estConstC{x^k}}{\sum_{j=1}^N\estConstC{x^j}} g(x^k)\eqsp,\\
&= \frac{1}{N\const}\int   \prod_{l=1}^N \proposal(\rmd x^l)  \frac{  \sum_{i=1}^N\estConstC{x^i} f(x^i) \sum_{k=1}^N\estConstC{x^k}g(x^k)}{\sum_{j=1}^N\estConstC{x^j}}\eqsp, \\
&= \int \tpi(\rmd y) P(y, \rmd y') f(y') g(y)\eqsp,
\end{align*}
which shows that $P$ is $\tpi$-reversible. We now establish that $P$ is $\tpi$-irreducible. 
  We have for $y \in \rset^d$, $\msa \in\mcbb(\rset^d)$,
  \begin{align*}
    P(y, \msa) &=\int \updelta_y(\rmd x^1)\sum_{i=1}^N \frac{\estConstC{x^i}}{N\estConstC{\chunku{x}{1}{N}}}\indi{\msa}(x^i) \prod_{j=2}^N\proposal(\rmd x^j) \\
               &=  \int \frac{\estConstC{y}}{\estConstC{y} + \sum_{j=2}^N \estConstC{x^j}}\indi{\msa}(x) \prod_{j=2}^N\proposal(\rmd x^j) + \int \sum_{i=2}^N \frac{\estConstC{x^i}}{\estConstC{y} + \sum_{j=2}^N \estConstC{x^j}}\indi{\msa}(x^i) \prod_{j=2}^N\proposal(\rmd x^j)\\
               &\geq \sum_{i=2}^N\int  \frac{\estConstC{x^i}}{\estConstC{y} +\estConstC{x^i}+ \sum_{j=2, j\neq i}^N \estConstC{x^j}}\indi{\msa}(x^i) \prod_{j=2}^N\proposal(\rmd x^j)\\
               &\geq\sum_{i=2}^N\int\tpi(\rmd x^i)\indi{\msa}(x^i)\int \frac{ \const}{\estConstC{y} +\estConstC{x^i} +\sum_{j=2, j\neq i}^N \estConstC{x^j}} \prod_{j=2, j\neq i}^N\proposal(\rmd x^j) \eqsp.
  \end{align*}
Since the function $f\colon z\mapsto (z+a)^{-1}$ is convex on $\rset_+$ for $a>0$, we get for $i\in\{2,\dots, N\}$,
  \begin{multline}
  \label{eq:convex_trick_i-SIR}
    \int \frac{\const}{\estConstC{y} + \estConstC{x^i} +\sum_{j=2, j\neq i}^N \estConstC{x^j}} \prod_{j=2, j\neq i}^N\proposal(\rmd x^j)
     \geq \frac{\const}{\estConstC{y} + \estConstC{x^i} +\int\sum_{j=2, j\neq i}^N \estConstC{x^j}\prod_{j=2, j\neq i}^N\proposal(\rmd x^j)}  \\
     \geq \frac{\const}{\estConstC{y} + \estConstC{x^i} + \const(N-2)}\eqsp.
  \end{multline}
  Therefore, for $\msa \in \mcbb(\rset^d)$ satisfying $\tpi(\msa)>0$, we get $P(y, \msa)>0$ for any $y \in \rset^d$ since $\estConstC{x} < \infty$ for any $x \in\rset^d$. By definition, $P$ is $\tpi$-irreducible.

  We show that $P$ is Harris recurrent using \citep[Corollary
  2]{tierney:1994}. To this end, since $P$ is $\tpi$-irreducible, it
  is sufficient to show that $P$ is a Metropolis type kernel.
  Define
  $\alpha(x^1, x^2) = (N-1) \int \prod_{j=3}^N \proposal(\rmd x^j)
  \estConstC{x^2}/\sum_{j=1}^N\estConstC{x^j}$ for
  $x^1,x^2 \in \rset^d$ and
  $\rho_{2:N}(\rmd \chunku{x}{2}{N}) =\{ \prod_{j=2}^N
  \rho_{2:N}(x^j)\} \rmd \chunku{x}{2}{N}$.  Then, by
  \eqref{eq:definition-P}, we get with this notation, for
  $y \in \rset^d$, $\msa \in\mcbb(\rset^d)$,
\begin{align}
    \nonumber
  &P(y, \msa) \\
    \nonumber
  &=  \int \updelta_{y}(\rmd x^1) \rho_{2:N}(\rmd \chunku{x}{2}{N}) \sum_{i=2}^N\frac{\estConstC{x^i}}{N\estConstC{\chunku{x}{1}{N}}}\indi{\msa}(x^i) + \int\updelta_{y}(\rmd x^1) \rho_{2:N}(\rmd \chunku{x}{2}{N})\frac{\estConstC{x^1}}{N\estConstC{\chunku{x}{1}{N}}}\indi{\msa}(x^1)
  \\
    \nonumber
  &=  \sum_{i=2}^N\int \updelta_{y}(\rmd x^1) \rho_{2:N}(\rmd \chunku{x}{2}{N})\frac{\estConstC{x^i}}{N\estConstC{\chunku{x}{1}{N}}}\indi{\msa}(x^i) + \int\updelta_{y}(\rmd x^1) \rho_{2:N}(\rmd \chunku{x}{2}{N})\frac{\estConstC{x^1}}{N\estConstC{\chunku{x}{1}{N}}}\indi{\msa}(x^1)\\
    \nonumber
  &=\sum_{i=2}^N \int \updelta_y(\rmd x^1)\proposal(\rmd x^i) \int \prod_{j=2, j\neq i}^N \proposal(x^j) \rmd x^j  \frac{\estConstC{x^i}\indi{\msa}(x^i)}{N\estConstC{\chunku{x}{1}{N}}} +\int\updelta_{y}(\rmd x^1) \rho_{2:N}(\rmd \chunku{x}{2}{N})\frac{\estConstC{x^1}\indi{\msa}(x^1)}{N\estConstC{\chunku{x}{1}{N}}}\\
  \nonumber
  &= \sum_{i=2}^N \int\frac{\alpha(y, x^i)}{(N-1)} \indi{\msa}(x^i)\proposal(\rmd x^i) +  \int\updelta_{y}(\rmd x^1) \rho_{2:N}(\rmd \chunku{x}{2}{N}) \left\{1- \sum_{i= 2} ^N \frac{\estConstC{x^i}}{N\estConstC{\chunku{x}{1}{N}}}\right\}\indi{\msa}(x^1) \\
  \label{eq:metropolis_type}
&=\int_\msa \alpha(y, y')\proposal(y')\rmd y' +  \left(1- \int\alpha(y,y')\proposal(y')\rmd y'\right) \updelta_y(\msa)\eqsp.
\end{align}
With the terminology of \citep[Corollary 2]{tierney:1994}, $P$ is
Metropolis type kernel and therefore is Harris recurrent.

Note that \Cref{algo:gibbs_partial} defines a Markov chain
$\{Y_i,U_i\}_{i \in\nset}$ taking for $U_0$ an arbitrary initial point
with Markov kernel denoted by $\tilde{P}$. By abuse of notation, we
denote by $\{Y_i,U_i\}_{i \in\nset}$ the canonical process on the
canonical space $(\rset^d \times \rset^d)^{\nset}$ endowed with the
corresponding $\sigma$-field and denote by $\PP_{y,u}$ the
distribution associated with the Markov chain with kernel $\tilde{P}$
and initial distribution $\updelta_y \otimes \updelta_u$.  Denote for
any $y \in\rset^d$ by $\PP_y$ the marginal distribution of $\PP_{y,u}$
with respect to $\{Y_i\}_{i \in\nset}$,
\ie~$\PP_y(\msa) = \PP_{(y,u)}(\{Y_i\}_{i\in\nset} \in \msa)$ for
$u \in \rset^d$, noting that by definition,
$\PP_{(y,u)}(\msa \times (\rset^d)^{\nset})$ does not depend on
$u$. In addition, under $\PP_y$, $\{Y_i\}_{i \in\nset}$ is a Markov chain associated with $P$. Therefore, since $P$ is $\tpi$-irreducible and Harris recurrent, we get by \citep[Theorem 11.3.1]{douc:moulines:priouret:2018} and \citep[Theorem 2, 3]{tierney:1994} 
for any $y \in \rset^d$, $\lim_{k \to \infty}\tvnorm{\updelta_y P^k - \tpi} = 0$ and for any bounded and measurable function $g$,
\begin{equation}
  \label{eq:lfgn_y}
  \text{$n^{-1} \sum_{k=1}^n g(Y_k) = \tpi(g)$} \eqsp, \qquad \qquad \text{$\PP_y$-almost surely} \eqsp.
\end{equation}

We now turn to proving the properties regarding $Q$. For any $\msb\in\borel(\rset^d)$, using  \eqref{eq:key-relation}, we obtain
\begin{equation*}
    \int_{}\tpi(y)Q(y, \msb)\rmd y
    = \const^{-1}\int_{} \proposal(y) \sum_{k\in\zset} {w_k(y)\likelihood(\transfo^k(y)) }\1_\msb(\transfo^k(y)) \rmd y = \pi(\msb)\eqsp.
  \end{equation*}
  Using for all $y\in \rset^d$,
  $\lim_{n\to\infty} \tvnorm{P^n(y, \cdot) - \tpi} = 0$, we get
  $\lim_{n\to\infty} \tvnorm{P^n Q(y, \cdot) - \pi} = 0$.  It remains
  to show the stated Law of Large Numbers.  Let $y,u\in\rset^d$ and $g$
  be a bounded measurable function.  Define for any $i\in\nsets$, 
  $\tU_i = g(U_i) - Qg(Y_i)$.  By definition, for any $i\in\nsets$, $\abs{\tU_i} \leq 2\sup_{x\in\rset^d}|g(x)|$ and $\PE_{(y,u)}[\tU_{i}|\mcf_{i-1}] = 0$, where $\{\mcf_k\}_{k\in\nset}$ is the canonical filtration. Therefore, $\{\tU_i\}_{i\in\nsets}$ are $\{\mcf_k\}_{k\in\nset}$-martingale increments and 
  $\{S_k =\sum_{i=1}^k \tU_i\}_{k \in\nset}$ is a $\{\mcf_k\}_{k\in\nset}$-martingale. Using 
  \citep[Theorem 2.18]{hallheydebook}, we get 
  \begin{equation}
      \lim_{n\to\infty}\{S_n/n\} = 0\eqsp, \quad\PP_{(y,u)} \text{-almost surely}\eqsp. 
    \end{equation}
The proof is completed using that $\lim_{n\to \infty}\{n^{-1} \sum_{i=1}^n Qg(Y_i)\} = \tpi(Qg) = \pi(g)$, $\PP_{y}$-almost surely by \eqref{eq:lfgn_y} and therefore by definition, $\PP_{(y,u)}$-almost surely.
\end{proof}
\begin{proof}[Proof of \Cref{theo:geom_ergodicity_infine}]
  We have for $(x, \msa)\in\rset^d\times\borel(\rset^d)$,
  \begin{align*}
    P(y, \msa)
    &\geq\sum_{i=2}^N\int\tpi(\rmd x^i)\indi{\msa}(x^i)\int \frac{ \const}{\estConstC{y} +\estConstC{x^i} +\sum_{j=2, j\neq i}^N \estConstC{x^j}} \prod_{j=2, j\neq i}^N\proposal(\rmd x^j) \eqsp.
  \end{align*}
   Moreover, as for any $x\in\rset^d$, $\estConstC{x}/\const\leq\bound$,
  \begin{equation*}
  \label{eq:convex_trick_i-SIR}
    \int \frac{\const}{\estConstC{y} + \estConstC{x^i} +\sum_{j=2, j\neq i}^N \estConstC{x^j}} \prod_{j=2, j\neq i}^N\proposal(\rmd x^j)
     \geq \frac{\const}{\estConstC{y} + \estConstC{x^i} + \const(N-2)}\geq \frac{1}{2\bound + N-2}\eqsp.
  \end{equation*}
  We finally obtain the  inequality
  \begin{equation}
    P(x, \msa)\geq \tpi(\msa)\times \frac{N-1}{2\bound + N -2} =  \epsilon_N \tpi(\msa)\eqsp.
  \end{equation}
  The proof for $P$ is concluded from \cite[Theorem~18.2.4]{douc:moulines:priouret:2018}.

  As $\tvnorm{P^{k}(y, \cdot) - \tpi}\leq \kappa_N^k$, for any bounded function $f$, $\|f\|_\infty\leq 1$, we have
  $|P^{k}f(y)- \tpi(f)|\leq \kappa_N^k$, by definition of the Total Variation Distance.
  Then, writing $f = Qg$ for any bounded function $g$, $\|g\|_\infty\leq 1$, we have $\|f\|_\infty\leq 1$ and
  \begin{equation}
      |P^{k}f(y)- \tpi(f)| =|P^{k}Qg(y)- \tpi Q(g)| = |P^{k}Qg(y)- \pi(g)|  \leq \kappa_N^k\eqsp.
  \end{equation}
  \end{proof}

  Write now $P$ the Markov kernel extending to correlated proposals:
  for $y \in\rset^d$ and $\msa \in \mathcal{B}(\rset^d)$,
\begin{equation}
     P(y, \msa) = N^{-1}\int\sum_{i=1}^N \updelta_{y}(\rmd x^i) r_i(x^i, \rmd \chunkum{x}{1}{n}{i}) \sum_{k=1}^N\frac{\estConstC{x^k}}{N\estConstC{\chunku{x}{1}{N}}}\indi{\msa}(x^k)\eqsp,
\end{equation}
where the Markov kernels $R_i$ are defined by $R_i(x^i, \rmd \chunkum{x}{1}{N}{i})= r_i(x^i, \chunkum{x}{1}{N}{i} )\rmd \chunkum{x}{1}{N}{i}$ and $r_i$ by \eqref{eq:defintion-r-i}.
\begin{theorem}
\label{theo:invariant_P_correlated}
$P$ is $\tpi$-invariant.
\end{theorem}
\begin{proof}
Define the $Nd$-dimensional probability measure $\bar{\proposal}_N(\rmd \chunku{x}{1}{N}) = \proposal(\rmd x^1) R_1(x^1, \rmd \chunku{x}{2}{n})$. Let $\msa\in\borel(\rset^d)$. Then, we have
\begin{align*}
\tpi P(\msa) &= N^{-1}\int \tpi(\rmd y) \int\sum_{i=1}^N \updelta_{y}(\rmd x^i) R_i(x^i, \rmd \chunkum{x}{1}{n}{i}) \sum_{k=1}^N\frac{\estConstC{x^k}}{N\estConstC{\chunku{x}{1}{N}}}\indi{\msa}(x^k)\\
&= (N\const)^{-1}\int\sum_{i=1}^N \proposal(\rmd x^i) \estConstC{x^i}R_i(x^i, \rmd \chunkum{x}{1}{n}{i}) \sum_{k=1}^N\frac{\estConstC{x^k}}{N\estConstC{\chunku{x}{1}{N}}}\indi{\msa}(x^k) \\
&= (N\const)^{-1}\int\bar{\proposal}_N(\rmd \chunku{x}{1}{N})\sum_{i=1}^N  \estConstC{x^i}\sum_{k=1}^N\frac{\estConstC{x^k}}{N\estConstC{\chunku{x}{1}{N}}}\indi{\msa}(x^k)
\\&= (N\const)^{-1}\int\sum_{k=1}^N\estConstC{x^k}\bar{\proposal}_N(\rmd \chunku{x}{1}{N})\indi{\msa}(x^k)
\\&= (N\const)^{-1}\int\sum_{k=1}^N\estConstC{x^k}\proposal(\rmd x^k)\indi{\msa}(x^k) = \tpi(\msa)\eqsp.
\end{align*}
\end{proof}

\section{Continuous-time limit of \NEO\ and NEIS}
\label{sec:continuous-time-limits}
\subsection{Proof for the continuous-time limit}
\label{subsec:continuous_time_proof}
Consider $\bar{h} >0$ and a family $\{\transfo_h \, : \, h \in \ocint{0,\bar{h}}\}$
of $\rmc^1$-diffeomorphisms. For $N \in \nsets$ and a bounded and continuous $f :\rset^d\to \rset$, write
\begin{equation}
  \label{eq:1}
  \infineISh(f) =  N^{-1} \sum_{i=1}^N \sum_{k\in\zset}w_{k,h}(X^i)  f(\transfo^k_h(X^i)) \eqsp,
\end{equation}
where $\{X_i\}_{i=1}^N \simiid \rho$ and for some weight function $\varpic : \rset \to \rset_+$ with bounded support (see \Cref{assum:continuitu_varpi}),
 $k \in \zset$ and $h >0$, setting $\varpi_{k,h} = \varpic(kh)$,
\begin{equation}
    \label{eq:def_w_k_h}
  { w_{k,h}(x)=\left.  \varpi_{k,h} \rho_{-k}(x) \middle / \left. \sum\nolimits_{i\in\zset}\varpi_{k+i,h} \rho_{i}(x) \right. \right. \eqsp.}
\end{equation}

We show in this section the convergence of the sequence of  \NEO-IS estimators $\{  \infineISh(f) \, : \, h \in \ocint{0,\bar{h}}\}$    as $h \downarrow 0$ to its 
continuous counterpart, the version \eqref{eq:infinecontinuous} of NEIS
\cite{rotskoff:vanden-eijden:2019}, with  weight
function $\varpi$,  in the case where for any $h \in \ocint{0,\bar{h}}$, 
$\transfo_h$  corresponds to one step of a 
discretization scheme with stepsize $h$ of the ODE
\begin{equation}
  \label{eq:ODE_b}
\dot{x}_t = b(x_t)   \eqsp,
\end{equation}
where $b :\rset^d \to \rset^d$ is a drift function. We are particularly interested in the case where \eqref{eq:ODE_b} corresponds to the conformal Hamilonian dynamics \eqref{eq:ODE_hamiltonian} and $\{\transfo_h \, : \, h \in \ocint{0,\bar{h}}\}$ to its conformal symplectic Euler discretization: for all $(q,p) \in \rset^{2d}$,
\begin{equation}
  \label{eq:def_SE_conformal_hamil}
  \transfo_h(q,p)=
 (q+h\mass^{-1}\{ \rme^{-h\gamma} p -h \nabla U(q)\},
 \rme^{-h \gamma } p -h \nabla U(q))\eqsp.
\end{equation}

 We make the following conditions on $b$, $\rho$, $\varpic$ and $\{\transfo_h \, : \, h \in \ocint{0,\bar{h}}\}$.
 \begin{assumption}
   \label{assum:drift_lip}
The function $b$ is continuously differentiable and $L_b$-Lipschitz.
\end{assumption}
Under \Cref{assum:drift_lip}, consider $(\phi_t)_{t \geq 0}$ the differential flow associated with \eqref{eq:ODE_b}, \ie~$\phi_t(x) = x_t$ where $(x_t)_{t \in \rset}$ is the solution of \eqref{eq:ODE_b} starting from $x$. Note that  \Cref{assum:drift_lip} implies that $(t,x) \mapsto \phi_t(x)$ is continuously differentiable on $\rset \times \rset^d$, see \cite[Theorem 4.1 Chapter V]{Har_1982}.

\Cref{assum:drift_lip} is satisfied in the case of the conformal Hamiltonian dynamics if the potential $U$ is continuously differentiable and with Lipschitz gradient, that is there exists $L_U\in\rsets_+$ such that for any $x_1, x_2\in\rset^d$, $\|\nabla U(x_1) - \nabla U(x_2)\|\leq L_U \|x_1 -x_2\|$.

  \begin{assumption}
   \label{assum:discretization_scheme}
For any $h \in\ocint{0,\bar{h}}$, $\transfo_h : \rset^d \to \rset^d$ is a  $\mathrm{C}^1$-diffeomorphism. In addition, it holds:
   \begin{enumerate}[label= (\roman*),wide]
       \item  \label{assum:discretization_scheme_i} there exist $C \geq 0$ and $\delta \in \ocint{0,1}$ such that for any $x\in\rset^d$,  $$\|\transfo_h(x) -(x+hb(x))\|\leq C h^{1+\delta} (1+\|x\|)\eqsp;$$
  \item \label{assum:jacobian_discretization} for any $x\in\rset^d$ and  $T\in\rsets_+$, $$ \lim_{h \downarrow 0} \max_{k\in[-\partint{T/h}: \partint{T/h}]}\|\JacOp{\phi_{kh}}(x)-\JacOp{\transfo_{h}^{k}}(x)\| = 0 \eqsp.$$
   \end{enumerate}
 \end{assumption}
 Note that \Cref{assum:discretization_scheme} is automatically satisfied for the conformal symplectic Euler discretization \eqref{eq:def_SE_conformal_hamil} of the conformal Hamiltonian dynamics.  Indeed, in that case $\divergence b(\phi_t(x)) = \gamma d$, and therefore $\JacOp{\phi_t}(x) = e^{\gamma d t}$ for $t \in \rset$, and for any $h>0, k\in\zset$, $\JacOp{\transfo_h^k}(x) = e^{\gamma d hk}$; see \cite{francca2019conformal}.  
 
 Define
 \begin{equation}
   \label{eq:def_support}
\support(\varpic) =    \{t \in \rset \, :\, \varpic(t) \neq 0\}\eqsp.
 \end{equation}
\begin{assumption}
  \label{assum:continuitu_varpi}
   \begin{enumerate}[label= (\roman*),wide]
       \item
$\proposal$ is continuous and positive on $\rset^d$
\item \label{assum:vapri_ratio}    $\varpic$ is piecewise continuous on $\rset$, its support $\support(\varpic)$ is bounded and $\sup_{(s,t)\in \msa_\varpi} \varpic(t)/\varpic(t+s)= m <\infty$ where 
  $$\msa_\varpi = \{(s,t)\in\rset^2;\eqsp t\in\operatorname{support}(\varpic), \eqsp(s+t)\in\operatorname{support}(\varpic)\}\eqsp.$$
\item Moreover, for any $x\in\rset^d$, we have $\rhoT^c(x) = \int\varpic(t)\proposal(\phi_t(x))\JacOp{\phi_t}(x)\rmd t >0$.
\end{enumerate}
\end{assumption}
Note that \Cref{assum:continuitu_varpi} implies that
$\sup_{t\in\rset}|\varpic(t)|<+\infty$.  \Cref{assum:continuitu_varpi}
is automatically satisfied for example in the case
$\varpic = \1_{\ccint{-T_1,T_2}}$ for $T_1,T_2 \geq 0$.

\begin{theorem}
\label{theo:convergence}
 Assume  \Cref{assum:drift_lip}, \Cref{assum:discretization_scheme},  
 \Cref{assum:continuitu_varpi}. 
 For any $x\in\rset^d$ and $f: \rset^d \to \rset$ continuous and bounded,

 \begin{equation*}
      \lim_{h\downarrow0} \left\vert \sum_{k\in\zset} w_{k,h}(x) f(\transfo^k_h(x)) - \int_{-\infty}^\infty \wcont(x)f(\phi_t(x))\rmd t  \right\vert=0 \eqsp,
    \end{equation*}
    where $\{w_{k,h}\}_{k \in\zset}$ and $\wcont$ are defined in  \eqref{eq:def_w_k_h} and \eqref{eq:continuous_limit_infine_weights} respectively, \ie~for $x \in\rset^d$ and $t \in \rset$,
    \begin{equation}
      \label{eq:2}
         \wcont( x) = \left.{\varpic(t)\rho(\phi_t(x)) \JacOp{\phi_t}(x)}\middle/{\int_{-\infty}^{\infty}\varpic({s+t})\rho(\phi_s(x)) \JacOp{\phi_s}(x)\rmd s}\right.\eqsp.
    \end{equation}
\end{theorem}
  \begin{proof}
    Let $f$ be a  bounded continuous function, $x\in\rset^d$.
Setting
 \begin{align*}
 & g_{k,h}(x) = \rho(\transfo_h^k(x))\varpic(kh)\JacOp{\transfo_h^k}(x)f(\transfo_h^k(x))\\
    & h\Delta_{k,h}(x) = h\sum_{i\in\zset}  \proposal(\transfo_h^i(x))\varpic((k+i)h) \JacOp{\transfo_h^i(x)}\eqsp,
 \end{align*}
we have that 
 \begin{align*}
   \sum_{k\geq 0} \frac{hg_{k,h}(x)}{h\Delta_{k,h}(x)}  =\int_0^{\Tvarpi} \frac{1}{h\Delta_{\partint{t/h},h}(x)}  g_{\partint{t/h}, h}(x)\rmd t +\int_{\Tvarpi}^{h\partint{\Tvarpi/h}+h}\frac{1}{h\Delta_{\partint{t/h},h}(x)}  g_{\partint{t/h}, h}(x)\rmd t\eqsp,
 \end{align*}
 as $g_{k,h}(x) = 0$ when $k>\partint{\Tvarpi/h}$.
Therefore,  we can consider the following decomposition,
   \begin{equation*}
       \left\vert \sum_{k\geq 0}\frac{\rho(\transfo_h^k(x))\varpic(kh)\JacOp{\transfo_h^k}(x)f(\transfo_h^k(x))}{\sum_{i\in\zset} \proposal(\transfo_h^i(x))\varpic((k+i)h) \JacOp{\transfo_h^i(x)}} - \int_0^{\Tvarpi}\frac{\varpic(t)\proposal(\phi_t(x))\JacOp{\phi_t}(x) f(\phi_t(x))\rmd t}{\int\varpic(t+s) \proposal(\phi_s(x))\JacOp{\phi_s}(x) \rmd s}\right\vert \leq A + B
   \end{equation*}
   with
   \begin{multline*}
       A  = \left\vert\int_0^{\Tvarpi} \frac{1}{h\Delta_{\partint{t/h},h}(x)}  \left\{ g_{\partint{t/h}, h}(x) - \varpic(t)\proposal(\phi_t(x))\JacOp{\phi_t}(x)f(\phi_t(x))\right\} \rmd t\right\vert\\
      + \left\vert \int_{\Tvarpi}^{h\partint{\Tvarpi/h}+h}\frac{1}{h\Delta_{\partint{t/h},h}(x)}   g_{\partint{t/h}, h}(x)\rmd t\right\vert\eqsp,
   \end{multline*}
 and 
    \begin{equation*}
    B =  \int_0^{\Tvarpi}\left\vert\frac{\varpic(t)\proposal(\phi_t(x))\JacOp{\phi_t}(x)f(\phi_t(x)) \rmd t}{h\Delta_{\partint{t/h},h}(x)} -  \frac{\varpic(t)\proposal(\phi_t(x))\JacOp{\phi_t}(x) f(\phi_t(x))}{\int\varpic(t+s) \proposal(\phi_s(x))\JacOp{\phi_s}(x) \rmd s}\right\vert \rmd t\eqsp,
    \end{equation*}
 We bound those terms separately. First of all, under \Cref{assum:continuitu_varpi}-\ref{assum:vapri_ratio}, for any $k$ such that $kh\in\ccint{0, \Tvarpi}$, we have  $h\Delta_{k,h}(x) \geq hm^{-1}\Delta_{0,h}(x)$. Second,  
 as $ \lim_{h\downarrow0}h\Delta_{0,h}(x)= \int_0^{\Tvarpi} \proposal(\phi_s(x))\JacOp{\phi_s}(x)\varpic(s)\rmd s>0$, there exists some $\tilde h>0$ and $c>0$ such that for all $k\in\zset$, $h<\tilde h$ implies 
 \begin{equation}
 \label{eq:majoration_delta}
\int_0^{\Tvarpi}\varpic(t)\proposal(\phi_t(x))\JacOp{\phi_t}(x)\rmd t>c\eqsp, \quad h\Delta_{k,h}(x) \geq hm^{-1}\Delta_{0,h}(x)>c\eqsp.
 \end{equation}
 Then, for $h<\tilde h$,
$$
 A \leq c^{-1} \int_{0}^{\Tvarpi} |g_{\partint{t/h}, h}(x) -  \varpic(t)\proposal(\phi_t(x))\JacOp{\phi_t}(x)f(\phi_t(x))|\rmd t +  c^{-1}\int_{\Tvarpi}^{h\partint{\Tvarpi/h}+h} \left\vert g_{\partint{t/h}, h}(x)\right\vert\rmd t\eqsp.
$$
By \Cref{assum:drift_lip} and \Cref{assum:continuitu_varpi}, the function $t\to \varpic(t)\proposal(\phi_t(x))\JacOp{\phi_t}(x)f(\phi_t(x))$ is continuous on the compact $[0, {2\Tvarpi}]$ and thus is bounded. Therefore, for any $h\in\ooint{0,\bar h}$,
\begin{equation}
\label{eq:bound_rho_phi_t}
\sup_{t\in\ccint{0, 2\Tvarpi}} \absLigne{\varpic(t)\proposal(\phi_{t}(x))\JacOp{\phi_{t}}(x)f(\phi_t(x))}
\leq\sup_{t\in\rset}|\varpic|\sup_{x\in\rset^d}|f(x)| \sup_{t\in\ccint{0, 2\Tvarpi}}\absLigne{\proposal(\phi_t(x))\JacOp{\phi_t}(x)}<\infty\eqsp.
\end{equation} 
Under \Cref{assum:discretization_scheme}, \eqref{eq:bound_rho_phi_t} and \Cref{lemma:control_discretization} imply that
$$
  \sup_{t\in\coint{0, h\partint{\Tvarpi/h}+h}} g_{\partint{t/h}, h}(x)\leq \sup_{t\in\rset}|\varpic(t)|\sup_{x\in\rset^d}|f(x)| \sup_{t\in\coint{0, h\partint{\Tvarpi/h}+h}}\rho(\transfo_h^{\partint{t/h}}(x))\JacOp{\transfo_h^{\partint{t/h}}}(x)<\infty\eqsp,
$$
Then, $\lim_{ h \downarrow 0}\int_{\Tvarpi}^{h\partint{\Tvarpi/h}+h} \left\vert g_{\partint{t/h}, h}(x)\right\vert\rmd t = 0$. 
Finally, \Cref{lem:riemann_1} implies  that $\lim_{h\downarrow0}A= 0$.  Moreover, setting
for $t\in\ccint{0, \Tvarpi}$,
   \begin{align}
   \label{eq:first_dominated}
   &\Delta^B_{t,h}(x)\\ 
   \nonumber =& \int \vert\proposal(\phi_{h\partint{s/h}}(x))\varpic(h(\partint{s/h} + \partint{t/h})) \JacOp{\phi_{h\partint{s/h}}(x)} - \varpic(s+t)\proposal(\phi_s(x))\JacOp{\phi_s}(x)) \vert\1_{\msa_\varpi}(s,t)\rmd s\\
  \nonumber &+ \int_{\Tvarpi-h\partint{t/h}}^{h(\partint{\Tvarpi/h}-\partint{t/h}+1)}\vert\proposal(\phi_{h\partint{s/h}}(x))\varpic(h(\partint{s/h} + \partint{t/h})) \JacOp{\phi_{h\partint{s/h}}(x)}\vert \1_{\msa_\varpi}(s,t) \rmd s \eqsp,
   \end{align}
   we have for  $h<\tilde h$, by \eqref{eq:majoration_delta} and \Cref{assum:continuitu_varpi}-\ref{assum:vapri_ratio},
  \begin{align}
\nonumber  B= &\int_0^{\Tvarpi} \left\vert\frac{\varpic(t)\proposal(\phi_t(x))\JacOp{\phi_t}(x) f(\phi_t(x))}{h\Delta_{\partint{t/h},h}(x)} - \frac{\varpic(t)\proposal(\phi_t(x))\JacOp{\phi_t}(x) f(\phi_t(x))}{\int\varpic(s+t)\proposal(\phi_s(x))\JacOp{\phi_s}(x) \rmd s}\right\vert\rmd t \\
 \nonumber &\leq \int_0^{\Tvarpi} \frac{\varpic(t)\proposal(\phi_t(x))\JacOp{\phi_t}(x) f(\phi_t(x))}{h\Delta_{\partint{t/h},h}(x)\int\varpic(s+t)\proposal(\phi_s(x))\JacOp{\phi_s}(x) \rmd s} \Delta^B_{t,h}(x)\rmd t\\
 \nonumber &\leq m c^{-2}\int_0^{\Tvarpi} \varpic(t)\proposal(\phi_t(x))\JacOp{\phi_t}(x) f(\phi_t(x))\Delta^B_{t,h}(x)\rmd t\\
 \label{eq:bound_b}
  &\leq m c^{-2}\sup_{t\in\rset}|\varpic(t)|\sup_{x\in\rset^d}|f(x)|\sup_{t\in[0, \Tvarpi]}\absLigne{\proposal(\phi_s(x))\JacOp{\phi_s}(x)}\int_0^{\Tvarpi}\Delta^B_{t,h}(x)\rmd t\eqsp.
  \end{align}
   By \Cref{assum:drift_lip} and \Cref{assum:continuitu_varpi}, the function $s\to\proposal(\phi_s(x))\JacOp{\phi_s}(x)$ is continuous on the interval $\ccint{-\Tvarpi, \Tvarpi}$ and thus is bounded. Therefore, for any $h\in\ooint{0,\bar h}$, 
\begin{multline}
\label{eq:bound_a_varpi}
\sup_{(s,t)\in\msa_\varpi} \absLigne{\varpic(h(\partint{t/h}+ \partint{s/h}))\proposal(\phi_{h\partint{s/h}}(x))\JacOp{\phi_{h\partint{s/h}}}(x)}\\
\leq \sup_{(s,t)\in\msa_\varpi}\absLigne{\varpic(s+t)\proposal(\phi_s(x))\JacOp{\phi_s}(x)}< \Tvarpi \sup_{s\in\rset}|\varpic(s)|\sup_{s\in[-\Tvarpi, \Tvarpi]}\absLigne{\proposal(\phi_s(x))\JacOp{\phi_s}(x)}<\infty\eqsp.
\end{multline} 
This implies that 
\begin{equation*}
    \lim_{h\downarrow0}\int_{\Tvarpi-h\partint{t/h}}^{h(\partint{\Tvarpi/h}-\partint{t/h}+1)}\vert\proposal(\phi_{h\partint{s/h}}(x))\varpic(h(\partint{s/h} + \partint{t/h})) \JacOp{\phi_{h\partint{s/h}}(x)}\vert\rmd s = 0\eqsp.
\end{equation*}
 Moreover, for any $t\in\ccint{0, \Tvarpi}$, the function $$s\mapsto \vert\varpic(h(\partint{t/h}+ \partint{s/h}))\proposal(\phi_{h\partint{s/h}}(x))\JacOp{\phi_{h\partint{s/h}}}(x)-\varpic(t+s)\proposal(\phi_s(x))\JacOp{\phi_s}(x)\vert\1_{\msa_\varpi}(s,t) $$ converges pointwise to $0$ for almost all $s\in\rset$ when $h\downarrow 0$   
 using \Cref{assum:drift_lip}, \Cref{assum:continuitu_varpi} and the continuity of $s\mapsto\phi_s(x)$. The Lebesgue dominated convergence theorem applies and by \eqref{eq:first_dominated}, for all $t\in\ccint{0,\Tvarpi}$,
 $$\lim_{h\downarrow0} \Delta^B_{t,h}(x)= 0\eqsp.$$
Moreover, using $h\Delta_{k,h}(x) = h\sum_{i\in\zset}  \proposal(\transfo_h^i(x))\varpic((k+i)h) \JacOp{\transfo_h^i(x)}$ and \eqref{eq:bound_a_varpi}, 
  \begin{equation*}
  \sup_{t\in\ccint{0, \Tvarpi}}\sup_{h\in\ooint{0,\bar h}}\Delta^B_{t,h}(x)<\infty \eqsp.
  \end{equation*} 
The Lebesgue dominated convergence theorem and \eqref{eq:bound_b} show that $\lim_{h\downarrow0}B = 0$ which
concludes the proof.
   \end{proof}

\subsubsection{Supporting Lemmas}
  For $f\in\rmc^1(\rset^d, \rset^d)$, define $\mathfrak{J}_f(x)$ the Jacobian matrix of $f$ evaluated at $x$ and the divergence operator by $\divergence f (x) = \trace[\mathfrak{J}_f(x)]$.
 \begin{lemma}
 \label{lem:jacobianflow}
 Let $b$ be a $\rmc^1$ vector field in $\rset^d$ and $(\phi_t)_{t\in\rset}$ be the flow  of the ODE \eqref{eq:ODE_b}. For any $t\in\rset$, the Jacobian of $\phi_t$ is given by
 $$\JacOp{\phi_t}(x) = \textstyle{\exp(\int_0^t\divergence b(\phi_s(x))\rmd s) }\eqsp.$$
 \end{lemma}
 \begin{proof}
 First, for $t \in \rset$ and $x\in\rset$, write $A(t, x)=\mathfrak{J}_{\phi_t}(x)$ the Jacobian matrix of $\phi_t$ evaluated at $x$.
By Jacobi's formula, $\dot{\det A}(t, x) = \trace[\adj(A(t, x)) \cdot \dot A(t, x)]$, where $\trace[M]$ denotes the trace of a matrix $M$ and $\adj(M)$ its adjugate, i.e. the transpose of the cofactor matrix of $M$ such that $\adj(M) M = \det(M) \Id$.
Since for all $t$ and $x$, $\dot A(t, x) = \mathfrak{J}_{b\circ \phi_t}(x) = \mathfrak{J}_{b} (\phi_t(x))\cdot A(t, x)$, then
\begin{equation}
    \dot\Jac_{\phi_t}(x)  
    = \trace[\adj(A(t, x)) \cdot \mathfrak{J}_{b} (\phi_t(x))\cdot A(t, x)] =  \trace[\mathfrak{J}_{b} (\phi_t(x))] \JacOp{\phi_t}(x)\eqsp.
\end{equation}
Integrating this ODE yields
    $\JacOp{\phi_t}(x) = \textstyle{\exp(\int_0^t\divergence b(\phi_s(x))\rmd s) }$.
 \end{proof}

\begin{lemma}
  \label{lem:gronwall}
  Assume \Cref{assum:drift_lip}.
  Then, there exists $C>0$ such that  for any $x\in\rset^d, t\in\rset$, $k\in\zset$, $h>0$,
  \begin{align*}
      &\|\phi_t(x)\|\leq C e^{C|t|}(\|x\|+1)\eqsp,
      \\&\|\transfo_h^k(x)\|\leq C e^{C|kh|}(\|x\|+1)\eqsp.
  \end{align*}
  \end{lemma}
  This lemma follows from Gronwall's inequality and \Cref{assum:drift_lip}.
 \begin{lemma}
  \label{lemma:control_discretization_1}
 Assume  \Cref{assum:drift_lip} and \Cref{assum:discretization_scheme}-\ref{assum:discretization_scheme_i}. There exists $C>0$ such that for any $x\in\rset^d, h\in\ooint{0, \bar{h}}$,
 \begin{equation}
 \label{eq:control_discretization_1}
      \|\transfo_h(x) -\phi_h(x)\| \leq C \{1+\|x\|\}\|h^{1+\delta}\eqsp.
 \end{equation}
 \end{lemma}
 \begin{proof}
  Under   \Cref{assum:drift_lip} and
 \Cref{assum:discretization_scheme}-\ref{assum:discretization_scheme_i}, we have
  \begin{equation*}
      \|\transfo_h(x) -\phi_h(x)\| \leq \| x+hb(x) - \phi_h(x)\| + C_Fh^{1+\delta}(1+\|x\|)\eqsp,
  \end{equation*}
  and as $\phi_h(x) = x + \int_0^hb(\phi_s(x))\rmd s$,
  \begin{multline} 
  \textstyle \| x+hb(x) - \phi_h(x)\| = \|hb(x) - \int_0^hb(\phi_s(x))\|\leq h L_b \sup_{s\in[0,h]}\|\phi_s(x) - x\|\\ \leq L_b h^2\{L_b \sup_{s\in[0,h]}\phi_s(x) + \|b(0)\|\}\eqsp.
  \end{multline}
The proof is completed using \Cref{lem:gronwall}.
  \end{proof}
   \begin{lemma}
  \label{lemma:control_discretization}
 Assume  \Cref{assum:drift_lip} and \Cref{assum:discretization_scheme}-\ref{assum:discretization_scheme_i}. There exists $C>0$ such that for any $x\in\rset^d, k\in\nset, h\in\ooint{0, \bar{h}}$, $kh\leq \Tvarpi$,
 \begin{equation}
 \label{eq:control_discretization}
     \|\transfo_h ^k(x) - \phi_{kh}(x)\| \leq C e^{khC} (1+\|x\|)h^{\delta}\eqsp.
 \end{equation}
 \end{lemma}
 \begin{proof}
  Using \Cref{lemma:control_discretization_1}, \Cref{assum:drift_lip} and \Cref{assum:discretization_scheme}-\ref{assum:discretization_scheme_i}, there exist $C_1, C_2, C_3>0$ such that for any $x\in\rset^d, k\in\nset, h\in\ooint{0, \bar{h}}$, $kh\leq \Tvarpi$,
  \begin{align*}
        &\|\transfo_h^{k+1}(x) -\phi_{(k+1)h}(x)\| \leq\|\transfo_h^{k+1}(x) -\transfo_h\circ\phi_{kh}(x)\| + \|\transfo_h\circ\phi_{kh}(x) -\phi_{(k+1)h}(x)\|\\
        &\leq (1+hL_b) \|\transfo_h^k(x) - \phi_{kh}(x)\| \\
        &\qquad\qquad\qquad\qquad+ h^{1+\delta}C_1\{2+\|\transfo_h^k(x)\| + \|\phi_{kh}(x)\|\} +\|\transfo_h\circ\phi_{kh}(x) -\phi_{(k+1)h}(x)\|
        \\
        &\leq (1+hL_b) \|\transfo_h^k(x) - \phi_{kh}(x)\| + h^{1+\delta}2C_1C_2e^{C_2\Tvarpi}\{1+\|x\|\} +C_3\{1+\|\phi_{kh}(x)\|\}h^{1+\delta}\\
        &\leq (1+hL_b) \|\transfo_h^k(x) - \phi_{kh}(x)\| \\
        &\qquad\qquad\qquad\qquad+ h^{1+\delta}2C_1C_2e^{C_2\Tvarpi}\{1+\|x\|\} +C_3\{1+C_2(1+\|x\|) \}h^{1+\delta}e^{C_2\Tvarpi}
        \\
        &\leq  (1+hL_b) \|\transfo_h^k(x) - \phi_{kh}(x)\| + A_T\{1+\|x\|\} h^{1+\delta}\eqsp,
      \end{align*}
      with $A_T = (2C_1C_2 + C_3(1 + C_2))e^{C_2\Tvarpi}$. 
      A straightforward induction yields
      \begin{equation*}
           \|\transfo_h ^k(x) - \phi_{kh}(x)\| \leq \frac{(1+hL_b)^{k}}{L_b}  A_T (1+\|x\|)h^{\delta}\eqsp.
      \end{equation*}
 \end{proof}
  \begin{lemma}
  \label{lem:riemann_1}
 Assume  \Cref{assum:drift_lip}, 
 \Cref{assum:discretization_scheme}, \Cref{assum:continuitu_varpi} . For any $x\in\rset^d$, and $f:\rset^d\to\rset^d$   bounded and continuous,
  \begin{equation*}
    \lim_{h\downarrow0}\int_{0}^{\Tvarpi}  \left\vert\varpic(h\partint{t/h}) \proposal(\transfo_h^{\partint{t/h}}(x))\JacOp{\transfo_h^{\partint{t/h}}}(x)f(\transfo_h^{\partint{t/h}}(x)) -  \varpic(t)\proposal(\phi_t(x))\JacOp{\phi_t}(x)f(\phi_t(x))\right\vert \rmd t= 0\eqsp.
  \end{equation*}
  \end{lemma}
  \begin{proof}
  Let $x\in\rset^d$. Consider the following decomposition, for any $h <\bar h$,    
   \begin{align*}
   &\int_{0}^{\Tvarpi}  \left\vert\varpic(h\partint{t/h}) \proposal(\transfo_h^{\partint{t/h}}(x))\JacOp{\transfo_h^{\partint{t/h}}}(x)f(\transfo_h^{\partint{t/h}}(x)) -  \varpic(t)\proposal(\phi_t(x))\JacOp{\phi_t}(x)f(\phi_t(x))\right\vert \rmd t  \\
 & \textstyle\leq \frac{h}{\Tvarpi}\sum_{k\in\zset} \varpic(kh) \absLigne{\proposal(\transfo_h^k(x))\JacOp{\transfo_h^k}(x)f(\transfo_h^k(x)) - \proposal(\phi_{kh}(x))\JacOp{\phi_{kh}}(x)f(\phi_{kh}(x)) } \\
   &+  \textstyle\int_{0}^{\Tvarpi} \absLigne{\varpic(t)\proposal(\phi_t(x))\JacOp{\phi_t}(x)f(\phi_t(x)) - \varpic(h\partint{t/h})\proposal(\phi_{h\partint{t/h}}(x))\JacOp{\phi_{h\partint{t/h}}}(x)f(\phi_{h\partint{t/h}}(x))} \rmd t\eqsp.
   \end{align*}
   The first term converges to 0 by \Cref{lemma:control_discretization} and \Cref{assum:discretization_scheme}-\ref{assum:jacobian_discretization} as $\varpic(kh) = 0$ for $kh>\Tvarpi$. 
   By \Cref{assum:drift_lip} and \Cref{assum:continuitu_varpi}, the function $t\to \varpic(t)\proposal(\phi_t(x))\JacOp{\phi_t}(x)f(\phi_t(x))$ is continuous on the compact $[0, {\Tvarpi}]$ and thus is bounded. Therefore, for any $h\in\ooint{0,\bar h}$,
\begin{multline}
\label{eq:bound_rho_phi_t_partint}
\sup_{t\in\ccint{0, \Tvarpi}} \absLigne{\varpic(h\partint{t/h})\proposal(\phi_{h\partint{t/h}}(x))\JacOp{\phi_{h\partint{t/h}}}(x)f(\phi_{h\partint{t/h}}(x))}\\
\leq\sup_{t\in\rset}|\varpic|\sup_{x\in\rset^d}|f(x)| \sup_{t\in\ccint{0, \Tvarpi}}\absLigne{\proposal(\phi_t(x))\JacOp{\phi_t}(x)}<\infty\eqsp.
\end{multline} 
 Moreover, $t\mapsto \varpic(h\partint{t/h})\proposal(\phi_{h\partint{t/h}}(x))\JacOp{\phi_{h\partint{t/h}}}(x)f(\phi_{h\partint{t/h}}(x))$ converges pointwise when $h\downarrow 0$ to $t\to \varpic(t)\proposal(\phi_t(x))\JacOp{\phi_t}(x)f(\phi_t(x))$  by continuity, using \Cref{assum:drift_lip} and \Cref{assum:continuitu_varpi}. The Lebesgue dominated convergence theorem applies and the second term goes to 0 as $h\downarrow 0$. 
 \end{proof}

\subsection{NEIS algorithm after \cite{rotskoff:vanden-eijden:2019}}
\label{sec:neis}
Non Equilibrium Importance Sampling (NEIS) has been introduced in the
pioneering work of \cite{rotskoff:vanden-eijden:2019}. NEIS relies on
the flow of the ODE $\dot{x}_t = b(x_t)$ and the introduction of a set
$\mso\subset\rset^d$. As in \Cref{sec:continuous-time-limits}, we assume \Cref{assum:drift_lip} holds and denote by  $(\phi_t)_{t \in\rset}$  the flow of this ODE.

Define for $x\in\mso$, the exit times
$\tau^+(x)\geq 0$ (resp. $\tau^-(x)\leq 0$) satisfying
\begin{equation}
 \label{eq:def_stopping_times}
    \tau^+(x) = \inf\{t\geq0 \, :\, \phi_t(x)\notin\mso\} \eqsp,\eqsp \tau^-(x) = \inf\{t\leq0 \, :\, \phi_t(x)\notin\mso\}\eqsp.
\end{equation}
The validity of NEIS relies on the following assumption.
\begin{assumption}
\label{assum:finite_constant}
The average time of an orbit in $\mso$ is finite, \ie\
\begin{equation}
    \label{eq:finite_const}
    \const_\tau = \int_\mso (\tau^+(x) - \tau^-(x))\proposal(x)\rmd x < \infty\eqsp.
\end{equation}
\end{assumption}
Under \Cref{assum:finite_constant}, we can define the proposal distribution
\begin{equation}
\label{eq:vde_proposal}
    \rhoT(x) = \const_\tau^{-1}\int_\mso \indi{[\tau^-(x),\tau^+(x)]}(t) \rho(\phi_t(x))\JacOp{\phi_t}(x)\rmd t\eqsp.
\end{equation}
Under \Cref{assum:finite_constant}, \citep[Equation (8)]{rotskoff:vanden-eijden:2019} derive the following estimator of $\rho(f)$, closely related to \eqref{eq:infinecontinuous}, in the case $\varpi\equiv 1$, on the restricted set $\mso\subset\rset^d$ :
\begin{align}
\label{eq:vde_estimator}
I^{\NEIS}_N(f) &= \frac{1}{N}\sum_ {i=1}^N \int_{\tau^-(X^i)}^{\tau^+(X^i)} w_t(X^i) f(\phi_t(X^i))\rmd t \\
w_t(x) &=   \frac{\rho(\phi_t(x))\JacOp{\phi_t(x)}}{\int_{\tau^-(x)}^{\tau^+(x)} \rho(\phi_t(x))\JacOp{\phi_t(x)} \rmd t}\eqsp.
\end{align}
Note that in practice, in order for \Cref{assum:finite_constant} to be verified, one typically requires that $\mso$ be bounded, as discussed in \cite{rotskoff:vanden-eijden:2019}.  

Following \cite{rotskoff:vanden-eijden:2019}, consider a $d$-dimensional system with position $q \in \rset^d$, momentum $p \in \rset^d$ and Hamiltonian $H(p,q)= (1/2) \| p \|^2 + U(q)$ where $U(q)$ is a potential assumed to be bounded from below. Denote by $V(E)$ the volume of the phase-space below some threshold energy $E$, 
\begin{equation}
\label{eq:volume-phase-space}
V(E)= \int   \indiacc{H(p,q) \leq E} \rmd p \rmd q \eqsp.
\end{equation}
To calculate \eqref{eq:volume-phase-space}, we set $x= (p,q)$, define $\mso= \{x; H(x) \leq E_{\max}\}$ for some $E_{\max} < \infty$, and use the dissipative Langevin dynamics with $b(x)=(p, - \nabla U(q) - \gamma p)$, \ie\
\[
\dot{q} = p \eqsp, \quad \dot{p} = - \nabla U(q) - \gamma p \eqsp,
\]
for some friction coefficient $\gamma > 0$. With this choice,  $\JacOp{\phi_t}(x) = \rme^{-d\gamma t}$. Taking $\proposal$ to be the uniform distribution on the (bounded) set $\mso$, write the estimator for $E\leq E_{\max}$, $V(E)/V(E_{\max})= \int \indiacc{H(p,q) \leq E} \proposal(p,q) \rmd p \rmd q$, where $\proposal(p,q)= \indi{\mso}(p,q)/V(E_{\max})$, we get
\begin{align}
    \nonumber V(E)/V(E_{\max}) &= \frac{1}{N}\sum_ {i=1}^N \frac{\int_{\tau^-(X^i)}^{\tau^+(X^i)} \JacOp{\phi_t(X^i)} \indiacc{H(\phi_t(X^i))\leq E}\rmd t}{\int_{\tau^-(X^i)}^{\tau^+(X^i)} \JacOp{\phi_t(X^i)} \rmd t} \\&= \frac{1}{N}\sum_ {i=1}^N \frac{\int_{\tau^E(X^i)}^{\tau^+(X^i)} \JacOp{\phi_t(X^i)} \rmd t}{\int_{\tau^-(X^i)}^{\tau^+(X^i)} \JacOp{\phi_t(X^i)} \rmd t}= \frac{1}{N}\sum_ {i=1}^N \rme^{-d\gamma(\tau^E(X^i) - \tau^-(X^i))}\eqsp,
\end{align}
where $\tau^E(x)$  denotes the (possibly infinite) time for a trajectory initiated at $x= (p,q)$ to reach the energy $E \leq E_{\max}$.

Finally, to estimate the normalizing constant, \cite{rotskoff:vanden-eijden:2019} discretize the energy levels $\{E_0,\dots, E_P\}$ and write their estimator as
\begin{equation}
   \widehat{\const}_{\chunku{X}{1}{N}}^{\NEIS} =  \frac{1}{N}\sum_ {i=1}^N \sum_{\ell = 1}^P \rme^{-d\gamma(\tau^E_\ell(X^i) - \tau^-(X^i))} (E_{\ell} - E_{\ell -1}) \eqsp,
\end{equation}
using an approximation of the identity
\[    \const=\int_\mso \int_{0}^{\infty}\indiacc{\likelihood(x)>L}\rho(x) \rmd L \rmd x =\int_0^\infty \PP_{X\sim\rho}(\likelihood(X)> L) \rmd L\eqsp,
\]
which is at the core of nested sampling \cite{chopin:robert:2010}.
\subsection{\InFiNE\ with exit times}
\label{sec:infine_stopping_times}
Consider $\mso \subset \rset^d$ and let $\transfo$ be a $\rmc^1$-diffeomorphism on $\rset^d$. We introduce here an estimator based on the forward and backward orbits in $\mso$ associated with $\transfo$. Define the exit times
$\tau^{+} : \rset^d \to \nset$ and $\tau^{-} : \rset^d \to \nset_-$, given, for
all $x \in \rset^d$, by
\begin{align}
\label{eq:definition-tau-+--}
&\tau^{+}(x)=\inf\{k\geq 1\, :  \,  \transfo^{k}(x) \not \in \mso\} \eqsp, \\
&\tau^{-}(x)=\sup\{k\leq -1\, :  \,  \transfo^{k}(x) \not \in \mso\} \eqsp,
\end{align}
with the convention $\inf \emptyset = +\infty$ and
$\sup \emptyset = - \infty$, and set
\begin{equation}
  \label{eq:def_rmi}
  \rmi = \{(x,k) \in \mso\times \zset\,:\, k \in
\intentier{\tau^-(x)+1}{\tau^+(x)-1}\} \eqsp.
\end{equation}
For any $k \in \zset$, define $\rho_k : \rset^d \to \rset_+$ by
\begin{equation}
\label{eq:definition-rho-k}
    \rho_k(x)= \rho(\transfo^{-k}(x))
    \JacOp{\transfo^{-k}}(x) \indi{\rmi}(x,-k)\eqsp.
\end{equation}
The density $\rho_k$ is the push-forward
of $\indi{\rmi}(x,k)\rho({x})$ by $\transfo^{k}$, \ie~for any $k \in \zset$ and any bounded function $g:\rset^d \to \rset$,
\begin{equation}
    \label{eq:inf_non_eq_av_0}
    \int_\mso g(y)    \rho_k(y)\rmd y =
  \int_\mso g(\transfo^{k}(x)) \indi{\rmi}(x,k)\rho(x)\rmd x  \eqsp.
\end{equation}
Consider the following assumption:
\begin{assumption}
  \label{assumption:z_ne_finite}
  The nonnegative sequence $(\varpi_k)_{k\in\zset}$ satisfies $\varpi_0 > 0$ and
\begin{equation}
\label{eq:def_z_ne}
    \constT^\varpi = \int_\mso\sum_{k\in \zset}  \varpi_k \rho_k(x) \rmd x = \int_\mso\sum_{k\in \zset}  \varpi_k \rho(\transfo^k(x))  {\JacOp{\transfo^k}(x)} \1_{\rmi}(x,k) \rmd x< \infty\eqsp.
  \end{equation}
 \end{assumption}
Consider the pdf
\begin{equation}\label{eq:rhoT_stopping_times}
    \rhoT(x) =  \frac{1}{\constT^{\varpi}} \sum_{k \in\zset}\varpi_k \rho_k(x)\eqsp,
  \end{equation}
where $\constT^{\varpi}$ is the normalizing constant.
This is a \textit{non-equilibrium} distribution, since $\rhoT$ is not invariant by $\transfo$ in general.
Using $\rhoT$ as an importance distribution to obtain an unbiased estimator of $\int \dummy(x) \proposal(x) \rmd x$ is feasible since as $\varpi_0>0$,  $\sup_{x \in \mso} \proposal(x)/\rhoT(x) \leq \constT/\varpi_0 < \infty$, hence
\[
\int_\mso \dummy(x) \rho(x)  \rmd x =\int_\mso \left(\dummy(x) \frac{\rho(x)}{\rhoT(x)}\right) \rhoT(x)  \rmd x\eqsp.
\]
From \eqref{eq:inf_non_eq_av_0}, the right hand side can be computed using the following key result.
\begin{theorem}
 \label{theo:inf_non_eq}
 For any $f:\rset^d \to \rset$ measurable bounded function, we have
\begin{equation}
\label{eq:key-relation-ft}
\int_{\mso} \dummy(x) \rho(x)  \rmd x =
\int_{\mso} \sum\nolimits_{k\in\zset}  \dummy(\transfo^{k}(x)) w_k(x) \rho(x)  \rmd x \eqsp,
\end{equation}
where, for any $x \in \rset^d$ and $k \in\zset$,
\begin{equation}
  \label{eq:def_w_k_stop}
    w_{k}(x) =  \left.  \varpi_k\rho_{-k}(x) \middle / \left. \sum\nolimits_{j\in\zset} \varpi_{j+k}\rho_{j}(x) \right. \right. \eqsp.
\end{equation}
\end{theorem}
\begin{proof}
Let $f:\rset^d\to\rset$ be a measurable bounded function. By \eqref{eq:inf_non_eq_av_0}, writing $g\leftarrow f\rho/\rhoT$,
\begin{multline*}
    \int_\mso \dummy(x) \rho(x)  \rmd x =\int_\mso \left(\dummy(x) \frac{\rho(x)}{\rhoT(x)}\right) \rhoT(x)  \rmd x \\
    = \int_\mso \sum_{k\in\zset}\left(\dummy(\transfo^k(x)) \frac{\varpi_k \rho(\transfo^k(x))\indi{\rmi}(x,k)}{\constT^\varpi \rhoT(\transfo^k(x))}\right) \rho(x)\rmd x\eqsp.
\end{multline*}
We now need to prove: 
\begin{multline*}\frac{\varpi_k \rho(\transfo^k(x))\indi{\rmi}(x,k)}{\constT^\varpi \rhoT(\transfo^k(x))}=  \frac{\varpi_k \rho(\transfo^k(x))\indi{\rmi}(x,k)}{\indi{\rmi}(x,k)\sum_{i\in\zset}\varpi_i\rho_i(\transfo^k(x))}=\frac{  \varpi_k\rho_{-k}(x)}{\left. \sum\nolimits_{j\in\zset} \varpi_{j+k}\rho_{j}(x) \right.}= w_k(x) \eqsp,
\end{multline*}
with the convention $0/0=0$.
We thus need to show that for any $x \in \mso$, $k\in\zset$,
\begin{align*}
  \indi{\rmi}(x,k)\sum_{i\in\zset}  \varpi_i\rho_i(\transfo^k(x))
 & =  \frac{\indi{\rmi}(x,k)}{\JacOp{\transfo^{k}}(x)} \sum_{j\in\zset} \varpi_{j+k} \rho_j(x)  \eqsp.
\end{align*}
Using the identity $\JacOp{\transfo^{-i+k}}(x)=\JacOp{\transfo^{-i}}(\transfo^k(x)) \JacOp{\transfo^{k}}(x)$, we obtain
\begin{align*}
    \indi{\rmi}(x,k)\sum_{i\in\zset} \varpi_i \rho_i(\transfo^k(x)) &=   \sum_{i\in\zset}  \indi{\rmi}(x,k)   \varpi_i\rho(\transfo^{-i}(\transfo^k(x))) {\JacOp{\transfo^{-i}}(\transfo^k(x))} \1_{\rmi}(\transfo^k(x),-i) \\
 &= \frac{1}{\JacOp{\transfo^{k}}(x)}\sum_{i\in\zset}  \indi{\rmi}(x,k) \varpi_i\rho(\transfo^{-i+k}(x)) {\JacOp{\transfo^{-i+k}}(x)} \1_{\rmi}(\transfo^k(x),-i) \\
 &=  \frac{1}{\JacOp{\transfo^{k}}(x)}\sum_{j \in \zset}  \varpi_{j+k}\rho(\transfo^{-j}(x)) {\JacOp{\transfo^{-j}}(x)} \1_{\rmi}(\transfo^k(x),-j-k)\indi{\rmi}(x,k)
\end{align*}
Note that if $(x,k) \in \rmi$, we have $(x,-j)\in \rmi$ if and only if
$(\transfo^k(x),-j-k) \in \rmi$ by definition of $\rmi$ \eqref{eq:def_rmi}. The proof is concluded by noting that:
\begin{equation*}
    \1_{\rmi}(\transfo^k(x)),-j-k)\indi{\rmi}(x,k) =\indi{\rmi}(x,-j)\indi{\rmi}(x,k) \eqsp.
\end{equation*}
\end{proof}

\section{Iterated SIR}
\label{app:i-SIR}
Let us recall the principle of the Sampling Importance Resampling method (SIR; \citet{rubin1987comment,smith1992bayesian}) whose goal is to approximately sample from the target distribution $\target$ using samples drawn from a proposal distribution $\proposal$.

In SIR, a $N$-\iid\ sample $\chunku{X}{1}{N}$ is first generated from the proposal distribution $\proposal$. A sample $X^*$ is approximately drawn from the target $\target$ by choosing randomly a value in $\chunku{X}{1}{N}$ with probabilities proportional to the importance weights $\{\likelihood(X^i)\}_{i = 1}^N$, where $\likelihood(x)= \target(x)/\proposal(x)$. Note that the importance weights are required to be known only up to a constant factor.

For SIR, as $N \to \infty$, the sample $X^*$ is \emph{asymptotically} distributed according to $\target$; see~\cite{smith1992bayesian}. 

A subsequent algorithm  is the \emph{iterated SIR} (i-SIR) \citep{andrieu2010particle}. Here, $N$ is not necessarily large ($N\geq 2$), the whole process of sampling a set of proposals, computing the importance weights, and picking a  candidate, is iterated. At the $n$-th step of i-SIR, the active set of $N$ proposals $\chunku{X_n}{1}{N}$ and the index $I_n \in [N]$ of the conditioning proposal are kept. First i-SIR  updates the active set  by setting $X_{n+1}^{I_n}= X_n^{I_n}$ (keep the conditioning proposal) and then draw independently $\chunkum{X_{n+1}}{1}{N}{I_n}$ from $\proposal$.
Then it selects the next proposal index $I_{n+1} \in [N]$ by sampling with probability
proportional to $\{\weightfunc(X_{n+1}^i)\}_{i=1}^N$.
As shown in \cite{andrieu2010particle}, this algorithm  defines  a partially collapsed Gibbs sampler (PCG) of the augmented distribution 
$$
\bar{\measpi}(\chunku{x}{1}{N},i)=\frac{1}{N}  \target(x^i) \prod_{j \neq i} \proposal(x^j) =\frac{1}{N} \weightfunc(x^i) \prod_{j=1}^N \proposal(x^j) \eqsp.
$$
The PCG sampler can be shown to be ergodic provided that $\proposal$ and $\target$ are continuous and $\proposal$  is  positive on the support of $\target$. If in addition the importance weights are bounded, the Gibbs sampler can be shown to be uniformly geometrically ergodic  \citep{lindsten2015uniform,andrieu2018uniform}.
It follows that the distribution of the conditioning proposal $X_n^*= X_n^{I_n}$ converges to $\pi$ as the iteration index $n$ goes to infinity. Indeed, for any integrable function $f$ on $\rset^d$, with $(\chunk{X}{1}{N},I) \sim \bar{\measpi}$,
$$
\PE_{}[f(X^I)]= \int \sum_{i=1}^N f(x^i)  \bar{\measpi}(\chunku{x}{1}{N},i) \rmd \chunku{x}{1}{N} = N^{-1} \sum_{i=1}^N \int f(x^i) \target(x^i) \rmd x_i = \int f(x) \target(x) \rmd x \eqsp.
$$
When the state space dimension $d$ increases, designing a proposal distribution $\proposal$ guaranteeing proper mixing properties becomes more and more difficult. A way to circumvent this problem is to use dependent proposals, allowing in particular \emph{local moves} around the conditioning orbit. To implement this idea, for each $i \in [N]$, we define a proposal transition, $r_i(x^i; \chunkum{x}{1}{N}{i})$ which defines the the conditional distribution of $\chunkum{X}{1}{N}{i}$ given  $X^i= x^i$. The key property validating i-SIR with dependent proposals  is that all one-dimensional marginal distributions are equal to $\proposal$, which requires that for each $i,j  \in [N]$,
\begin{equation}
\label{eq:conditional-decomposition}
\proposal(x^i) r_i(x^i;\chunkum{x}{1}{N}{i})=
\proposal(x^j) r_j(x^{j};\chunkum{x}{1}{N}{j})\eqsp.
\end{equation}
The (unconditional) joint distribution of the particles is therefore defined as
\begin{equation}
\label{eq:joint-distribution}
\proposal_N\bigl(\chunku{x}{1}{N}\bigr) = \proposal(x^1) r_1(x^1;\chunkum{x}{1}{N}{1}) \eqsp.
\end{equation}
The resulting modification of the i-SIR algorithm is straightforward: $\chunkum{X}{1}{N}{I_n}$ is sampled jointly from the conditional distribution $r_{I_n}(X_n^{I_n},\cdot)$ rather than independently from $\proposal$.

\section{Additional Experiments}
\subsection{Normalizing constant estimation}
\label{sup:sec:additional_xp}
We consider here the problem of the estimation of the normalizing constant of Cauchy mixtures. The Cauchy distribution with scale $\sigma$ has a pdf defined by $\operatorname{Cauchy}(x; \mu, \sigma) = [\pi\sigma(1+ \{(x-\mu)/\sigma\}^2]^{-1}$. The target distribution is a product of mixtures of two Cauchy distributions, \begin{equation*}
\label{eq:cauchymixture}
\target(x)= \prod_{i=1}^{n} \frac{1}{2}\left[\operatorname{Cauchy}\left(x_{i} ; \mu, \sigma\right)+\operatorname{Cauchy}\left(x_{i} ; -\mu, \sigma\right)\right]
,\quad \mu = 5, \sigma = 1\eqsp.
\end{equation*}
\IFIS-IS is compared with IS estimator using the same proposal $\proposal$. We also compare \NEO-IS to Neural IS \cite{muller2018neural} with a Cauchy as base distribution.
\begin{figure}[!ht]
    \centering
    \includegraphics[width=0.4 \linewidth]{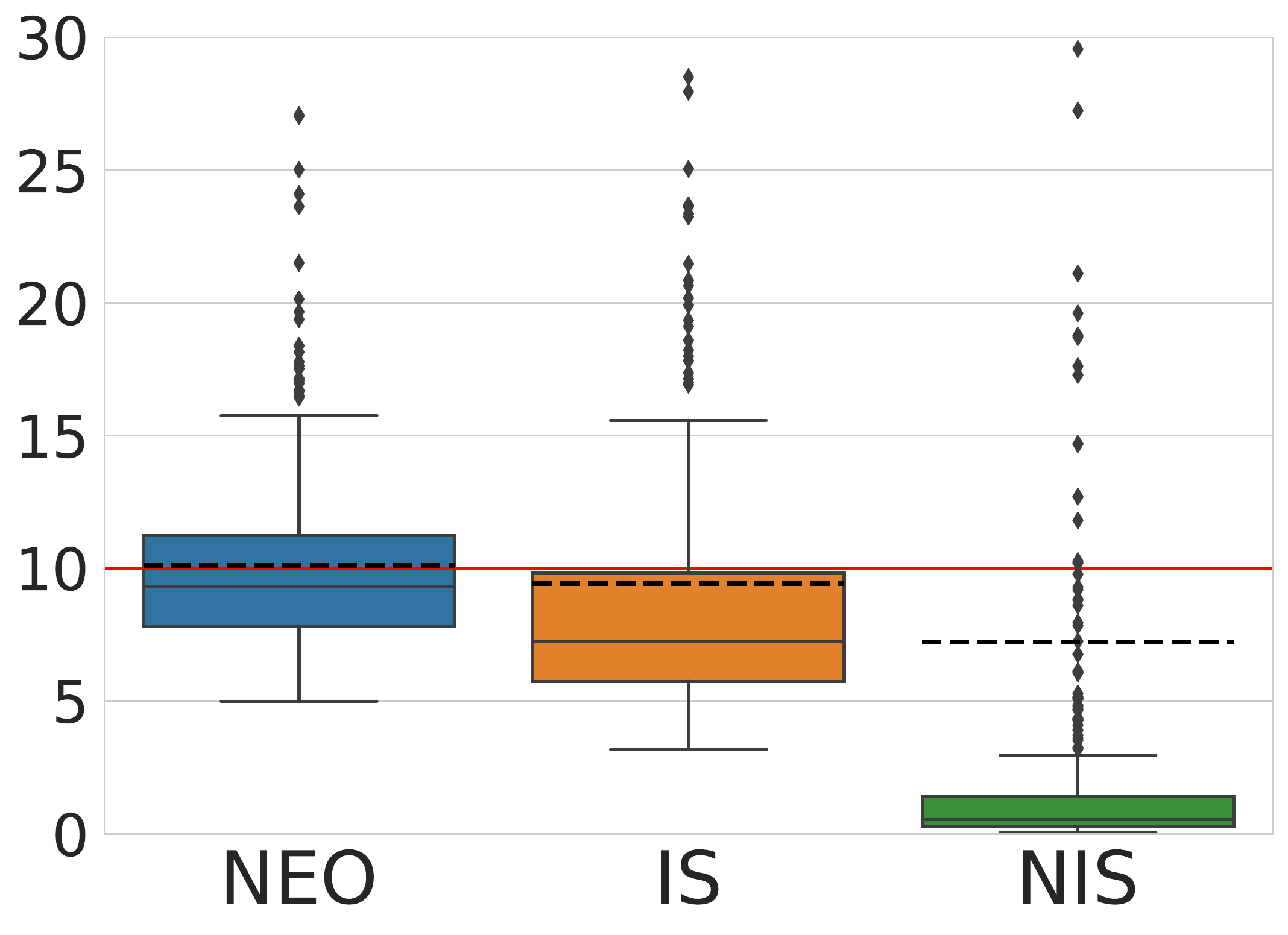}
    \includegraphics[width=0.4 \linewidth]{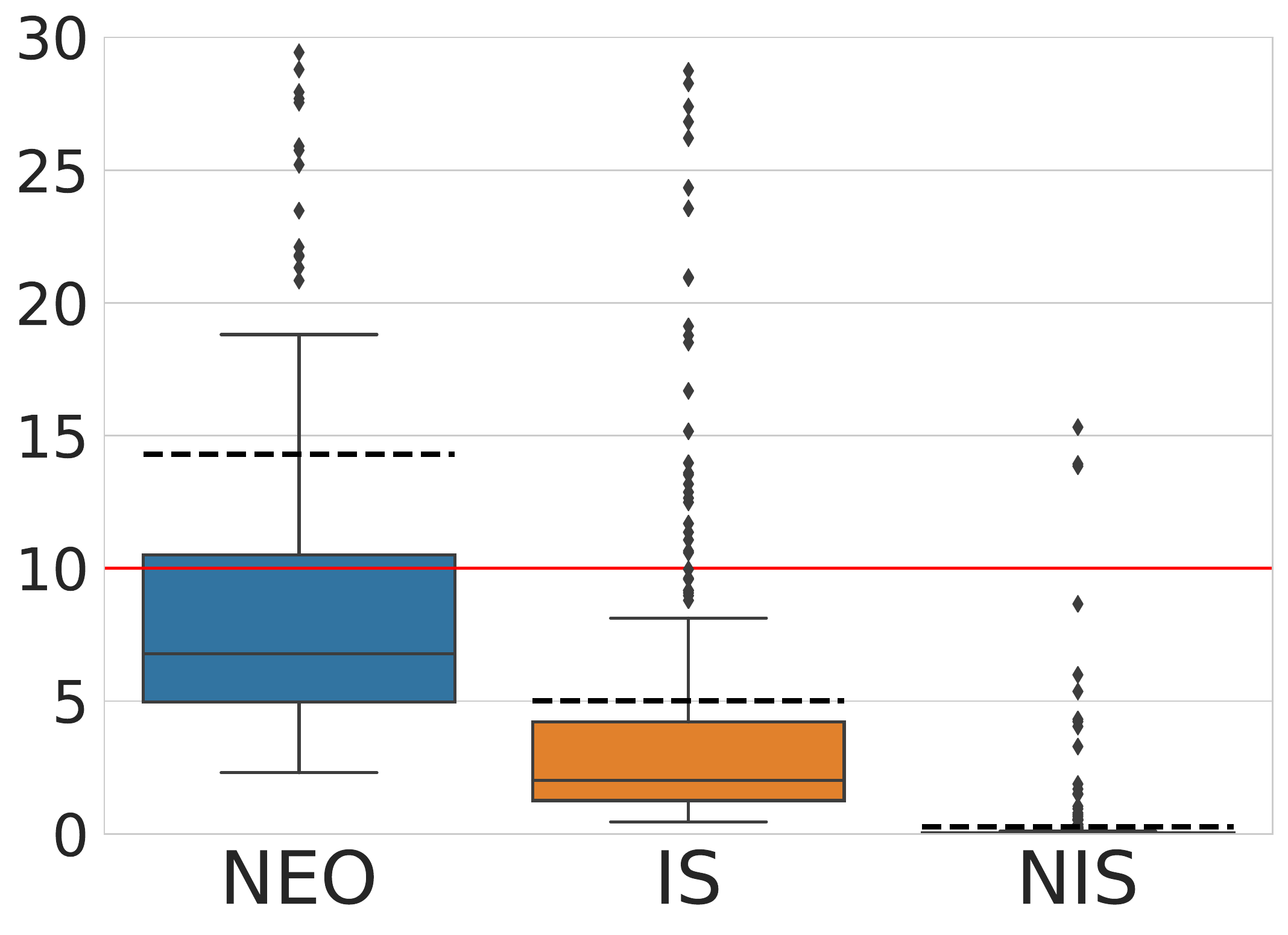}
    \caption{Boxplots of 500 independent estimations of the normalizing constant of the Cauchy mixture in dimension $d={10,15}$ (top, bottom). The true value is given by the red line. The figure displays the median (solid lines), the interquartile range, and the mean (dashed lines) over the 500 runs}
    \label{fig:simple_gauss}
\end{figure}

Finally, we compare \NEO-IS with \NEIS\footnote{The code from
  \cite{rotskoff:vanden-eijden:2019} we run is available at
  \url{https://gitlab.com/rotskoff/trajectory_estimators/-/tree/master}.}. We
consider here \texttt{MG25} in dimension 5 and 10, where all the
covariances of the Gaussian distributions are diagonal and equal to
$0.005\Id$. \NEIS\ and \NEO-IS are run for the same computational
time. We add an IS scheme as a baseline for comparison. All algorithms (\NEO-IS, \NEIS, IS) are run for  7.20s and 11.30s wall clock time respectively for $d = 5$ and $d = 10$. For \NEO-IS,
we use a conformal transform with $h=0.1$, $K=10$ and $\gamma =
1$. For \NEIS, we choose $\gamma=1$ and consider a stepsize
$h= 10^{-4}$ corresponding to an optimal trade-off between the
discretization bias inherent to \NEIS and its computational budget. We
can observe that \NEO-IS always outperforms \NEIS, which suffers from a non-negligeable bias if the stepsize $h$ is not chosen small enough. 
\begin{figure}[!ht]
    \centering
    \includegraphics[width=0.4 \linewidth]{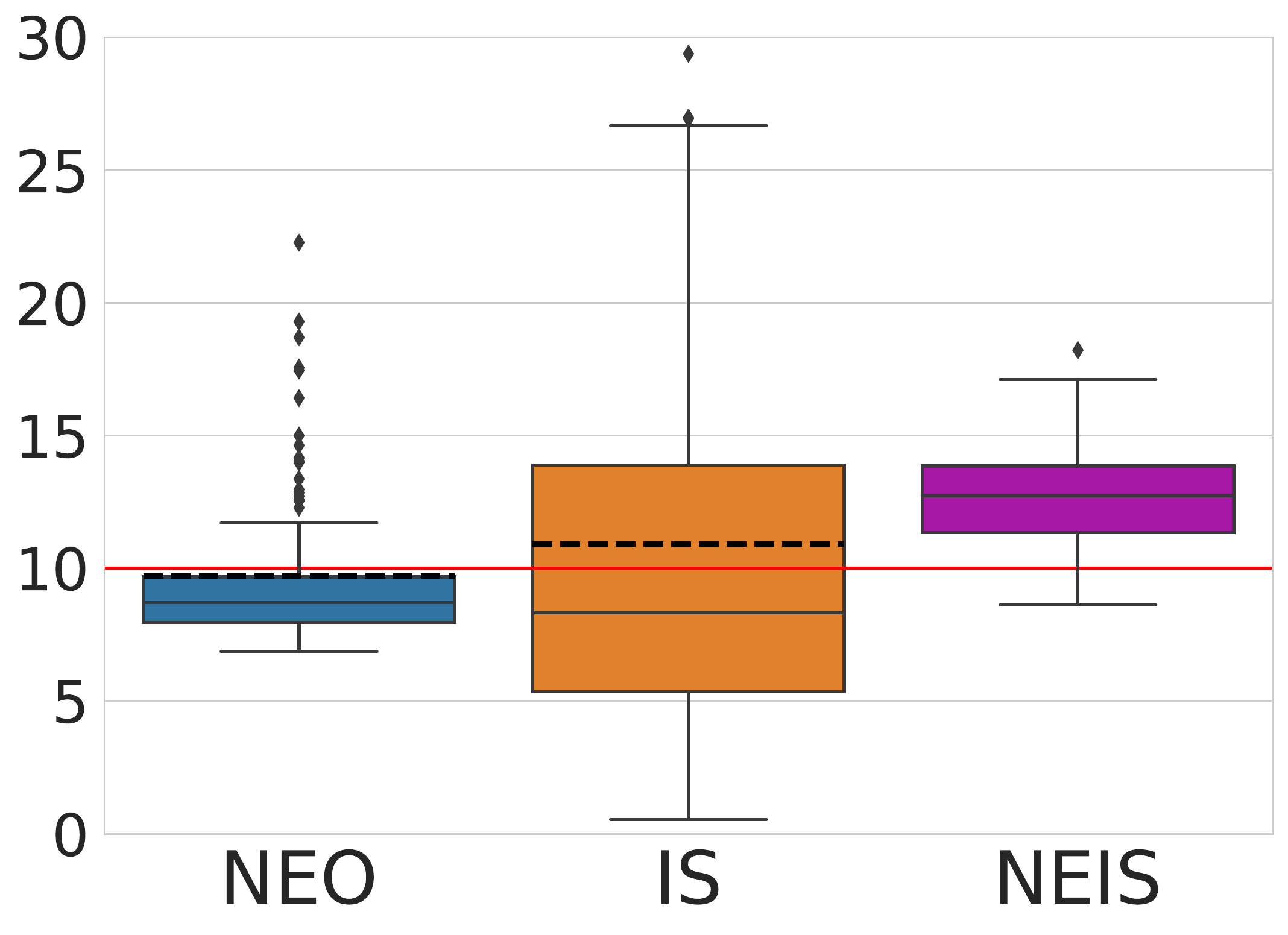}
    \includegraphics[width=0.4 \linewidth]{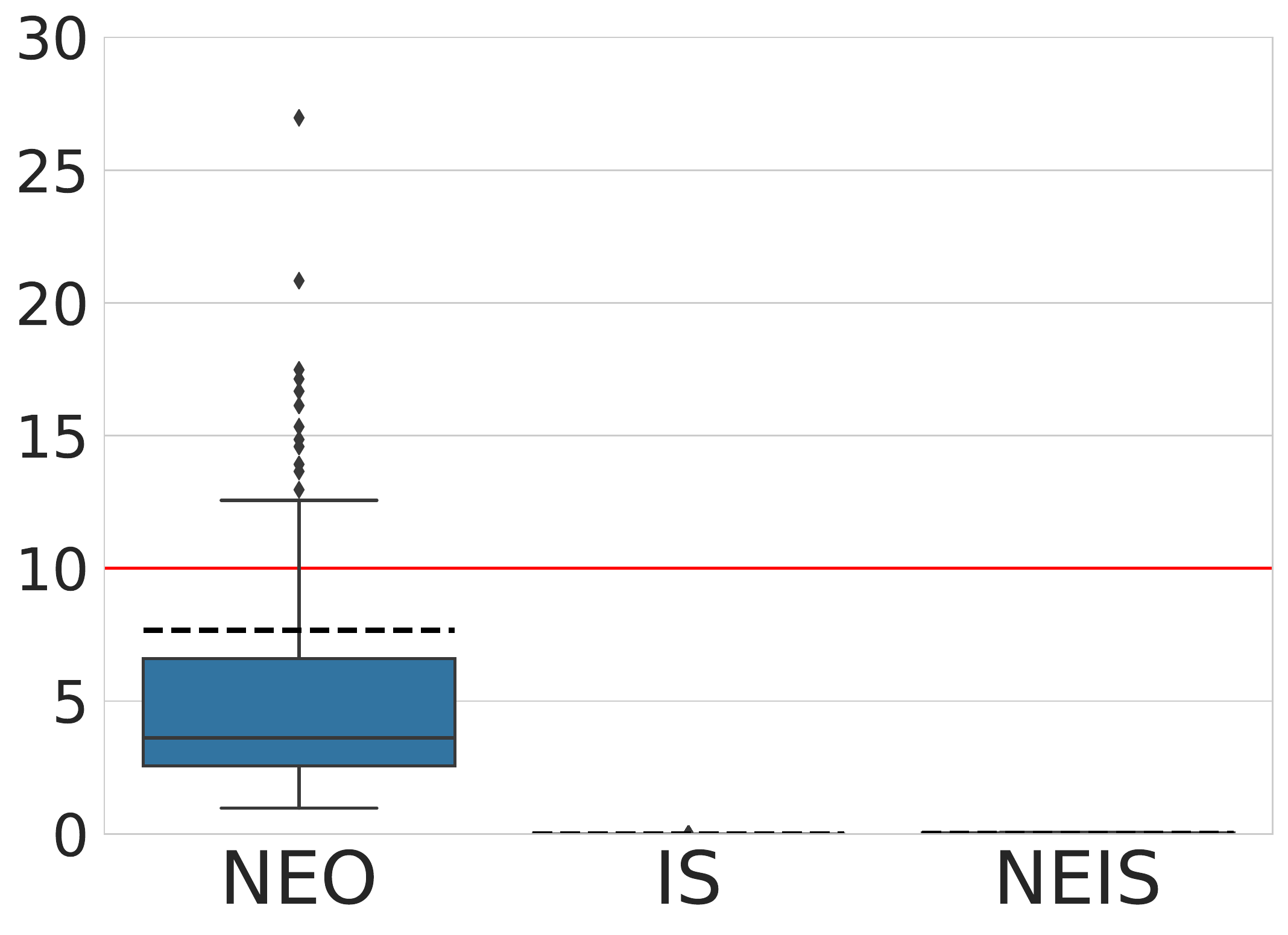}
    \caption{$\NEO$ v. $\NEIS$. 25 GM with $\sigma^2 = 0.005$, $ d = 5$. 
     500 runs each.}
    \label{fig:simple_gauss_neis}
\end{figure}

\subsection{Log-likelihood estimation}
\label{subsec:log-likelihood-constant}
We finally present the evaluation of the log-likelihood of the VAE introduced in \Cref{subsec:gibbs_inpaint}: given a test set $\mathcal{T} = \{\obs_i\}_{i=1}^{M_{\mathcal{T}}}$, we estimate $\sum_{i=1}^{M_{\mathcal{T}}}\log p_{\theta^*} (\obs_i)$. We also estimate similarly  the log-likelihood of an Importance Weighted Auto Encoder (IWAE) \cite{burda:grosse:2015}. 
Following \cite{wu:burda:grosse:2016}, we compare IS, AIS, and \NEO-IS. As previously, AIS, IS, and \NEO-IS are given a similar computational budget, choosing here $K=12$, $N= 5\cdot 10^3$. For \NEO, we choose $\gamma = 1.$ and $h=0.2$. Similarly, the stepsize of HMC transitions in AIS is $h=0.1$ in order to achieve an acceptance ratio of around $0.6$ in the HMC transitions. We report in \Cref{tab:ll_vae} the log-likelihood computed on the test set for VAE, IWAE with latent dimension in $\{16, 32\}$. For the same computational budget, \IFIS-IS yields consistently better values for the estimation of the log-likelihood of the VAE.
\begin{table}[]
\centering
\begin{tabular}{ |c|c|c|c|c| }
\hline
Model & VAE, $d = 32$ & VAE, $d = 16$ & IWAE, $d = 32$ & IWAE, $d = 16$ \\
\hline
IS & -90.17 & -90.44 &  -88.76 &-90.13\\
\hline
AIS  & -89.67 & -89.97 &  -88.30 &-89.61 \\
\hline
\InFiNE-IS & -88.81  & -89.17 &-87.46 & -88.99 \\
\hline
\end{tabular}
\caption{Evaluation of the log-likelihood (normalizing constant) of different Variational Auto Encoders.}
\label{tab:ll_vae}
\end{table}
\subsection{Gibbs inpainting}
We display here additional results for the Gibbs inpainting experiment presented in \Cref{subsec:gibbs_inpaint}. We emphasize that the starting images are chosen at random in the test set.

 \begin{figure}
     \centering
\includegraphics[width = .99\linewidth]{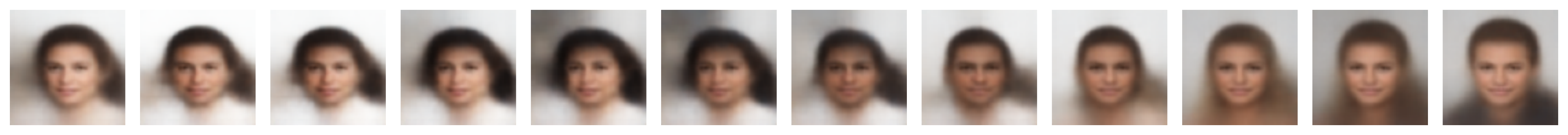}
\includegraphics[width = .99\linewidth]{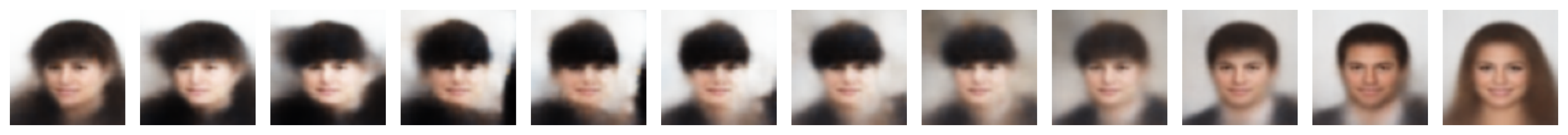}
\includegraphics[width = .99\linewidth]{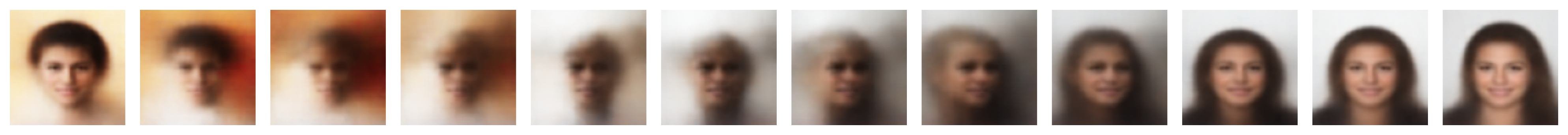}
\includegraphics[width = .99\linewidth]{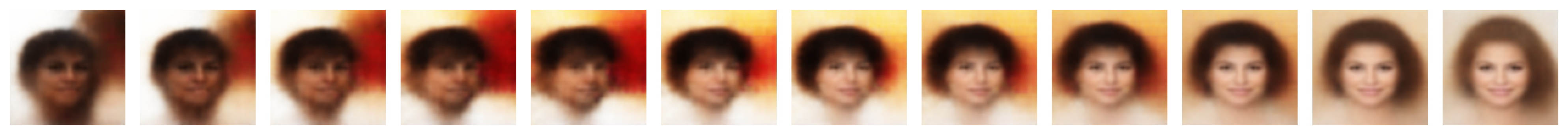}
     \caption{Forward orbits of \NEO-MCMC.}
     \label{fig:gibbs_inpainting}
 \end{figure}
 \begin{figure}
     \centering
     \includegraphics[width = .99\linewidth]{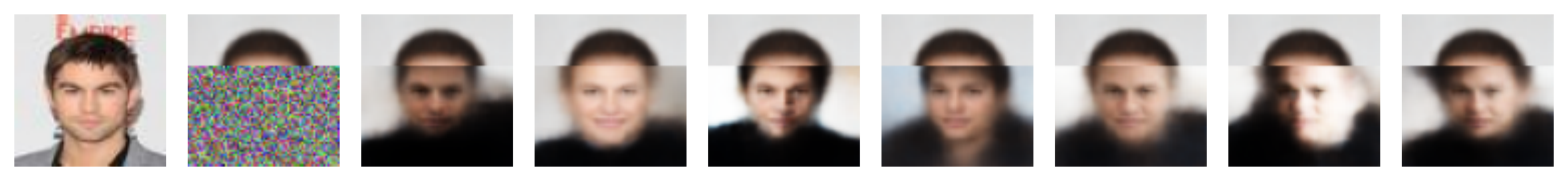}
\includegraphics[width = .99\linewidth]{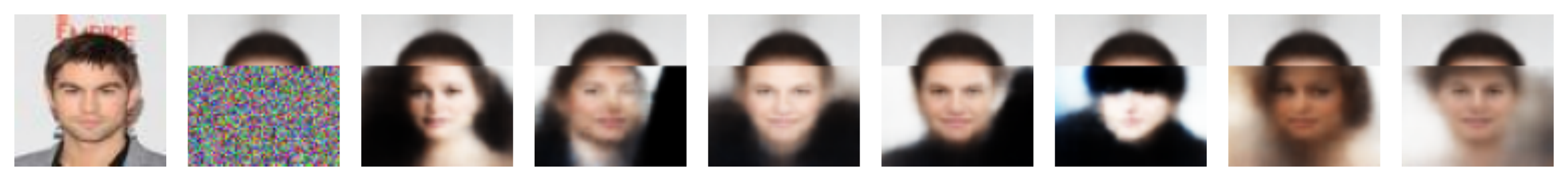}
\includegraphics[width = .99\linewidth]{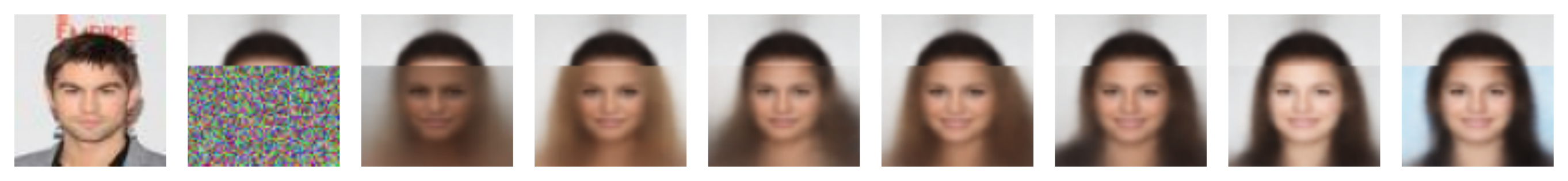}
\vspace{10pt}
\\
     \includegraphics[width = .99\linewidth]{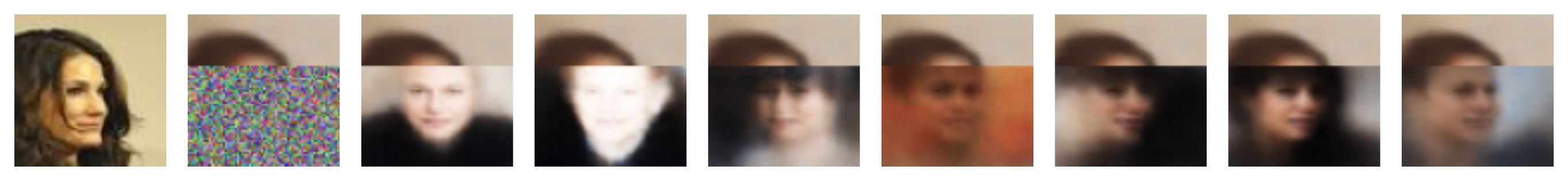}
\includegraphics[width = .99\linewidth]{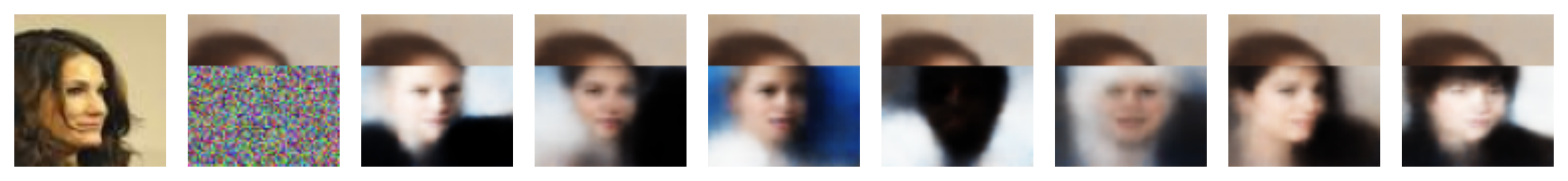}
\includegraphics[width = .99\linewidth]{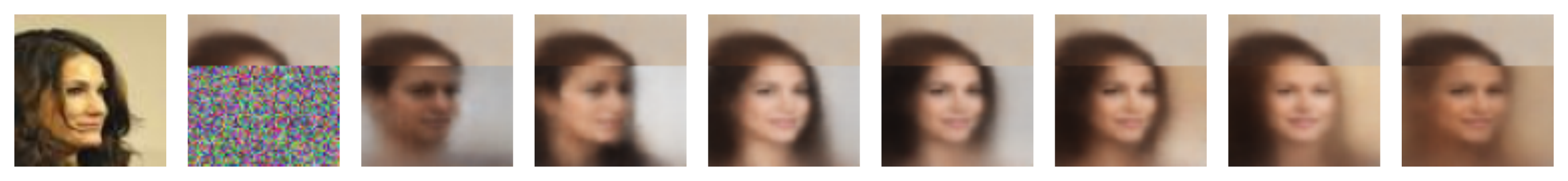}
\vspace{10pt}
\\
     \includegraphics[width = .99\linewidth]{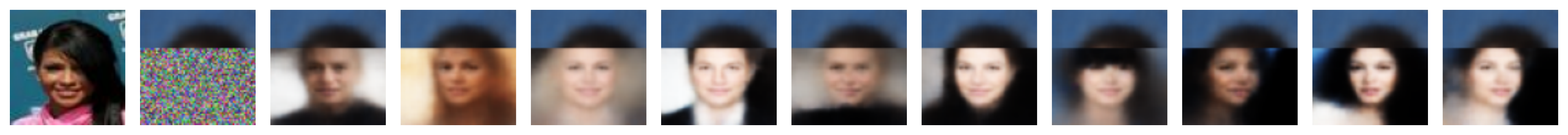}
\includegraphics[width = .99\linewidth]{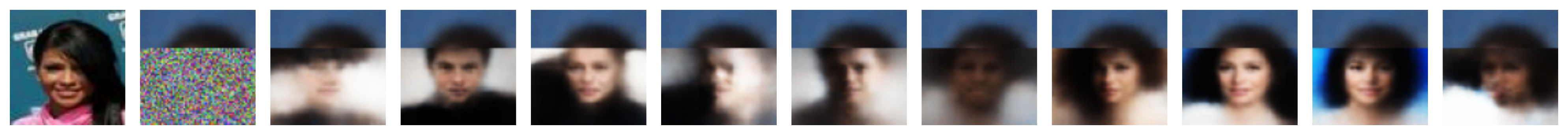}
\includegraphics[width = .99\linewidth]{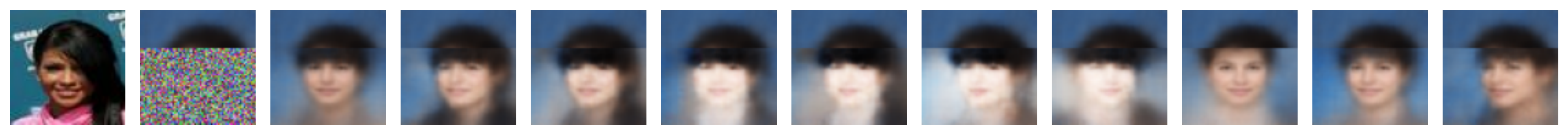}
     \caption{Gibbs inpainting for CelebA dataset. 
    From top to bottom: i-SIR, HMC and \NEO-MCMC: From left to right, original image, blurred image to reconstruct, and output every 5 iterations of the Markov chain.}
     \label{fig:gibbs_inpainting}
 \end{figure}
\section{\IFIS-VAE}
\label{sec:NEO-VAE}

\def\Uset{\mathsf{U}}
Denote by $p_\theta(x,\obs)$ the joint distribution of the observation $\obs \in \rset^p$  and the latent variable $x \in \rset^d$. The marginal likelihood is given, for $\obs \in \rset^p$ by $p_\theta(\obs)= \int p_\theta(x,\obs) \rmd x$. Given a training set $\mathcal{D}= \{ \obs_i \}_{i=1}^M$, the objective is to estimate $\theta$ by maximizing the likelihood, \ie\ maximizing $\log p_\theta(\mathcal{D}) = \sum_{i=1}^M \log p_\theta(\obs_i)$. 
Variational inference (VI)
provides us with a tool to simultaneously approximate the intractable
posterior $p_\theta(x|\obs)$ and maximize the marginal likelihood
$p_\theta(\mathcal{D})$ in the parameter $\theta$. This is achieved by introducing a
 parametric family  $\{ q_\phi(x|\obs), \phi \in \Phi\}$ to approximate the posterior $p_\theta(x|\obs)$ and maximizing
the Evidence Lower Bound (ELBO) (see \cite{kingma2019introduction}) $\mathcal{L}_{\operatorname{ELBO}}(\mathcal{D},\theta,\phi)= \sum_{i=1}^M \mathcal{L}_{\operatorname{ELBO}}(\obs_i,\theta,\phi)$ where
\begin{align}
\label{eq:elbo}
\mathcal{L}_{\textup{ELBO}}(\obs,\theta,\phi)&= \int \log\left(\frac{p_\theta(x,\obs)}{q_\phi(x\mid \obs)}\right) q_\phi(x\mid \obs)\rmd x \eqsp\\
&=\log p_\theta(\obs)-\operatorname{KL}(q_\phi(\cdot\mid \obs) \| p_\theta(\cdot\mid \obs) )\eqsp,\nonumber
\end{align}
and $\operatorname{KL}$ is the Kullback–Leibler divergence. 
In the sequel, we set  $\proposal(x)=q_\phi(x\mid \obs)$ and $\likelihood(x)= p_\theta(x,\obs)/ q_\phi(x \mid \obs)$. In such a case,  $\target(x)= \proposal(x) \likelihood(x)/ \const= p_{\theta}(x\mid \obs)$ and  $\const= p_\theta(\obs)$ (in these notations, the dependence in the observation $\obs$ is implicit).

We follow the the auxiliary variational inference framework (AVI) provided by \cite{agakov2004auxiliary}.
We consider a joint distribution $\bar{p}_\theta(x,u,\obs)$ which is such that $p_\theta(\obs)= \int p_\theta(x,u,\obs) \rmd x \rmd u$  where $u \in \mathsf{U}$ is an auxiliary variable (the auxiliary variable can both have discrete and continuous components; when $u$ has discrete components the integrals should be replaced by a sum). Then as the usual VI approach, we consider a parametric family $\{\bar{q}_\phi(x,u|\obs), \phi \in \Phi\}$.
Introducing auxiliary variables loses the tractability of \eqref{eq:elbo} but they allow for their own ELBO as suggested in \cite{agakov2004auxiliary,lawson2019energy} by minimizing 
\begin{equation}\label{eq:AVI_ELBO}
\operatorname{KL}(\bar{q}_\phi(\cdot\mid \obs) \| \bar{p}_\theta(\cdot\mid \obs) ) = \int \bar{q}_\phi(x,u|\obs) \log \left( \frac{\bar{p}_\theta(x,u,\obs)}{\bar{q}_\phi(x,u|\obs)} \right)  \rmd x \rmd u \eqsp.
\end{equation}
 The auxiliary variable $u$ is naturally associated with the extended target $\bar{\measpi}$ defined similar to \Cref{rem:gibbs-interpretation}, 
 \begin{equation}
\label{eq:extended-target}
\bar{\measpi}_N([x,\chunkum{x}{1}{N}{i}],i)= \frac{\estConstC{x}}{N \const} \proposal_N(\chunku{x}{1}{N})
 \end{equation}
 with $(x,u)=([x,\chunkum{x}{1}{N}{i}],i)$,  a shorthand notation for a $N$-tuple $\chunku{x}{1}{N}$ with $x^i= x$,  and, with $r_i$ defined in \eqref{eq:defintion-r-i}, 
\begin{equation}
\label{eq:definition-proposal}
\proposal_N(\chunku{x}{1}{N})= \proposal(x^1) r_1(x^1,\chunku{x}{2}{N}) = \proposal(x^j) r_j(x^j, \chunkum{x}{1}{N}{j}) \eqsp, \quad j \in \{1,\dots,N\}  \eqsp. 
\end{equation}
An extended proposal playing the role of $ \bar{q}_\phi(x,u|\obs)$ is derived from the \IFIS~MCMC sampler, i.e.
\begin{equation}
\label{eq:proposal-extended}
\bar{\proposal}_N([x,\chunkum{x}{1}{N}{i}],i)=  \frac{\estConstC{x}}{N \estConstC{\chunku{x}{1}{N}}} \proposal_N(\chunku{x}{1}{N})  \eqsp.
\end{equation}
where $\estConstC{\chunku{x}{1}{N}}$ is the \IFIS\ estimator \eqref{eq:def_estimator_normal_const} of the normalizing constant.
Note that, by construction,
\begin{equation}
\label{eq:expression-marginal}
\sum_{i=1}^N  \bar{\proposal}_N(\chunku{x}{1}{N},i) = \proposal_N(\chunku{x}{1}{N})
\end{equation}
showing that this joint proposal can be sampled by drawing the proposals $\chunku{x}{1}{N} \sim \rho_N$, then sampling the path index $i \in [N]$ with probability proportional to $(\estConstC{x^i})_{i=1}^N$ (with $\estConstC{x}$ defined in \eqref{eq:def_estimator_normal_const}).
The ratio of \eqref{eq:extended-target} over \eqref{eq:proposal-extended} is
\begin{equation}
\label{eq:ratio-extended}
{\bar{\measpi}_N(\chunku{x}{1}{N},i)}\big/{\bar{\proposal}_N(\chunku{x}{1}{N},i)}= {\estConstC{\chunku{x}{1}{N}}}\big/{\const} \eqsp.
\end{equation}
Thus, we write the augmented ELBO   \eqref{eq:AVI_ELBO}
\begin{align}
 \label{eq:infine_elbo-alt}
\elboneq &= \int_{} \proposal_N( \chunku{x}{1}{N})   \log \estConstC{\chunku{x}{1}{N}} \rmd \chunku{x}{1}{N} = \log \const - \operatorname{KL}( \bar{\proposal}_N | \bar{\measpi}_N )\eqsp,
\end{align}
where we have used \eqref{eq:expression-marginal} and that the ratio ${\bar{\measpi}_N(\chunku{x}{1}{N},i)}\big/{\bar{\proposal}_N(\chunku{x}{1}{N},i)}$ does not depend on
the path index $i$. When $\varpi_k= \delta_{k,0}$, where $\delta_{i,j}$ is the Kronecker symbol, and $\proposal_N(\chunku{x}{1}{N})= \prod_{j=1}^N \proposal(x^j)$, we exactly retrieve the Importance Weighted AutoEncoder (IWAE); see e.g.  \cite{burda:grosse:2015} and in particular the interpretation in \cite{cremer2017reinterpreting}.

Choosing the conformal Hamiltonian introduced in \Cref{subsec:NISestimators} allows for a family of invertible flows that depends on the parameter $\theta$ which itself is directly linked to the target distribution.
\begin{table*}[!t]
\centering
\caption{Negative Log Likelihood estimates for VAE models for different latent space dimensions.}
\label{tab:vae_results2}
\begin{tabular}{c|c|c||c|c||c|c||c|c|}
\cline{2-9}
 & \multicolumn{2}{c||}{$d = 4$} & \multicolumn{2}{c||}{$d = 8$} & \multicolumn{2}{c||}{$d = 16$} & \multicolumn{2}{c|}{$d = 50$} \\ \hline
\multicolumn{1}{|c|}{model} & IS & \InFiNE  & IS & \InFiNE& IS  & \InFiNE & IS & \InFiNE \\ \hline
\multicolumn{1}{|c|}{VAE} & $115.01$&$113.49$&$97.96$&$97.64$&$90.52$&$90.42$&$88.22$&$88.36$\\
\multicolumn{1}{|c|}{IWAE, $N=5$} & $113.33$&$111.83$&$97.19$&$96.61$&$89.34$&$89.05$&$87.49$&$87.27$ \\ 
\multicolumn{1}{|c|}{IWAE, $N=30$} & $111.92$&$110.36$&$96.81$&$95.94$&$88.99$&$88.64$&$86.97$&$86.93$ \\ \hline
\multicolumn{1}{|c|}{\InFiNE\ VAE, $K=3$} & $109.14$&$107.47$&$94.50$&$94.26$&$89.03$&$88.92$&$88.14$&$88.16$ \\
\multicolumn{1}{|c|}{\InFiNE\ VAE, $K=10$} & $110.02$&$107.90$&$94.63$&$94.22$&$89.71$&$88.68$&$88.25$&$86.95$ \\ \hline
\end{tabular}
\end{table*}
Table~\ref{tab:vae_results2} displays the estimated NLL of all models provided by IS and the \InFiNE\ method. It is interesting to note here again that \InFiNE\ improves the training of the VAE when the dimension of the latent space is small to moderate.

\end{document}